\documentclass[12pt,draftclsnofoot,onecolumn]{IEEEtran}

\usepackage{cite}

\ifCLASSINFOpdf
  \usepackage[pdftex]{graphicx}
  \usepackage{epstopdf}
\else
  \usepackage[dvips]{graphicx}
\fi

\usepackage{etoolbox}
\usepackage{float}
\usepackage{array}
\usepackage{amsmath}

\usepackage{subfig}
\usepackage{bm}
\usepackage{amssymb}
\usepackage{multirow}
\usepackage{multicol}
\usepackage{xtab}
\usepackage{tabularx}
\usepackage{color}
\usepackage{dsfont}
\usepackage{optidef}

\usepackage{amsthm}
\newtheorem{theorem}{Theorem}
\newtheorem{proposition}{Proposition}
\newtheorem{lemma}{Lemma}

\newtheorem{definition}{Definition}

\usepackage{algpseudocode}
\usepackage{algorithm}
\algnewcommand{\algorithmicgoto}{\textbf{go to}}%
\algnewcommand{\Goto}[1]{\algorithmicgoto~\ref{#1}}%
\algnewcommand\algorithmicinput{\textbf{Input:}}
\algnewcommand\Input{\item[\algorithmicinput]}
\algnewcommand\algorithmicoutput{\textbf{Output:}}
\algnewcommand\Output{\item[\algorithmicoutput]}

\usepackage[nolist]{acronym}
\makeatletter
\newif\if@in@acrolist
\AtBeginEnvironment{acronym}{\@in@acrolisttrue}
\newrobustcmd{\LU}[2]{\if@in@acrolist#1\else#2\fi}
\newcommand{\ACF}[1]{{\@in@acrolisttrue\acf{#1}}}
\makeatother

\usepackage[allbordercolors={0 1 0},bookmarksopen=true,bookmarksnumbered=true]{hyperref}

\begin{document}

\title{Multi-Hop Wireless Optical Backhauling\\for LiFi Attocell Networks:\\Bandwidth Scheduling and Power Control}
	
\author{Hossein~Kazemi,~\IEEEmembership{Student~Member,~IEEE,}
		Majid~Safari,~\IEEEmembership{Member,~IEEE,}
		and~Harald~Haas,~\IEEEmembership{Fellow,~IEEE}
		\thanks{This work was presented in part at the IEEE Global Communications Conference (GLOBECOM), Dec. 2018.\vspace{-20pt}}
		}
	
\maketitle

\vspace{-1.2cm}

\begin{abstract}
The backhaul of hundreds of light fidelity (LiFi) base stations (BSs) constitutes a major challenge. Indoor wireless optical backhauling is a novel approach whereby the interconnections between adjacent LiFi BSs are provided by way of directed line-of-sight (LOS) wireless infrared (IR) links. Building on the aforesaid approach, this paper presents the top-down design of a multi-hop wireless backhaul configuration for multi-tier optical attocell networks by proposing the novel idea of super cells. Such cells incorporate multiple clusters of attocells that are connected to the core network via a single gateway based on multi-hop decode-and-forward (DF) relaying. Consequently, new challenges arise for managing the bandwidth and power resources of the bottleneck backhaul. By putting forward user-based bandwidth scheduling (UBS) and cell-based bandwidth scheduling (CBS) policies, the system-level modeling and analysis of the end-to-end multi-user sum rate is elaborated. In addition, optimal bandwidth scheduling under both UBS and CBS policies are formulated as constrained convex optimization problems, which are solved by using the projected subgradient method. Furthermore, the transmission power of the backhaul system is opportunistically reduced by way of an innovative fixed power control (FPC) strategy. The notion of backhaul bottleneck occurrence (BBO) is introduced. An accurate approximate expression of the probability of BBO is derived, and then verified using Monte Carlo simulations. Several insights are provided into the offered gains of the proposed schemes through extensive computer simulations, by studying different aspects of the performance of super cells including the average sum rate, the BBO probability and the backhaul power efficiency (PE).
\end{abstract}


\begin{acronym}[ACO{-}OFDM]
	\acro{AP}{\LU{A}{a}ccess \LU{P}{p}oint}
	\acro{AF}{\LU{A}{a}mplify-and-\LU{F}{f}orward}
	\acro{ACO{-}OFDM}{\LU{A}{a}symmetrically \LU{C}{c}lipped \LU{O}{o}ptical OFDM}
	\acro{APC}{\LU{A}{a}daptive \LU{P}{p}ower \LU{C}{c}ontrol}
	\acro{ARPC}{\LU{A}{a}verage \LU{R}{r}ate \LU{P}{p}ower \LU{C}{c}ontrol}
	\acro{ASE}{\LU{A}{a}rea \LU{S}{s}pectral \LU{E}{e}fficiency}
	\acro{ASPC}{\LU{A}{a}verage SINR \LU{P}{p}ower \LU{C}{c}ontrol}
	\acro{AWGN}{\LU{A}{a}dditive \LU{W}{w}hite Gaussian \LU{N}{n}oise}
	\acro{BBO}{\LU{B}{b}ackhaul \LU{B}{b}ottleneck \LU{O}{o}ccurrence}
	\acro{BER}{\LU{B}{b}it \LU{E}{e}rror \LU{R}{r}atio}
	\acro{BPSK}{\LU{B}{b}inary \LU{P}{p}hase \LU{S}{s}hift \LU{K}{k}eying}
	\acro{BS}{\LU{B}{b}ase \LU{S}{s}tation}
	\acro{CBS}{\LU{C}{c}ell-based \LU{B}{b}andwidth \LU{S}{s}cheduling}
	\acro{CCDF}{\LU{C}{c}omplementary \LU{C}{c}umulative \LU{D}{d}istribution \LU{F}{f}unction}
	\acro{CCI}{\LU{C}{c}o-\LU{C}{c}hannel \LU{I}{i}nterference}
	\acro{CDF}{\LU{C}{c}umulative \LU{D}{d}istribution \LU{F}{f}unction}
	\acro{CFFR}{\LU{C}{c}ooperative FFR}
	\acro{CLT}{\LU{C}{c}entral \LU{L}{l}imit \LU{T}{t}heorem}
	\acro{CoMP}{\LU{C}{c}oordinated \LU{M}{m}ulti-\LU{P}{p}oint}
	\acro{CP}{\LU{C}{c}yclic \LU{P}{p}refix}
	\acro{DAS}{\LU{D}{d}istributed \LU{A}{a}ntenna \LU{S}{s}ystem}
	\acro{DSL}{\LU{D}{d}igital \LU{S}{s}ubscriber \LU{L}{l}ine}
	\acro{DC}{\LU{D}{d}irect \LU{C}{c}urrent}
	\acro{DCO{-}OFDM}{\LU{D}{d}irect \LU{C}{c}urrent-biased \LU{O}{o}ptical \LU{O}{o}rthogonal \LU{F}{f}requency \LU{D}{d}ivision \LU{M}{m}ultiplexing}
	\acro{DF}{\LU{D}{d}ecode-and-\LU{F}{f}orward}
	\acro{EMI}{\LU{E}{e}lectromagnetic \LU{I}{i}nterference}
	\acro{eU{-}OFDM}{\LU{E}{e}nhanced \LU{U}{u}nipolar \LU{O}{o}ptical OFDM}
	\acro{FDE}{\LU{F}{f}requency \LU{D}{d}omain \LU{E}{e}qualization}
	\acro{FFR}{\LU{F}{f}ractional \LU{F}{f}requency \LU{R}{r}euse}
	\acro{FFT}{\LU{F}{f}ast Fourier \LU{T}{t}ransform}
	\acro{FOV}{\LU{F}{f}ield \LU{O}{o}f \LU{V}{v}iew}
	\acro{FPC}{\LU{F}{f}ixed \LU{P}{p}ower \LU{C}{c}ontrol}
	\acro{FR}{\LU{F}{f}ull \LU{R}{r}euse}
	\acro{FRF}{\LU{F}{f}requency \LU{R}{r}euse \LU{F}{f}actor}
	\acro{FR{-}VL}{\LU{F}{f}ull \LU{R}{r}euse \LU{V}{v}isible \LU{L}{l}ight}
	\acro{FSO}{\LU{F}{f}ree \LU{S}{s}pace \LU{O}{o}ptical}
	\acro{FTTB}{\LU{F}{f}iber-\LU{T}{t}o-\LU{T}{t}he-\LU{B}{b}uilding}
	\acro{FTTH}{\LU{F}{f}iber-\LU{T}{t}o-\LU{T}{t}he-\LU{H}{h}ome}
	\acro{FTTP}{\LU{F}{f}iber-\LU{T}{t}o-\LU{T}{t}he-\LU{P}{p}remises}
	\acro{IB{-}VL}{\LU{I}{i}n-\LU{B}{b}and \LU{V}{v}isible \LU{L}{l}ight}
	\acro{ICI}{\LU{I}{i}nter-\LU{C}{c}ell \LU{I}{i}nterference}
	\acro{IM{-}DD}{\LU{I}{i}ntensity \LU{M}{m}odulation and \LU{D}{d}irect \LU{D}{d}etection}
	\acro{i.i.d.}{\LU{I}{i}ndependent and \LU{I}{i}dentically \LU{D}{d}istributed}
	\acro{IFFT}{\LU{I}{i}nverse \LU{F}{f}ast Fourier \LU{T}{t}ransform}
	\acro{IR}{\LU{I}{i}nfrared}
	\acro{ISI}{\LU{I}{i}nter-\LU{S}{s}ymbol \LU{I}{i}nterference}
	\acro{JTDF}{\LU{J}{j}oint \LU{T}{t}ransmission with \LU{D}{d}ecode-and-\LU{F}{f}orward}
	\acro{LAN}{\LU{L}{l}ocal \LU{A}{a}rea \LU{N}{n}etwork}
	\acro{LED}{\LU{L}{l}ight \LU{E}{e}mitting \LU{D}{d}iode}
	\acro{LiFi}{\LU{L}{l}ight \LU{F}{f}idelity}
	\acro{LOS}{\LU{L}{l}ine-\LU{O}{o}f-\LU{S}{s}ight}
	\acro{LTE}{\LU{L}{l}ong-\LU{T}{t}erm \LU{E}{e}volution}
	\acro{MAC}{\LU{M}{m}edium \LU{A}{a}ccess \LU{C}{c}ontrol}
	\acro{MC}{\LU{M}{u}ulti-\LU{C}{c}arrier}
	\acro{MIMO}{\LU{M}{m}ultiple \LU{I}{i}nput \LU{M}{m}ultiple \LU{O}{o}utput}
	\acro{MSE}{\LU{M}{m}ean \LU{S}{s}quare \LU{E}{e}rror}
	\acro{MSPC}{\LU{M}{m}aximum SINR \LU{P}{p}ower \LU{C}{c}ontrol}
	\acro{MMSE}{\LU{M}{m}inimum \LU{M}{m}ean \LU{S}{s}quare \LU{E}{e}rror}
	\acro{mmWave}{\LU{M}{m}illimeter \LU{W}{w}ave}
	\acro{NPC}{\LU{N}{n}o \LU{P}{p}ower \LU{C}{c}ontrol}
	\acro{LAN}{\LU{L}{l}ocal \LU{A}{a}rea \LU{N}{n}etwork}
	\acro{NLOS}{\LU{N}{n}on-\LU{L}{l}ine-\LU{O}{o}f-\LU{S}{s}ight}
	\acro{NODF}{\LU{N}{n}on-\LU{O}{o}rthogonal \LU{D}{d}ecode-and-\LU{F}{f}orward}
	\acro{OFDM}{\LU{O}{o}rthogonal \LU{F}{f}requency \LU{D}{d}ivision \LU{M}{m}ultiplexing}
	\acro{OFDMA}{\LU{O}{o}rthogonal \LU{F}{f}requency \LU{D}{d}ivision \LU{M}{m}ultiple \LU{A}{a}ccess}
	\acro{OOK}{\LU{O}{o}n-\LU{O}{o}ff \LU{K}{k}eying}
	\acro{PAM}{\LU{P}{p}ulse \LU{A}{a}mplitude \LU{M}{m}odulation}
	\acro{PAPR}{\LU{P}{p}eak-to-\LU{A}{a}verage \LU{P}{p}ower \LU{R}{r}atio}
	\acro{PD}{\LU{P}{p}hotodiode}
	\acro{PDF}{\LU{P}{p}robability \LU{D}{d}ensity \LU{F}{f}unction}
	\acro{PE}{\LU{P}{p}ower \LU{E}{e}fficiency}
	\acro{PHY}{\LU{P}{p}hysical \LU{L}{l}ayer}
	\acro{PLC}{\LU{P}{p}ower \LU{L}{l}ine \LU{C}{c}ommunication}
	\acro{PMF}{\LU{P}{p}robability \LU{M}{m}ass \LU{F}{f}unction}
	\acro{PoE}{\LU{P}{p}ower-over-Ethernet}
	\acro{P{-}OFDM}{\LU{P}{p}olar OFDM}
	\acro{PON}{\LU{P}{p}assive \LU{O}{o}ptical \LU{N}{n}etwork}
	\acro{PPP}{\LU{P}{p}oisson \LU{P}{p}oint \LU{P}{p}rocess}
	\acro{PSD}{\LU{P}{p}ower \LU{S}{s}pectral \LU{D}{d}ensity}
	\acro{PTP}{\LU{P}{p}oint-\LU{T}{t}o-\LU{P}{p}oint}
	\acro{PTMP}{\LU{P}{p}oint-\LU{T}{t}o-\LU{M}{m}ulti-\LU{P}{p}oint}
	\acro{QAM}{\LU{Q}{q}uadrature \LU{A}{a}mplitude \LU{M}{m}odulation}
	\acro{QoS}{\LU{Q}{q}uality of \LU{S}{s}ervice}
	\acro{QPSK}{\LU{Q}{q}uadrature \LU{P}{p}hase \LU{S}{s}hift \LU{K}{k}eying}
	\acro{RGB}{\LU{R}{r}ed-\LU{G}{g}reen-\LU{B}{b}lue}
	\acro{RF}{\LU{R}{r}adio \LU{F}{f}requency}
	\acro{RHS}{\LU{R}{r}ight \LU{H}{h}and \LU{S}{s}ide}
	\acro{RMS}{\LU{R}{r}oot \LU{M}{m}ean \LU{S}{s}quare}
	\acro{RoF}{\LU{R}{r}adio-over-\LU{F}{f}iber}
	\acro{RTP}{\LU{R}{r}elative \LU{T}{t}otal \LU{P}{p}ower}
	\acro{SC}{\LU{S}{s}ingle \LU{C}{c}arrier}
	\acro{SE}{\LU{S}{s}pectral \LU{E}{e}fficiency}
	\acro{SEE{-}OFDM}{\LU{S}{s}pectral and \LU{E}{e}nergy \LU{E}{e}fficient OFDM}
	\acro{SINR}{\LU{S}{s}ignal-to-\LU{N}{n}oise-plus-\LU{I}{i}nterference \LU{R}{r}atio}
	\acro{SMF}{\LU{S}{s}ingle \LU{M}{m}ode \LU{F}{f}iber}
	\acro{SNR}{\LU{S}{s}ignal-to-\LU{N}{n}oise \LU{R}{r}atio}
	\acro{UB}{\LU{U}{u}nlimited \LU{B}{b}ackhaul}
	\acro{UBS}{\LU{U}{u}ser-based \LU{B}{b}andwidth \LU{S}{s}cheduling}
	\acro{UE}{\LU{U}{u}ser \LU{E}{e}quipment}
	\acro{VPPM}{\LU{V}{v}ariable \LU{P}{p}ulse \LU{P}{p}osition \LU{M}{m}odulation}
	\acro{VL}{\LU{V}{v}isible \LU{L}{l}ight}
	\acro{VLC}{\LU{V}{v}isible \LU{L}{l}ight \LU{C}{c}ommunication}
	\acro{WiFi}{\LU{W}{w}ireless \LU{F}{f}idelity}
	\acro{WLAN}{\LU{W}{w}ireless \LU{L}{l}ocal \LU{A}{a}rea \LU{N}{n}etwork}
	\acro{WOC}{\LU{W}{w}ireless \LU{O}{o}ptical \LU{C}{c}ommunication}
\end{acronym}

\section{Introduction}\label{section1}
The advent of \acp{LED} has radically changed the modern lighting industry due to their distinguished features including high energy efficiency, long operational lifetime, a compact form factor, easy maintenance and low cost. It is expected that \ac{LED} lighting will reach a market share of $84\%$ by the year $2030$ \cite{EnergySavingsForecast}. The application of \acp{LED} for indoor illumination has provided the possibility to deliver luminous efficacies of more than $100$ lm/W \cite{Cree}. Additionally, the intensity of their output light can be switched at high frequencies while the rate of variations is imperceptible to the human eye. In fact, the \ac{VL} spectrum offers a vast amount of unregulated bandwidth in $400\text{--}790$ THz. This unique opportunity is exploited for the deployment of value-added services based on \ac{VLC} to piggyback the wireless communication functionality onto the future lighting network in homes or offices \cite{Figueiredo}.

As the advanced version of \ac{VLC}, \ac{LiFi} transforms \ac{LED} luminaires into broadband wireless access points to support multi-user networking \cite{Haas2}. In the realm of heterogeneous networks, \ac{LiFi} can coexist synergistically with \ac{WiFi}. To this end, \ac{LiFi} realizes a high-bandwidth, uncongested and unregulated downlink path, while \ac{WiFi} constitutes a reliable uplink channel where congestion is less likely \cite{Ayyash}. From a network deployment perspective, the dense distribution of indoor luminaires lays the groundwork for establishing ultra-dense \ac{LiFi} networks, also known as optical attocell networks. Studies on the downlink performance show that through a judicious system configuration and by using rate-adaptive \ac{DCO{-}OFDM}, optical attocell networks generally outperform both \ac{RF} femtocell and indoor \ac{mmWave} networks in terms of the area spectral efficiency \cite{Stefan,CChen5}.

Backhaul is an essential part of the cellular network architecture, granting \acp{BS} access to the core network. Therefore, it is crucial to provide high data rate and reliable backhaul links for transporting the busy wireless traffic between \acp{BS} and the core network. Developing cost-effective backhauling solutions for massively deployed small cells is considered as one of the most important challenges in the rollout of the forthcoming $5$G cellular networks \cite{NWang}. To achieve multi-Gbits/s connectivity for indoor broadband wireless networks, a fiber-to-the-home/premises technology based on a \ac{PON} architecture is used \cite{Koonen}. For multi-dwelling buildings, signal distribution from the optical fiber hub to individual dwellings is also a major component of the access network. In-building backhauling can be done either wired or wirelessly. To this end, wired solutions based on Ethernet and \ac{PLC} have been considered \cite{WNi,Papaioannou}. In addition, it is possible to realize the distribution network within buildings wirelessly using \ac{mmWave} communications in the $60$ GHz band, which has been found suitable for indoor environments \cite{Dehos}. An efficient alternative to complement fiber-based \ac{PON}, namely G.fast, has been standardized \cite{Timmers}. G.fast is a high speed digital subscriber line standard which utilizes copper wires and promises Gbits/s connectivity for distances up to $250$ m.

When it comes to densely deployed optical attocell networks, because of the sophisticated structure of backhaul connections for multiple \ac{LiFi} \acp{BS}, designing an efficient backhaul network is more challenging. Prior studies have addressed the problem of backhauling for indoor \ac{VLC} systems by three main approaches: employing \ac{PLC} to reach light fixtures through the existing electricity wiring infrastructure in buildings, thus creating hybrid \ac{PLC}-\ac{VLC} systems \cite{Komine2,Komine4,Song,HMa2}; interfacing Ethernet technology with \ac{VLC} that allows the distribution of both data and electricity to \ac{LED} luminaires by a single Category $5$ cable based on the Power-over-Ethernet standard \cite{Mark,Delgado}; and extending single mode optical fiber cables to \ac{LED} lamps to enable multi-Gbits/s connectivity based on an integrated \ac{PON}-\ac{VLC} architecture \cite{YWang2,Chow,YWang1}.

As an alternative to the aforementioned approaches, backhauling for indoor \ac{LiFi} networks can be designed based on wireless optical communications. In particular, the idea of using \ac{VLC} to build inter-\ac{BS} links in optical attocell networks with a star topology was first put forward in \cite{Kazemi3}. The work in \cite{Kazemi5} carried out an extended design and optimization of the wireless optical backhaul system in both \ac{VL} and \ac{IR} bands by using a tree topology. In these works, the bandwidth of the shared backhaul was assumed to be equally apportioned among multiple downlink paths. The study in \cite{Kazemi4} proposed heuristic methods for bandwidth scheduling in a two tier \ac{LiFi} network, and introduced new criteria to control the total power of the backhaul system. However, the problem of optimal bandwidth scheduling remains unexplored. Furthermore, although preliminary results for power control and backhaul bottleneck performance were presented in \cite{Kazemi4}, an in-depth analysis of such new aspects is subject to an extended study.

This paper primarily attempts to address the above-mentioned shortcomings by putting forward the design and analysis of multi-hop wireless optical backhauling for multi-tier optical attocell networks through the introduction of the novel concept of super cells. Note this extension is not trivial due to the intricate configuration of a multi-tier multi-hop super cell. Furthermore, this work makes multiple contributions including:
\begin{itemize}
	\item Novel \ac{UBS} and \ac{CBS} policies are proposed for dividing the shared bandwidth of the backhaul system.
	\item By employing \ac{DCO{-}OFDM} combined with \ac{DF} relaying, the end-to-end multi-user sum rate is derived for the generalized case of multi-tier super cells for both \ac{UBS} and \ac{CBS} policies.
	\item For each policy, the optimal bandwidth allocation is formulated as an optimization problem and novel optimal bandwidth scheduling algorithms are developed.
	\item A \ac{FPC} mechanism is proposed to set a controlled operating point for the total backhaul power. Concerning the access system performance, three main schemes are devised: \ac{MSPC}, \ac{ASPC} and \ac{ARPC}. For each scheme, the corresponding power control coefficient is derived in closed form.
	\item The notion of \ac{BBO} is scrutinized by a thorough analysis and a tight approximation of the \ac{BBO} probability is derived analytically.
	\item Using illustrative numerical examples, new insights are provided into the performance of multi-tier super cells by studying the average sum rate, the \ac{BBO} probability and the backhaul \ac{PE}.
\end{itemize}

\section{Multi-Hop Wireless Backhaul System Design}\label{sec_4_1}
This section presents system-level principles and preliminaries required for the design and analysis of a multi-hop wireless optical backhaul network using a top-down approach.

\begin{figure}[!t]
	\centering
	\includegraphics[width=0.6\textwidth]{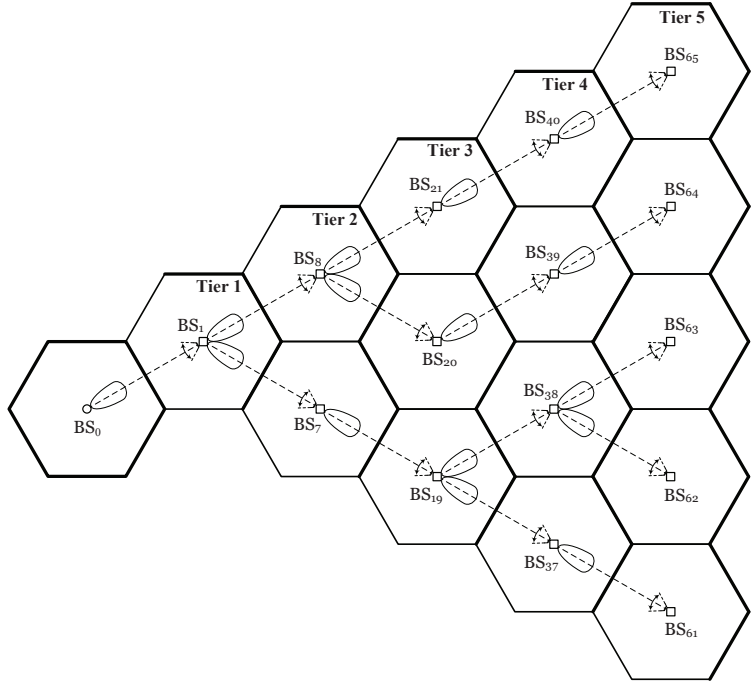}
	\caption{One branch of a five tier super cell with multi-hop wireless optical backhaul links.}
	\label{fig_4_2_1}
\end{figure}

\subsection{Network Configuration and Super Cells}\label{sec_4_1_1}
In this paper, an unbounded optical attocell network with a hexagonal tessellation is considered. Such a model is appropriate for network deployments in spacious office environments \cite{CChen5}. The network incorporates multi-tier bundles of hexagonal attocells which are referred to as \textit{super cells} in this work, with each bundle encompassing one, two or possibly several tiers. The entire network coverage is then \textit{tiled} by multiple super cells. Within every super cell, only the central \ac{BS} is directly connected to the gateway while the remaining \acp{BS} are connected using a tree topology that extends from a root at the central \ac{BS} toward the outer tiers. Let $N_{\rm{T}}$ denote the total number of tiers deployed. For clarity, one branch of a super cell with $N_{\rm{T}}=5$ is illustrated in Fig.~\ref{fig_4_2_1}. Note that the picture of the whole super cell is constituted by rotating and repeating the shown branch every $60^\circ$ counterclockwise. Nevertheless, this is just an illustration and the generality of presentation is maintained throughout the paper by adopting a parametric modeling methodology, i.e., for a general case of the $k$th branch for $k=1,2,\ldots,6$. A wireless optical communication technology operating in the \ac{IR} optical band is employed to establish inter-\ac{BS} backhaul links. The use of the \ac{IR} band allows to cancel unwanted backhaul-induced interference on the \ac{VL} access network \cite{Kazemi3}.

In conventional multi-hop wireless systems, a half duplex signaling protocol allows each relay to transmit only on its preallocated (time or frequency) resource slot to eliminate \ac{RF} interference within the network. Such interference avoidance comes at the expense of a remarkable loss in \ac{SE}. For the multi-hop wireless optical backhaul system under consideration, by using a sufficiently focused optical beam and a directed \ac{LOS} configuration, the crosstalk among backhaul links is effectively canceled \cite{Kazemi3,Kazemi5}. Hence, half duplex relaying on the path results in an unnecessary misutilization of resources and to avoid this, \acp{BS} are permitted to perform full duplex relaying.

The employment of \ac{DCO{-}OFDM} for data transmission in both access and backhaul systems allows an efficient management of network resources. To maintain the generality of presentation, the parameters related to the access (resp. backhaul) system are denoted using a subscript $\rm{a}$ (resp. $\rm{b}$). More specifically, an $N_{\rm{a}}$-point (resp. $N_{\rm{b}}$-point) \ac{IFFT}/\ac{FFT} is used for \ac{DCO{-}OFDM} transmission in the access (resp. backhaul) system. The remaining assumptions are similar to those used in \cite{Kazemi5}.

\subsection{Signal-to-Noise-plus-Interference Ratio}\label{sec_4_1_2}
\subsubsection{Downlink SINR Statistics}
A number of \ac{UE} devices are randomly scattered in the coverage of a super cell with a uniform distribution, attempting to obtain a downlink connection from optical \acp{BS}. The downlink channel follows a \ac{LOS} light propagation model\footnote{Except small regions in proximity to the network boundaries where the \ac{NLOS} effect is manifested most, in the rest of areas under coverage, more than $90\%$ of the received optical power comes solely from the \ac{LOS} component \cite{CChen5}.}. With the assumption of the whole bandwidth being fully reused across all attocells, the downlink quality in each attocell is influenced by \ac{CCI} from neighboring \acp{BS}. When the number of interfering \acp{BS} is large, the aggregate effect of the received \ac{CCI} signals is commonly treated as a white Gaussian noise. The received signal is also perturbed by an additive noise comprising signal-independent shot noise and thermal noise, which is modeled by a zero mean Gaussian distribution with a single-sided \ac{PSD} of $N_0$.

According to a polar coordinate system with $\text{BS}_0$ at the origin, the electrical \ac{SINR} per subcarrier for the $u$th \ac{UE} associated with $\text{BS}_i$ at $z_u=(r_u,\theta_u)$ is given by \cite{CChen5}:
\begin{equation}\label{eq_4_1_1}
\gamma_u = \frac{\xi_{\rm{a}}^{-1}(r_i^2(z_u)+h^2)^{-m-3}}{\sum\limits_{j\in\mathcal{J}_i}(r_j^2(z_u)+h^2)^{-m-3}+\Omega},
\end{equation}
where $\xi_{\rm{a}}=\frac{N_{\rm{a}}-2}{N_{\rm{a}}}$ is the subcarrier utilization factor; ${r_i}(z_u)=\sqrt{r_u^2+R_i^2-2{R_i}r_u\cos(\theta_u-\Theta_i)}$ indicates the horizontal distance of $z_u$ from $\text{BS}_i$; $h$ is the vertical separation between the \ac{BS} plane and the receiver plane; $m=-\frac{\ln2}{\ln(\cos{\Phi_{\rm a}})}$ is the Lambertian order and $\Phi_{\rm a}$ is the half-power semi-angle of the downlink \acp{LED}; and $\mathcal{J}_i$ denotes the index set of the interfering \acp{BS} for $\text{BS}_i$. The parameter $\Omega$ in \eqref{eq_4_1_1} is given by:
\begin{equation}
\Omega = \frac{4{\pi^2}{N_0}B_{\rm{a}}\xi_{\rm{a}}}{{{((m+1){h^{m+1}}{A_{{\rm{PD}}}}{R_{\rm{PD}}})}^2}{P_{\rm{a}}}},
\end{equation}
where $B_{\rm{a}}$ is the bandwidth of the access system\footnote{The \ac{LiFi} access system is assumed to have a low-pass and flat frequency response with a bandwidth of $B_{\rm{a}}$.}; $A_{\rm{PD}}$ is the photosensitive area of \ac{PD}; $R_{\rm{PD}}$ is the \ac{PD} responsivity; and $P_{\rm{a}}$ is the transmission power used for every \ac{BS}.

The downlink \ac{SINR} is a random variable through a transformation of the random coordinates of the \ac{UE}. For an unbounded hexagonal attocell network, the \ac{CDF} of the downlink \ac{SINR} is presented in \cite{CChen5}. A similar methodology is adopted to derive an analytical expression for the \ac{CDF} of $\gamma_u$ in \eqref{eq_4_1_1} as follows:
\begin{equation}\label{eq_4_1_2}
\mathbb{P}\left[\gamma_u\leq\gamma\right] = \frac{1}{2}-\frac{2}{\pi R_{\rm{e}}^2}\int\limits_0^{R_{\rm{e}}}\arcsin^\dagger\left(\mathcal{Z}(r,\gamma)\right)rdr,
\end{equation}
where $R_{\rm{e}}$ represents the radius of an equivalent circular cell preserving the area of the hexagonal cell with radius $R$; and:
\begin{equation}\label{eq_4_1_3}
\mathcal{Z}(r,\gamma) = \frac{2\gamma^{-1}\xi_{\rm{a}}^{-1}(r^2+h^2)^{-m-3}-2\Omega}{|\mathcal{I}_{0^\circ}(r)-\mathcal{I}_{30^\circ}(r)|}-\frac{\mathcal{I}_{0^\circ}(r)+\mathcal{I}_{30^\circ}(r)}{|\mathcal{I}_{0^\circ}(r)-\mathcal{I}_{30^\circ}(r)|},
\end{equation}
\begin{equation}\label{eq_4_1_4}
\arcsin^\dagger(x) = \left\{
\begin{matrix}
\frac{\pi}{2}, & x>1\\
\arcsin(x), & |x|\leq 1\\
-\frac{\pi}{2}, & x<-1
\end{matrix}\right..
\end{equation}
The functions $\mathcal{I}_{0^\circ}(r)$ and $\mathcal{I}_{30^\circ}(r)$ appearing in \eqref{eq_4_1_3} are available in closed form in \cite{CChen5}. Based on \eqref{eq_4_1_2}, the \ac{CDF} of $\gamma_u$ is efficiently computed by using numerical integration methods. Note that $\gamma_u$ is a bounded random variable such that:
\begin{subequations}\label{eq_4_1_6}
	\begin{gather}
	\gamma_{\rm{min}}\leq\gamma_u\leq\gamma_{\rm{max}},\label{eq_4_1_5}\\
	\gamma_{\rm{min}} = \frac{\xi_{\rm{a}}^{-1}(R_{\rm{e}}^2+h^2)^{-m-3}}{\mathcal{I}_{30^\circ}(R_{\rm{e}})+\Omega},\label{eq_4_1_6a}\\
	\gamma_{\rm{max}} = \frac{\xi_{\rm{a}}^{-1}h^{-2m-6}}{\mathcal{I}_{0^\circ}(0)+\Omega}.\label{eq_4_1_6b}
	\end{gather}
\end{subequations}

\subsubsection{Backhaul Signal-to-Noise-Ratio}
Because of having an equal link distance, backhaul links exhibit an identical \ac{SNR}\footnote{The wireless optical backhaul system operates over a frequency-flat channel dominated by the \ac{LOS} path.}. The received \ac{SNR} per subcarrier for ${\rm{b}}_i$ $\forall i$ is derived in \cite{Kazemi5}:
\begin{subequations}\label{eq_4_1_7}
	\begin{gather}
	\gamma_{{\rm{b}}_i} = K_i\gamma_{\rm{b}},\label{eq_4_1_7a}\\
	\gamma_{\rm{b}} = \frac{((\ell+1)A_{\rm{PD}}R_{\rm{PD}})^2{P_{\rm{a}}}}{72{\pi^2}{R^4}{N_0}B_{\rm{b}}\xi_{\rm{b}}^2}\label{eq_4_1_7b},
	\end{gather}
\end{subequations}
where $K_i=\frac{P_{{\rm{b}}_i}}{P_{\rm{a}}}$ is the power control coefficient for the link ${\rm{b}}_i$, and $P_{{\rm{b}}_i}$ is the corresponding transmission power; $\ell=-\frac{\ln2}{\ln(\cos{\Phi_{\rm{b}}})}$ is the Lambertian order with $\Phi_{\rm{b}}$ denoting the half-power semi-angle of the backhaul \acp{LED}; $B_{\rm{b}}$ is the bandwidth of the backhaul system; and $\xi_{\rm{b}}=\frac{N_{\rm{b}}-2}{N_{\rm{b}}}$.

\subsection{Achievable Rates of Access and Backhaul Systems}\label{sec_4_2_1}
The subchannel bandwidths of access and backhaul systems are matched so that $\frac{B_{\rm{a}}}{N_{\rm{a}}}=\frac{B_{\rm{b}}}{N_{\rm{b}}}$. This leads to the same symbol periods for \ac{DCO{-}OFDM} frames of the two systems. Denote by $\mathcal{L}_i$ the index set of \acp{BS} that use the link $\text{b}_i$ to connect to the gateway and denote by $\mathcal{U}_i$ the index set of \acp{UE} associated with $\text{BS}_i$ such that $|\mathcal{U}_i|=M_i$. Every \ac{UE} served by $\text{BS}_i$ acquires an equal bandwidth. Furthermore, let $\mathcal{R}_{{\rm{a}}_i}$ be the access sum rate for $\text{BS}_i$ and let $\mathcal{R}_{{\rm{b}}_i}$ be the overall achievable rate of $\text{b}_i$. It follows that:
\begin{subequations}\label{eq_3_2_12}
	\begin{gather}
	\mathcal{R}_{{\rm{a}}_i} = \frac{\xi_{\rm{a}}B_{\rm{a}}}{M_i}\sum_{u\in\mathcal{U}_i}\log_2(1+\gamma_u),\label{eq_3_2_12a}\\
	\mathcal{R}_{{\rm{b}}_i} = \xi_{\rm{b}}B_{\rm{b}}\log_2(1+\gamma_{{\rm{b}}_i}).\label{eq_3_2_12b}
	\end{gather}
\end{subequations}

\subsection{Decode-and-Forward Relaying and Backhaul Bandwidth Sharing}\label{sec2_num3}\label{sec_4_2_3}
In an $N_{\rm{T}}$-tier super cell, the $n$th tier encompasses $\frac{6n}{6}=n$ \acp{BS} for each branch so that $|\mathcal{T}_n|=n$ for $n=1,2,\dots,N_{\rm{T}}$, where $\mathcal{T}_n$ is the index set of \acp{BS} in the $n$th tier. Therefore, the total number of \acp{BS} per branch  excluding the central \ac{BS} is calculated by:
\begin{equation}\label{eq_4_2_3}
N_{\rm{BS}} = \sum_{n=1}^{N_{\rm{T}}}n = \frac{N_{\rm{T}}\left(N_{\rm{T}}+1\right)}{2}.
\end{equation}
For the $k$th branch of the backhaul network, the downlink data traffic for all $N_{\rm{BS}}$ \acp{BS} is carried by the link between the gateway and the first tier, i.e. $\text{b}_k$ for some $k\in\mathcal{T}_1$. This requires sufficient capacity for $\text{b}_k$ to respond to the aggregate sum rate of all $N_{\rm{BS}}$ \acp{BS}. However, such a challenging requirement is not always possible to be fulfilled in realistic scenarios where the limited capacity of $\text{b}_k$ may result in a \textit{backhaul bottleneck}. In this paper, the link $\text{b}_k$ $\forall k\in\mathcal{T}_1$ is generally referred to as a \textit{bottleneck link}. 

The use of \ac{DCO{-}OFDM} in conjunction with \ac{DF} relaying allows data multiplexing to be realized in the frequency domain. This way, the bandwidth of the bottleneck link ${\rm{b}}_k$ is divided into $N_{\rm{BS}}$ orthogonal sub-bands, with each sub-band allocated to an independent data flow. The symbols encapsulated in different sub-bands are individually and fully decoded at $\text{BS}_i$ in the first tier, which thereafter are reassembled into $N_{\rm{BS}}$ distinct groups. One group alone is modulated with a \ac{DCO{-}OFDM} frame and directly transmitted for the downlink of $\text{BS}_i$. The remaining $N_{\rm{BS}}-1$ groups are repackaged into separate \ac{DCO{-}OFDM} frames and forwarded in their desired directions toward higher tiers. The orthogonal decomposition of the effective bandwidth $\xi_{\rm{b}}B_{\rm{b}}$ into $N_{\rm{BS}}$ parts entails a weight coefficient $\mu_i\in[0,1]$ satisfying $\sum_{i\in\mathcal{L}_k}\mu_i=1$, thereby allocating a dedicated share of $\mu_i\xi_{\rm{b}}B_{\rm{b}}$ to $\text{BS}_i$ $\forall i\in\mathcal{L}_k$. In other words, the \ac{DCO{-}OFDM} frame is fragmented into $N_{\rm{BS}}$ segments, with each one independently loaded with the downlink data for $\text{BS}_i$. Hence, the required signal processing to discriminate between different sub-bands is performed in the frequency domain by using the \ac{FFT} of the received signal from ${\rm{b}}_k$.

\section{End-to-End Sum Rate Analysis}\label{sec_4_2}
The \textit{end-to-end sum rate} refers to the sum of the end-to-end rates of individual \acp{UE}. In this paper, two main policies are proposed for bandwidth allocation: \ac{UBS} and \ac{CBS}. The end-to-end sum rate under both policies are derived in the following.

\subsection{User-based Bandwidth Scheduling}\label{sec_4_2_4}
After performing bandwidth sharing, an independent pipeline is created to transport data from the gateway to every \ac{BS}. In \ac{UBS}, the dedicated portion of the backhaul bandwidth and the bandwidth of the access system are equally allocated to \acp{UE} for each \ac{BS}. The end-to-end rate of each \ac{UE} cannot be greater than the allocated capacity of each intermediate hop based on the maximum flow--minimum cut theorem \cite{Cover1}. Also, bandwidth sharing introduces a loss factor of $\mu_i$ into the end-to-end \ac{SE} of every \ac{UE}. For $\text{BS}_i$ $\forall i\in\mathcal{T}_1$, the $u$th UE $\forall u\in\mathcal{U}_i$ experiences an end-to-end rate of:
\begin{subequations}\label{eq_3_3_10}
	\begin{align}
	\mathcal{R}_u^{\rm{UBS}} &= \min\left[\frac{\mu_i\xi_{\rm{b}}B_{\rm{b}}}{M_i}\log_2(1+\gamma_{{\rm{b}}_i}),\frac{\xi_{\rm{a}}B_{\rm{a}}}{M_i}\log_2(1 + \gamma_u)\right],\label{eq_3_3_10a}\\
	&= \frac{\xi_{\rm{a}}B_{\rm{a}}}{M_i}\min\left[\mu_i\zeta\log_2(1+\gamma_{{\rm{b}}_i}),\log_2(1+\gamma_u)\right],\label{eq_3_3_10b}
	\end{align}
\end{subequations}
where $\zeta$ is defined as the effective bandwidth ratio:
\begin{equation}\label{eq5_sec3}
\zeta = \frac{\xi_{\rm{b}}B_{\rm{b}}}{\xi_{\rm{a}}B_{\rm{a}}}.
\end{equation}

To extend the analysis for the $n$th tier, note that the signals intended for \acp{BS} in the $n$th tier need to traverse exactly $n$ intermediate hops through backhaul links. The effective achievable rates of all those $n$ links are input to the $\rm{min}$ operator. Let $\mathcal{P}_i=\{j_1,j_2,\dots,j_n\}$ denote the path from the gateway to $\text{BS}_i$ for some $i\in\mathcal{T}_n$. The elements of $\mathcal{P}_i$ specify the indexes of backhaul links on the way to $\text{BS}_i$, among which $j_1$ indicates the bottleneck link. For example, $\mathcal{P}_{20}=\{1,8,20\}$ according to Fig.~\ref{fig_4_2_1}. Let $\mu_{i,j}$ be the bandwidth sharing ratio that is allocated to $\text{BS}_i$ at ${\rm{b}}_j$. To be consistent with the notation used for the first tier, $\mu_{i,j}=\mu_i$ for $j=j_1$. Obviously, for the last tier of an $N_{\rm{T}}$-tier super cell $\mu_{i,j_{N_{\rm{T}}}}=1$. Therefore, for $\text{BS}_i$ in the $n$th tier, the end-to-end rate of the $u$th \ac{UE} is written in a compact form:
\begin{equation}\label{eq_3_3_14}
\mathcal{R}_u^{\rm{UBS}} = \frac{\xi_{\rm{a}}B_{\rm{a}}}{M_i}\min\left[\underset{j\in\mathcal{P}_i}{\min}~\mu_{i,j}\zeta\log_2(1+\gamma_{{\rm{b}}_j}),\log_2(1+\gamma_u)\right].
\end{equation}
Note that for a one-tier super cell, \eqref{eq_3_3_14} reduces to \eqref{eq_3_3_10}, as the $\rm{min}$ operator is associative. The generalized end-to-end sum rate for $\text{BS}_i$ in the $n$th tier for $n=1,2,\dots,N_{\rm{T}}$ becomes:
\begin{equation}\label{eq_3_3_15}
\mathcal{R}_{{\rm{BS}}_i}^{\rm{UBS}} = \sum\limits_{u\in\mathcal{U}_i}\mathcal{R}_u^{\rm{UBS}},~\forall i\in\mathcal{T}_n
\end{equation}

\subsection{Cell-based Bandwidth Scheduling}\label{sec_4_2_5}
The point that distinguishes \ac{CBS} from \ac{UBS} is that in \ac{CBS}, the gateway puts up the entire data intended for each \ac{BS} in an exclusive set of subcarriers of the bottleneck backhaul. Then, the desired \ac{BS} assigns that given bandwidth equally to the associated \acp{UE}. The end-to-end sum rate of $\text{BS}_i$ in the $n$th tier is expressed mathematically as follows:
\begin{equation}\label{eq_3_3_16}
\mathcal{R}_{{\rm{BS}}_i}^{\rm{CBS}} = \min\left[\underset{j\in\mathcal{P}_i}{\min}~\mu_{i,j}\xi_{\rm{b}}B_{\rm{b}}\log_2(1+\gamma_{{\rm{b}}_j}),\frac{\xi_{\rm{a}}B_{\rm{a}}}{M_i}\sum\limits_{u\in\mathcal{U}_i}\log_2(1+\gamma_u)\right],~\forall i\in\mathcal{T}_n.
\end{equation}

\subsection{A System-Level Simplification}
With the assumption that a fixed power $P_{\rm{b}}$ is equally assigned to every individual backhaul link, the received \ac{SNR} of all the backhaul links become identical:
\begin{equation}\label{eq_4_4_1temp}
\gamma_{{\rm{b}}_i}=K_{\rm{b}}\gamma_{\rm{b}},~\forall i\in\mathcal{L}_k
\end{equation}
where $K_{\rm{b}}=\frac{P_{\rm{b}}}{P_{\rm{a}}}$ is a common power control coefficient for the backhaul system\footnote{$K_{\rm{b}}$ also represents the total power of the backhaul system normalized by that of the access system, i.e. $\frac{\sum_{i\in\mathcal{L}_k}P_{{\rm{b}}_i}}{N_{\rm{BS}}P_{\rm{a}}} = K_{\rm{b}}$.}. A judicious design consists in choosing bandwidth allocation ratios for the outer tiers so that intermediate hops do not restrict the effective achievable rate in the path from the gateway to the desired \ac{BS}. One such design is to make the bandwidth sharing coefficients in the outer tiers proportional to that of the bottleneck link according to the following normalization:
\begin{equation}\label{eq_4_4_2temp}
\mu_{i,j} = \frac{\mu_i}{\sum_{i'\in\mathcal{L}_j}\mu_{i'}}>\mu_i,~\forall i\in\mathcal{L}_j
\end{equation}
The inequality $\mu_{i,j}>\mu_i$ is derived from the fact that $\sum_{i'\in\mathcal{L}_j}\mu_{i'}<1$ when $j\in\mathcal{T}_n$ $\forall n>1$. As a result:
\begin{equation}\label{eq_4_4_3temp}
\underset{j\in\mathcal{P}_i}{\min}~\mu_{i,j} = \mu_i.
\end{equation}

\subsubsection{UBS}
By using \eqref{eq_4_4_1temp} and \eqref{eq_4_4_3temp}, the term representing the rate of $\mathcal{P}_i$ in \eqref{eq_3_3_14} simplifies to:
\begin{equation}\label{eq_4_4_1}
\underset{j\in\mathcal{P}_i}{\min}~\mu_{i,j}\zeta\log_2(1+\gamma_{{\rm{b}}_j}) = \mu_i\zeta\log_2(1+\gamma_{{\rm{b}}_k}),
\end{equation}
where $k$ signifies the index of the bottleneck link, which can be calculated by $k=\left\lfloor\frac{i-(3n-1)(n-1)}{n}\right\rfloor$ for $\text{BS}_i$ $\forall i\in\mathcal{T}_n$ for $n=1,2,\dots,N_{\rm{T}}$. As a sanity check, for a special case of $n=1$, this generalized indicator returns $k=i$, conforming with \eqref{eq_3_3_10}. In effect, the dominant hop along the backhaul path is merely posed by the link ${\rm{b}}_k$. For $\text{BS}_i$ in the $n$th tier, the end-to-end transmission rate of the $u$th \ac{UE} in \eqref{eq_3_3_14} reduces to a more tractable form of:
\begin{equation}\label{eq_4_2_9}
\mathcal{R}_u^{\rm{UBS}} = \frac{\xi_{\rm{a}}B_{\rm{a}}}{M_i}\min\left[\mu_i\zeta\log_2(1+\gamma_{{\rm{b}}_k}),\log_2(1 + \gamma_u)\right],~\forall u\in\mathcal{U}_i
\end{equation}

\subsubsection{CBS}
Based on \eqref{eq_4_4_1temp} and \eqref{eq_4_4_2temp}, the end-to-end sum rate of $\text{BS}_i$ in \eqref{eq_3_3_16} is simplified to:
\begin{equation}\label{eq_4_2_13}
\mathcal{R}_{{\rm{BS}}_i}^{\rm{CBS}} = \min\left[\mu_i\xi_{\rm{b}}B_{\rm{b}}\log_2(1+\gamma_{{\rm{b}}_k}),\frac{\xi_{\rm{a}}B_{\rm{a}}}{M_i}\sum\limits_{u\in\mathcal{U}_i}\log_2(1+\gamma_u)\right],~\forall i\in\mathcal{T}_n.
\end{equation}
For completeness, the end-to-end rate of the $u$th \ac{UE} $\forall u\in\mathcal{U}_i$ for \ac{CBS} is obtained by using \eqref{eq_3_2_12}:
\begin{equation}\label{eq_4_2_14}
\mathcal{R}_u^{\rm{CBS}} =
\begin{cases}
\frac{\mu_i\xi_{\rm{b}}B_{\rm{b}}}{M_i}\log_2(1+\gamma_{{\rm{b}}_k}), & \mu_i\leq\frac{\mathcal{R}_{{\rm{a}}_i}}{\mathcal{R}_{{\rm{b}}_k}} \\
\frac{\xi_{\rm{a}}B_{\rm{a}}}{M_i}\log_2(1+\gamma_u), & \mu_i>\frac{\mathcal{R}_{{\rm{a}}_i}}{\mathcal{R}_{{\rm{b}}_k}}
\end{cases}
\end{equation}

\section{Optimal Bandwidth Scheduling}\label{sec_4_4}
This section focuses on the problem of optimal bandwidth scheduling. In particular, the design of bandwidth sharing coefficients for the generalized case of multi-tier super cells is formulated as an optimization problem aiming for the end-to-end sum rate maximization.

\subsection{Optimal User-based Bandwidth Scheduling}\label{sec_4_4_1}
The purpose of optimal \ac{UBS} is to maximize the sum of \textit{per-user} end-to-end rates under the \ac{UBS} policy. Based on \eqref{eq_4_2_9}, the optimization problem for the $k$th branch of the super cell is stated in the global form:
\begin{maxi!}[2]
	{\{\mu_i\in\mathbb{R}\}}{\sum_{i\in\mathcal{L}_k}\sum_{u\in\mathcal{U}_i}\frac{\xi_{\rm{a}}B_{\rm{a}}}{M_i}\min\left[\mu_i\zeta\log_2(1+\gamma_{{\rm{b}}_k}),\log_2(1 + \gamma_u)\right]\label{eq_4_4_1objective}}
	{\label{eq_4_4_1problem}}{}
	\addConstraint{\sum_{i\in\mathcal{L}_k}\mu_i}{=1\label{eq_4_4_1a}}
	\addConstraint{0\leq\mu_i\leq 1,}{~\forall i\in\mathcal{L}_k\label{eq_4_4_1b}}
\end{maxi!}
The constraints \eqref{eq_4_4_1a} and \eqref{eq_4_4_1b} are discussed in Section~\ref{sec2_num3}. For global optimization of the bandwidth allocation, the downlink \ac{SINR} for entire \acp{UE} in the $k$th branch is processed by a central controller. Such an assumption is justified for indoor wireless optical channels for two reasons: 1) the short wavelength of the optical carrier along with the large photosensitive area of the \ac{PD} eliminate rapid signal fluctuations due to multipath fading \cite{Kahn1}; 2) in realistic indoor scenarios, the \acp{UE} are inclined to be static or slowly moving. Under such quasi-static conditions, it is possible to acquire an accurate estimate of the downlink channel state with a small overhead based on a limited content feedback mechanism, which relies upon updating the average received power \cite{Soltani}. Consequently, each \ac{BS} collects the \ac{SINR} information from an uplink channel and sends it to the central controller for optimization of the bandwidth allocation.

The objective function in \eqref{eq_4_4_1objective} can be expanded by factorizing a constant term $\zeta\log_2(1+\gamma_{{\rm{b}}_k})$ and defining a variable $\rho_u$ to be the normalized achievable rate for the $u$th \ac{UE}:
\begin{equation}\label{eq_4_4_3}
\rho_u = \frac{\log_2(1+\gamma_u)}{\zeta\log_2(1+\gamma_{{\rm{b}}_k})}.
\end{equation}
The factor $\xi_{\rm{b}}B_{\rm{b}}\log_2(1+\gamma_{{\rm{b}}_k})$ is independent of optimization variables and it can be put aside without affecting the problem in \eqref{eq_4_4_1problem}. This leads to a compact form of:
\begin{maxi!}
	{\{\mu_i\in\mathbb{R}\}}{\sum_{i\in\mathcal{L}_k}\sum_{u\in\mathcal{U}_i}\frac{1}{M_i}\min[\mu_i,\rho_u]\label{eq_4_4_4a}}
	{\label{eq_4_4_4}}{}
	\addConstraint{\eqref{eq_4_4_1a}~\&~\eqref{eq_4_4_1b}}{\label{eq_4_4_4b}}
\end{maxi!}
The objective function in \eqref{eq_4_4_4a} is a composite of concave operators, comprising summation and minimization. Such a composition preserves concavity and the objective function is concave \cite{Boyd1}. Therefore, this is a convex optimization problem with linear constraints, for which Slater's condition holds and there is a global optimum \cite{Bertsekas}. However, standard methods such as Lagrange multipliers cannot be directly applied to find an analytical solution because the objective function is not differentiable in $\boldsymbol{\mu}=[\mu_i]_{N_{\rm{BS}}\times1}$, where $\boldsymbol{\mu}$ is the vector of optimization variables.

For nonsmooth optimization, the subgradient method is a means to deal with nondifferentiable convex functions \cite{Boyd2}. Particularly, the constrained optimization problem in \eqref{eq_4_4_4} can be efficiently solved by using the \textit{projected} subgradient method. Analogous to common subgradient methods, the vector $\boldsymbol{\mu}$ is sequentially updated using a subgradient of the objective function at $\boldsymbol{\mu}$. Compared with an ordinary subgradient method, there is an additional constraint $\mathbf{1}^T\boldsymbol{\mu}=1$, with $\mathbf{1}$ denoting an all-ones vector of size $N_{\rm{BS}}\times1$, which is required by \eqref{eq_4_4_1a}. To fulfil this constraint, at each iteration, the projected approach maps the components of $\boldsymbol{\mu}$ onto a unit space before proceeding with the next update, to bring them back to the feasible set. The convergence is attained upon setting a suitable step size for executing iterations \cite{Boyd2}. To develop an efficient iterative algorithm, an appropriate subgradient vector is required to provide a descent direction for a local maximizer to approach the global maximum when updating. To this end, the problem statement needs to be properly modified. The users in the attocell of $\text{BS}_i$ are split into two disjoint groups: those for whom $\mu_i>\rho_u$ and those for whom $\mu_i\leq\rho_u$. The index sets for these two groups are denoted by $\hat{\mathcal{U}}_i$ and $\check{\mathcal{U}}_i$, respectively, implying $\hat{\mathcal{U}}_i\cup\check{\mathcal{U}}_i=\mathcal{U}_i$. The number of elements corresponding to $\hat{\mathcal{U}}_i$ and $\check{\mathcal{U}}_i$ is represented by $\hat{M}_i$ and $\check{M}_i$ so that $\hat{M}_i+\check{M}_i=M_i$. The optimization problem in \eqref{eq_4_4_4} is then stated in the desired form:
\begin{maxi!}
	{\{\mu_i\in\mathbb{R}\}}{\sum_{i\in\mathcal{L}_k}\left[\sum_{u\in\hat{\mathcal{U}}_i}\frac{\rho_u}{M_i}+\frac{\check{M}_i}{M_i}\mu_i\right]\label{eq_4_4_5a}}
	{\label{eq_4_4_5}}{}
	\addConstraint{\eqref{eq_4_4_1a}~\&~\eqref{eq_4_4_1b}}{\label{eq_4_4_5b}}
\end{maxi!}
Note that the arrangements of $\check{\mathcal{U}}_i$ and $\hat{\mathcal{U}}_i$ depend on the value of $\mu_i$. Based on \eqref{eq_4_4_5a}, the derivative of the objective function with respect to $\mu_i$ is estimated by $\frac{\check{M}_i}{M_i}$, resulting in the subgradient vector $\mathbf{g}=[g_i]_{N_{\rm{BS}}\times1}$ where $g_i=\frac{\check{M}_i}{M_i}$. The projected subgradient method for solving the primal problem is summarized in Algorithm~\ref{Algorithm1}. In the first line of this algorithm, $\alpha$ is the step size for updating, which is chosen to be sufficiently small; and in step~\ref{Update}, $\mathbf{P}$ is an $N_{\rm{BS}}\times N_{\rm{BS}}$ unitary space projection matrix \cite{EChong}, which is obtained as follows:
\begin{equation}
\mathbf{P} = \mathbf{I}-\mathbf{1}\left(\mathbf{1}^T\mathbf{1}\right)^{-1}\mathbf{1}^T = \mathbf{I}-\frac{1}{3}\mathbf{J},
\end{equation}
where $\mathbf{I}$ and $\mathbf{J}$ respectively represent an identity matrix and an all-ones matrix of size $N_{\rm{BS}}\times N_{\rm{BS}}$.

\begin{algorithm}[t!]
	\caption{Projected Subgradient Algorithm for Optimal User-based Bandwidth Scheduling.} {\fontsize{12}{10}\selectfont
		\begin{algorithmic}[1]
			\State Choose $\alpha$
			\State Initialize $\boldsymbol{\mu}^{(0)}$
			\For {all $i\in\mathcal{L}_k$} \label{Loop1}
			\State Let $\check{\mathcal{U}}_i^{(l)}=\left\{u\in\mathcal{U}_i\Big|\mu_i^{(l)}\leq\rho_u\right\}$
			\State Compute $\check{M}_i^{(l)}=\big|\check{\mathcal{U}}_i^{(l)}\big|$ 
			\State Compute $g_i^{(l)}=\dfrac{\check{M}_i^{(l)}}{M_i}$
			\EndFor
			\State Update $\boldsymbol{\mu}^{(l)}$ through $\boldsymbol{\mu}^{(l+1)}=\boldsymbol{\mu}^{(l)}-\alpha\mathbf{P}\mathbf{g}^{(l)}$ \label{Update}
			\State $l\gets l+1$ 
			\State \Goto{Loop1}
			\State Return $\boldsymbol{\mu}$
		\end{algorithmic}}
	\label{Algorithm1}
\end{algorithm}

\subsection{Optimal Cell-based Bandwidth Scheduling}\label{sec_4_4_2}
The scheduler aims to maximize the aggregate \textit{per-cell} end-to-end sum rates under the \ac{CBS} policy by computing an optimal solution to the following bandwidth allocation problem. For the $k$th branch of the super cell, by using \eqref{eq_4_2_13}, the optimization problem is:
\begin{maxi!}
	{\{\mu_i\in\mathbb{R}\}}{\sum_{i\in\mathcal{L}_k}\min\left[\mu_i\xi_{\rm{b}}B_{\rm{b}}\log_2(1+\gamma_{{\rm{b}}_k}),\frac{\xi_{\rm{a}}B_{\rm{a}}}{M_i}\sum\limits_{u\in\mathcal{U}_i}\log_2(1+\gamma_u)\right]\label{eq_4_4_7objective}}
	{\label{eq_4_4_7problem}}{}
	\addConstraint{\eqref{eq_4_4_1a}~\&~\eqref{eq_4_4_1b}}{\label{eq_4_4_7b}}
\end{maxi!}
The central controller only gathers the overall access sum rate information sent individually by each \ac{BS} via the feedback channel for further processing. This reduces the feedback overhead with respect to \ac{UBS}, which appeals to applications where limited feedback is available \cite{Soltani}.

Similar to the optimal \ac{UBS} case, the optimal \ac{CBS} problem in \eqref{eq_4_4_7problem} is reformulated as follows:
\begin{maxi!}
	{\{\mu_i\in\mathbb{R}\}}{\sum_{i\in\mathcal{L}_k}\min\left[\mu_i,\frac{1}{M_i}\sum_{u\in\mathcal{U}_i}\rho_u\right]\label{eq_4_4_6a}}
	{\label{eq_4_4_6}}{}
	\addConstraint{\eqref{eq_4_4_1a}~\&~\eqref{eq_4_4_1b}}{\label{eq_4_4_6b}}
\end{maxi!}
where $\rho_u$ is given by \eqref{eq_4_4_3}. The projected subgradient method is used to solve the primal problem. With the current expression in \eqref{eq_4_4_6a}, the objective function is not differentiable in $\boldsymbol{\mu}$. To find the candidate subgradient vector, the \acp{BS} of the $k$th branch are classified into two categories: those that fulfil the condition $\mu_i>\frac{1}{M_i}\sum_{u\in\mathcal{U}_i}\rho_u$ and those that satisfy $\mu_i\leq\frac{1}{M_i}\sum_{u\in\mathcal{U}_i}\rho_u$. The former category is represented by an index set of $\hat{\mathcal{L}}_k$ and the latter case by $\check{\mathcal{L}}_k$. The optimization problem in \eqref{eq_4_4_6} turns into:
\begin{maxi!}
	{\{\mu_i\in\mathbb{R}\}}{\sum_{i\in\hat{\mathcal{L}}_k}\frac{1}{M_i}\sum_{u\in\mathcal{U}_i}\rho_u+\sum_{i\in\check{\mathcal{L}}_k}\mu_i\label{eq_4_4_8a}}
	{\label{eq_4_4_8}}{}
	\addConstraint{\eqref{eq_4_4_1a}~\&~\eqref{eq_4_4_1b}}{\label{eq_4_4_8b}}
\end{maxi!}
Therefore, the derivative of the objective function with respect to $\mu_i$ is equal to $1$, leading to the subgradient vector $\mathbf{g}=[g_i]_{N_{\rm{BS}}\times1}$ where:
\begin{equation}\label{eq_4_4_9}
g_i = \left\{
\begin{matrix}
1, & i\in\check{\mathcal{L}}_k\\
0, & i\in\hat{\mathcal{L}}_k
\end{matrix}\right.
\end{equation}
The projected subgradient method used to solve the primal problem is outlined in Algorithm~\ref{Algorithm2}.

\begin{algorithm}[t!]
	\caption{Projected Subgradient Algorithm for Optimal Cell-based Bandwidth Scheduling.} {\fontsize{12}{10}\selectfont
		\begin{algorithmic}[1]
			\State Choose $\alpha$
			\State Initialize $\boldsymbol{\mu}^{(0)}$			
			\State Let $\check{\mathcal{L}}_k^{(l)}=\left\{i\in\mathcal{L}_k\Big|\mu_i^{(l)}\leq\frac{1}{M_i}\sum_{u\in\mathcal{U}_i}\rho_u\right\}$ \label{Loop2}
			\For {all $i\in\mathcal{L}_k$}
			\If {$i\in\check{\mathcal{L}}_k$}
			\State Set $g_i^{(l)}=1$
			\Else 
			\State Set $g_i^{(l)}=0$
			\EndIf
			\EndFor
			\State Update $\boldsymbol{\mu}^{(l)}$ through $\boldsymbol{\mu}^{(l+1)}=\boldsymbol{\mu}^{(l)}-\alpha\mathbf{P}\mathbf{g}^{(l)}$
			\State $l\gets l+1$ 
			\State \Goto{Loop2}
			\State Return $\boldsymbol{\mu}$
		\end{algorithmic}}
	\label{Algorithm2}
\end{algorithm}

\begin{table}[t!]
	\centering
	\caption{Simulation Parameters}
	\begin{tabular}{l | c | l}
		{Parameter} & {Symbol} & {Value} \\
		\hline
		{Downlink LED Optical Power} & $P_{\rm opt}$ & $10$ $\rm W$ \\
		{Downlink LED Semi-Angle} & $\Phi_{\rm{a}}$ & $40^\circ$ \\
		{Vertical Separation} & $h$ & $2.25$ $\rm m$ \\
		{Hexagonal Cell Radius} & $R$ & $2.5$ $\rm m$ \\
		{Total VLC Bandwidth} & $B$ & $20$ $\rm MHz$ \\
		{IFFT/FFT Length} & $N$ & $1024$ \\
		{Noise Power Spectral Density} & $N_0$ & $5\times10^{-22}$ $\rm A^2/Hz$ \\
		{UE Receiver Field of View} & $\Psi_{\rm a}$ & $85^\circ$ \\
		{PD Effective Area} & $A_{\rm PD}$ & $10^{-4}$ $\rm m^2$ \\
		{PD Responsivity} & $R_{\rm PD}$ & $0.6$ $\rm A/W$ \\
		{DC Bias Scaling Factor} & $\alpha$ & $3$ \\
		\hline
	\end{tabular}
	\label{tbl_3_4_1}
\end{table}

\subsection{Numerical Results and Discussions}\label{sec_4_4_3}
This section presents performance results for optimal \ac{UBS} and optimal \ac{CBS} policies based on Algorithm~\ref{Algorithm1} and Algorithm~\ref{Algorithm2}, respectively. To assess the optimality of the proposed algorithms, equal bandwidth scheduling is also included as a baseline policy. It allocates an equal fraction of bandwidth to every \ac{BS} in the same backhaul branch without distinction, i.e. $\mu_i=\frac{1}{N_{\rm{BS}}}$ $\forall i\in\mathcal{L}_k$ for the $k$th branch of an $N_{\rm{T}}$-tier super cell. The optimal and equal scheduling cases are marked with `OPT' and `EQL', respectively. The end-to-end sum rate performance is evaluated based on Section~\ref{sec_4_2}. The achievable rate of the access network with an unlimited backhaul capacity is considered and labeled as `Access Limit'. Monte-Carlo simulations are conducted over many random realizations to distribute multiple \acp{UE} uniformly over the network. For a fair comparison between super cells with a different number of tiers, the results are presented in terms of the average \ac{UE} density, which is defined as the ratio of the total number of \acp{UE} to that of \acp{BS}:
\begin{equation}\label{eq_4_3_1}
\lambda = \frac{M}{N_{\rm{BS}}}~\text{UE/Cell}.
\end{equation}
Table~\ref{tbl_3_4_1} lists the system parameters used for simulations. The configurations for cell radius and downlink \ac{LED} semi-angle are adopted from the guidelines provided in \cite{CChen5}.

\begin{figure}[!t]
	\centering
	\subfloat[\label{fig_4_4_1a} $\lambda=1$ UE/Cell]{\includegraphics[width=0.5\linewidth]{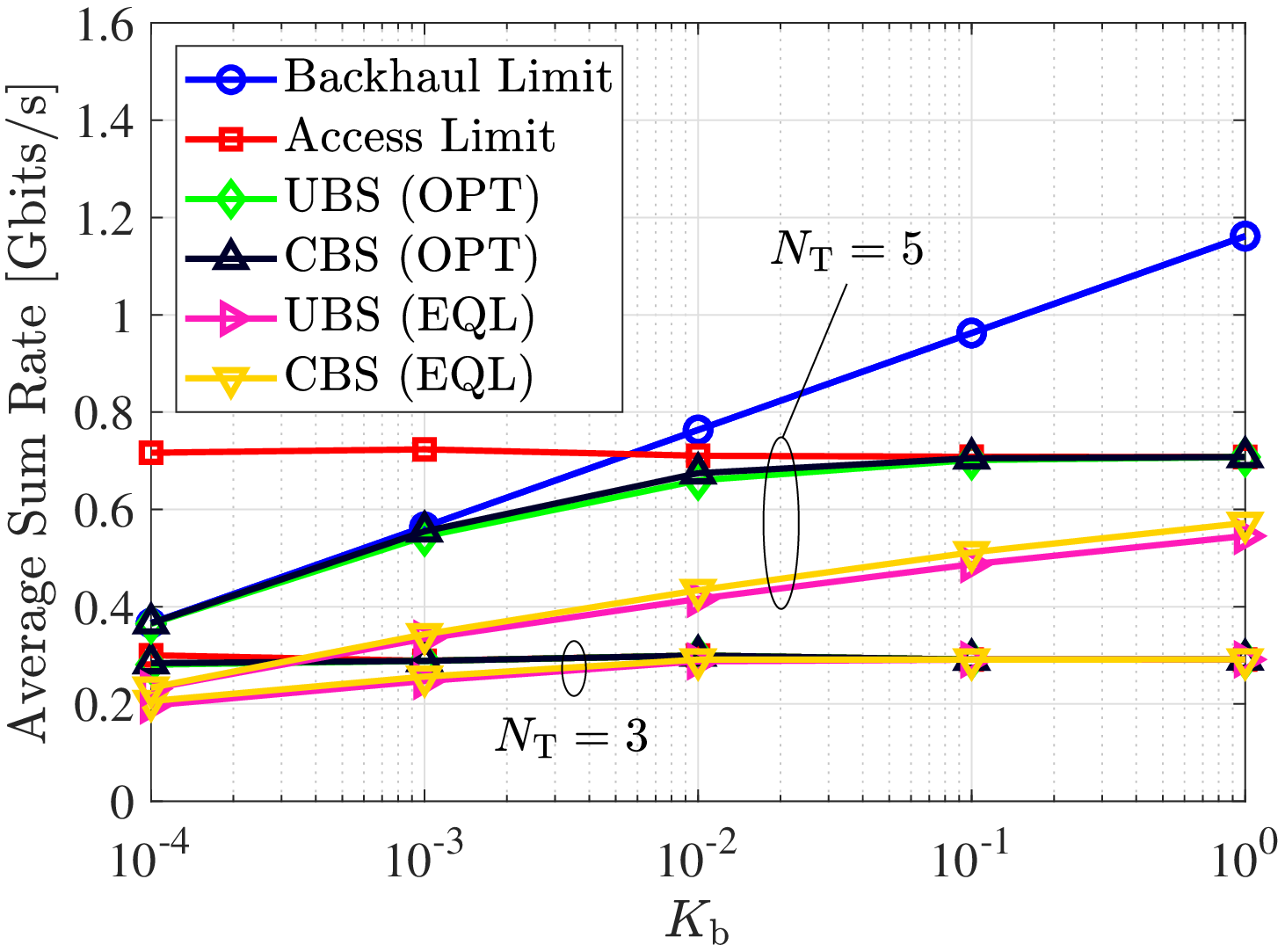}}
	\subfloat[\label{fig_4_4_1b} $\lambda=5$ UE/Cell]{\includegraphics[width=0.5\linewidth]{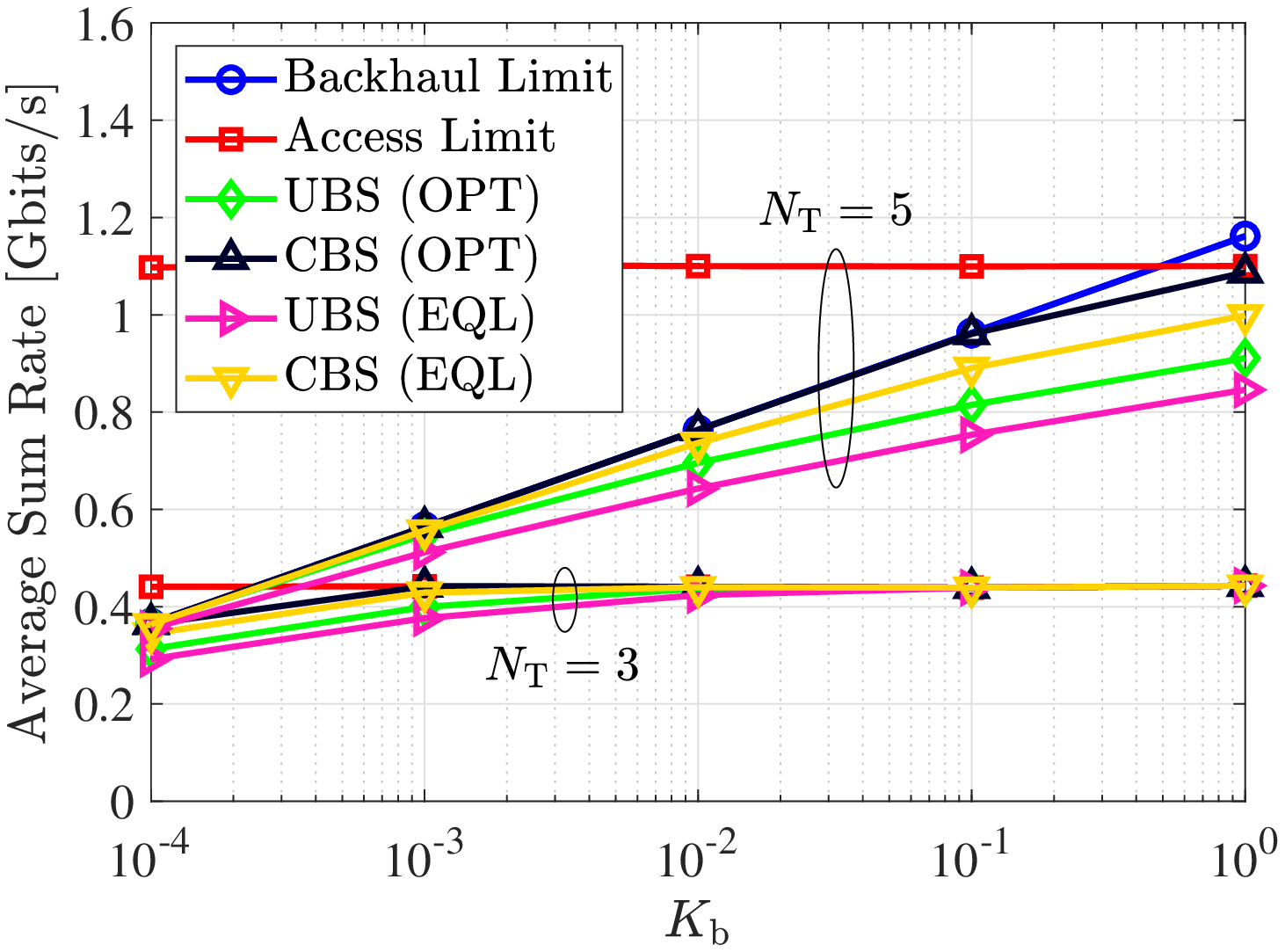}}
	\caption{Average sum rate performance of optimal UBS and optimal CBS policies as a function of the power ratio $K_{\rm{b}}$ for different values of $N_{\rm{T}}$ and $\lambda$; and $B_{\rm{b}}=3B_{\rm{a}}$.}
	\label{fig_4_4_1}
\end{figure}

Fig.~\ref{fig_4_4_1} shows the average sum rate performance for one branch of an $N_{\rm{T}}$-tier super cell as a function of the backhaul power ratio $K_{\rm{b}}$ for different values of $N_{\rm{T}}$ and $\lambda$. A key principle for understanding the impact of backhaul and access networks on the end-to-end performance relates to \textit{rate limit}. This concept indicates the effective upper bound of the end-to-end sum rate as imposed by both backhaul and access systems, i.e. $\min[\text{Backhaul Limit},\text{Access Limit}]$. For a low \ac{UE} density scenario as shown in Fig.~\ref{fig_4_4_1a}, for $N_{\rm{T}}=5$, both optimal policies maximally achieve the end-to-end rate limit over a broad range of values for $K_{\rm{b}}$. Note that the optimal algorithms operate whether backhaul or access limits the end-to-end performance. Fig.~\ref{fig_4_4_1a} demonstrates when the difference between backhaul and access limits is large enough, both UBS-OPT and CBS-OPT fully attain the rate limit, which is the case for $K_{\rm{b}}<10^{-3}$ and $K_{\rm{b}}\geq10^{-1}$. Moreover, it can be observed that both UBS-OPT and CBS-OPT cases improve the performance against their respective baseline policies of UBS-EQL and CBS-EQL. The improvement is as much as $250$ Mbits/s by choosing $K_{\rm{b}}=10^{-2}$. For $N_{\rm{T}}=3$, the overall rate of backhaul is sufficiently higher than that of access especially for $K_{\rm{b}}\geq10^{-2}$, in which case the performance for all scheduling policies coincide.

Fig.~\ref{fig_4_4_1b} plots the same set of results as in Fig.~\ref{fig_4_4_1a}, by considering a high \ac{UE} density scenario of $\lambda=5$ UE/Cell. Foremost, such an increase in the \ac{UE} density causes the access rate limit to rise, which is more pronounced for $N_{\rm{T}}=5$. In this case, the backhaul enforces a bottleneck on the end-to-end transmission, and evidently CBS-OPT makes perfect use of the limited backhaul capacity by following its growing trend when $K_{\rm{b}}$ increases. For instance, CBS-OPT successfully reaches an average sum rate of just below $1$ Gbits/s for $K_{\rm{b}}=10^{-1}$, as supplied by the backhaul system. Compared to Fig.~\ref{fig_4_4_1a}, the extent of improvement offered by optimal scheduling relative to equal scheduling is lower in Fig.~\ref{fig_4_4_1b}, still this is enhanced by heightening the backhaul power. Furthermore, it is observed that \ac{CBS} performs even better than \ac{UBS}. There is also a small gap between the results of \ac{CBS} and \ac{UBS} in Fig.~\ref{fig_4_4_1a}, but the difference in performance is manifested in Fig.~\ref{fig_4_4_1b} when the number of \acp{UE} per cell is multiplied fivefold.

\begin{figure}[!t]
	\centering
	\subfloat[\label{fig_4_4_2a} $K_{\rm{b}}=1$]{\includegraphics[width=0.5\linewidth]{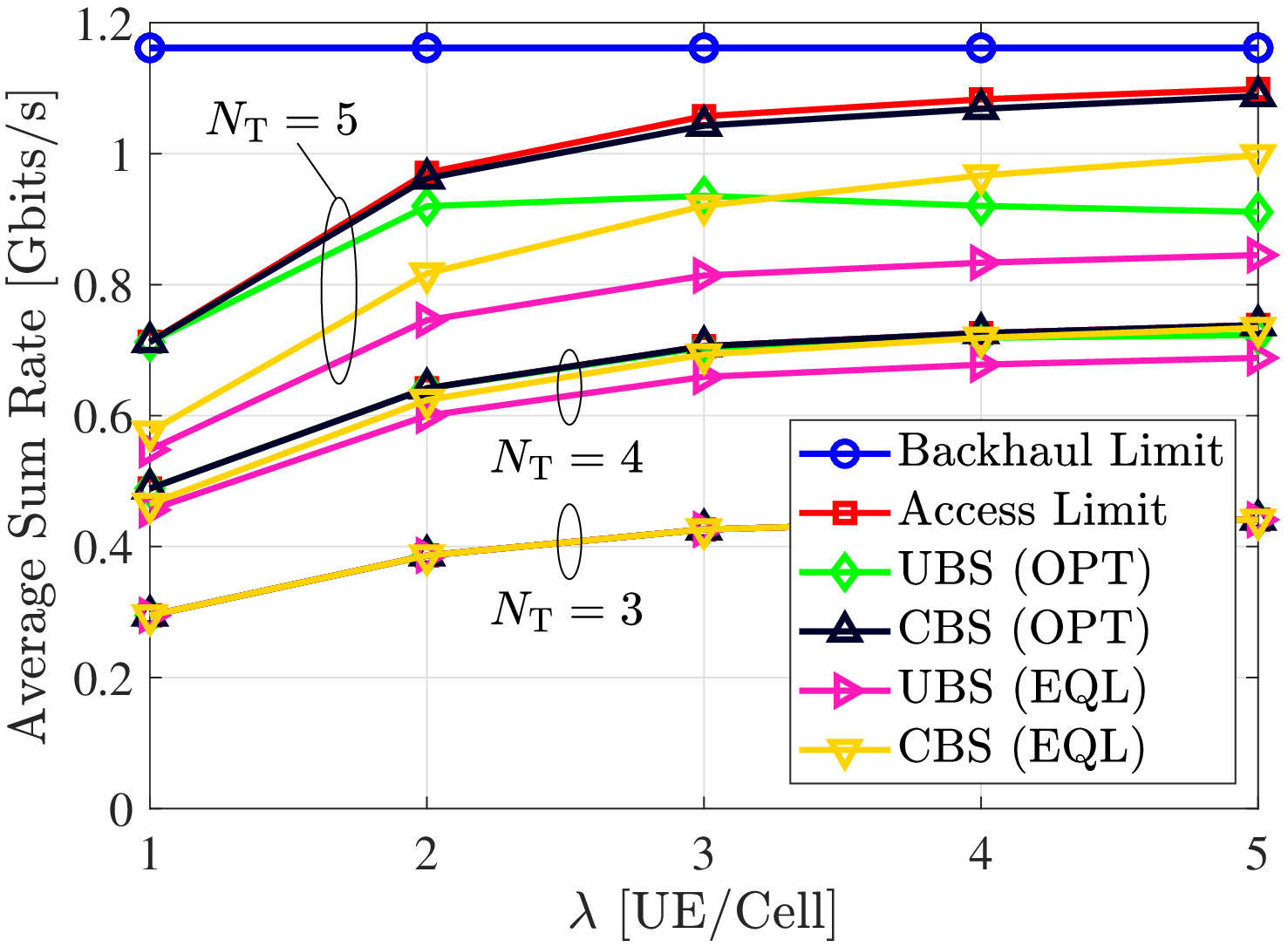}}
	\subfloat[\label{fig_4_4_2b} $K_{\rm{b}}=10^{-2}$]{\includegraphics[width=0.5\linewidth]{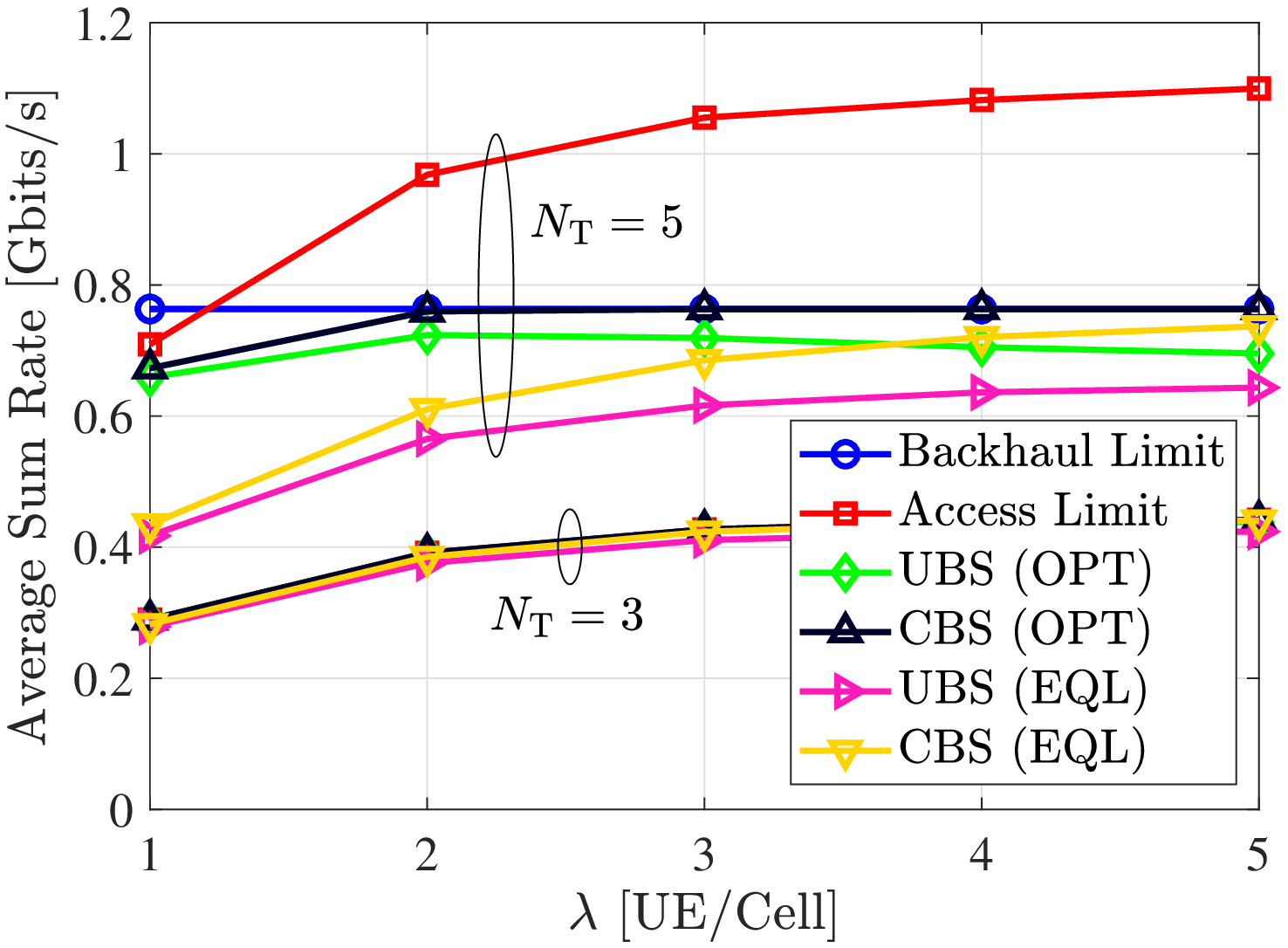}}
	\caption{Average sum rate performance of optimal UBS and optimal CBS policies as a function of the UE density $\lambda$ for different values of $N_{\rm{T}}$ and $K_{\rm{b}}$; and $B_{\rm{b}}=3B_{\rm{a}}$.}
	\label{fig_4_4_2}
\end{figure}

Fig.\ref{fig_4_4_2} illustrates the average sum rate performance with respect to the \ac{UE} density $\lambda$ for different combinations of $N_{\rm{T}}$ and $K_{\rm{b}}$. For $N_{\rm{T}}=5$, Fig.\ref{fig_4_4_2a} ($K_{\rm{b}}=1$) represents a case where the access limit is located under the backhaul limit, while Fig.\ref{fig_4_4_2b} ($K_{\rm{b}}=10^{-2}$) constitutes the converse case in which the backhaul limit dominates for the majority of values of $\lambda$. In either case, similar to Fig.~\ref{fig_4_4_1}, the optimal algorithms outperform their baseline counterparts. It is observed that CBS-OPT consistently retains the achievable rate limit as the \ac{UE} density is increased. Also, CBS-OPT performs better than UBS-OPT, like the case in Fig.~\ref{fig_4_4_1b}. An explanation for this effect can be given by noting the operation principals of \ac{CBS} and \ac{UBS} systems. The per cell bandwidth allocation in \ac{CBS} is compatible with the notion of the rate limit, which means it can efficiently adapt to the limits of access and backhaul networks. By contrast, the \ac{UBS} system assigns the backhaul bandwidth in a per user basis and therefore introduces a degree of loss into the sum rate performance when aggregating the end-to-end rates achieved by individual \acp{UE}.

\begin{figure}[!t]
	\centering
	\subfloat[\label{fig_4_4_3a}]{\includegraphics[width=0.5\linewidth]{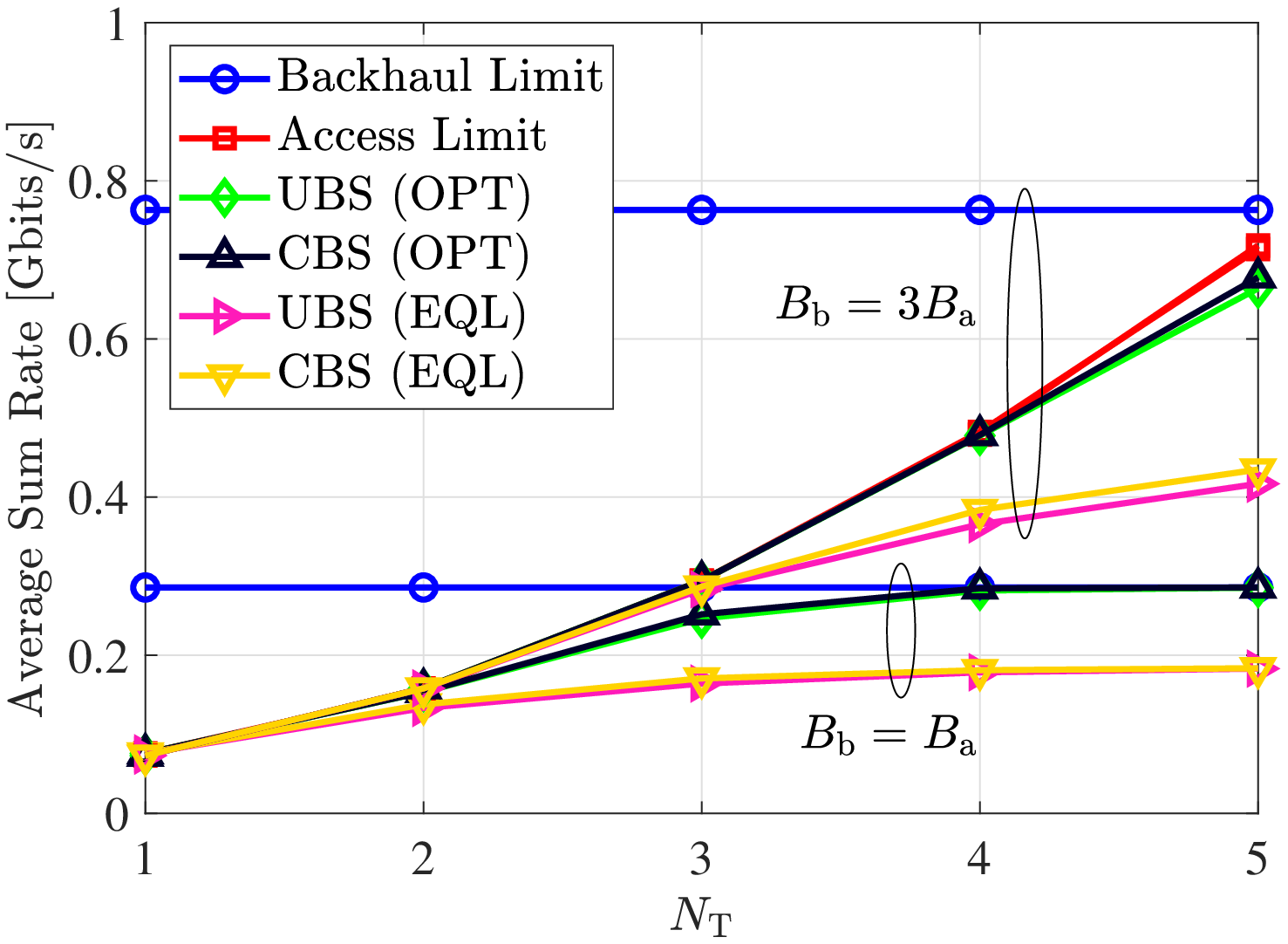}}
	\subfloat[\label{fig_4_4_3b}]{\includegraphics[width=0.5\linewidth]{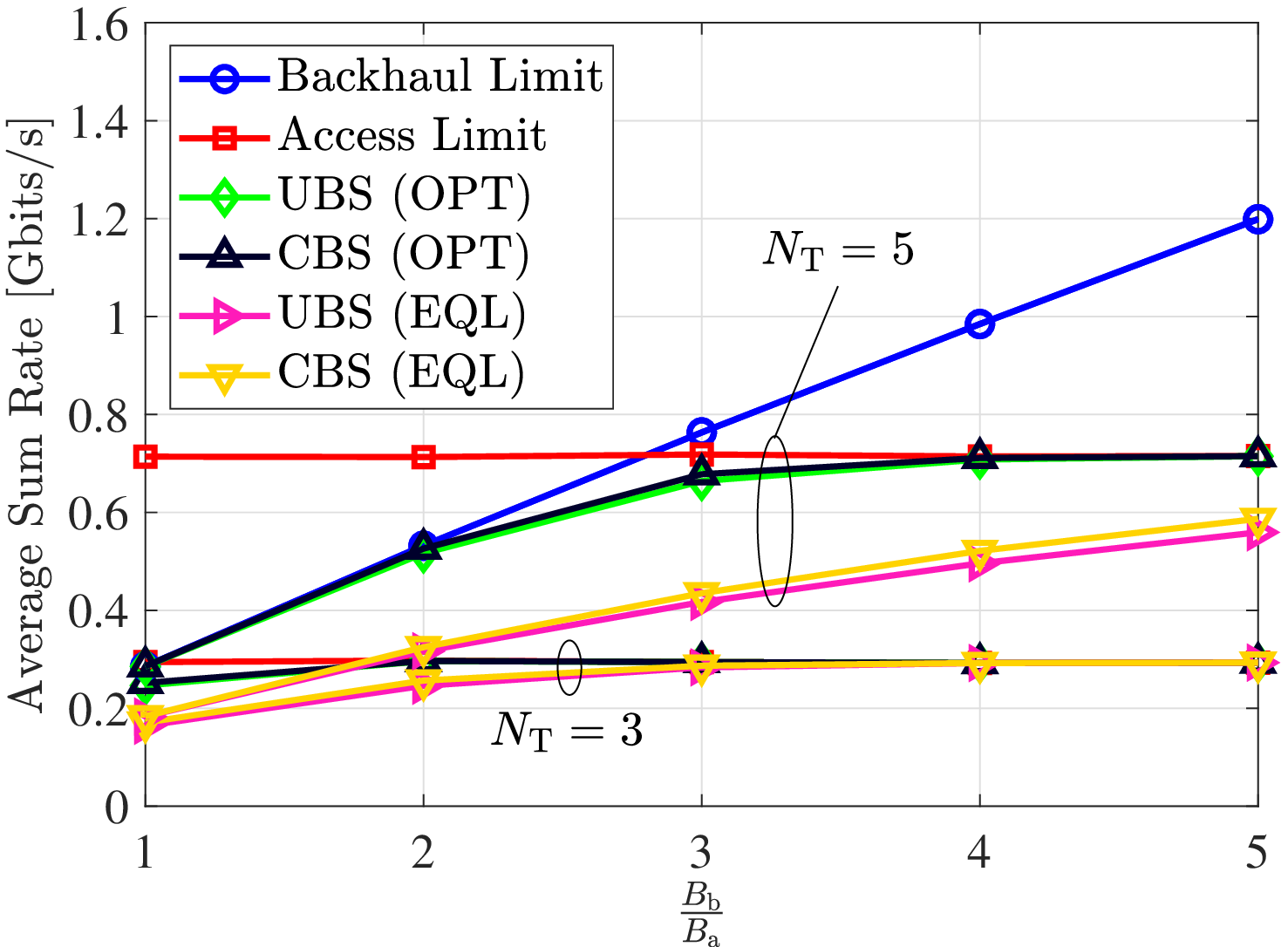}}
	\caption{Average sum rate performance of optimal UBS and optimal CBS policies for $K_{\rm{b}}=10^{-2}$ and $\lambda=1$ UE/Cell: (a) versus the total number of tiers $N_{\rm{T}}$; (b) versus the bandwidth ratio $\frac{B_{\rm{b}}}{B_{\rm{a}}}$.}
	\label{fig_4_4_3}
\end{figure}

For completeness, the average sum rate performance versus the number of tiers $N_{\rm{T}}$ is presented in Fig.~\ref{fig_4_4_3a}; for $K_{\rm{b}}=10^{-2}$ and $\lambda=1$ UE/Cell. The effect of changing the backhaul bandwidth is also studied. For both cases of $B_{\rm{b}}=B_{\rm{a}}$ and $B_{\rm{b}}=3B_{\rm{a}}$, by increasing $N_{\rm{T}}$, performance gains of UBS-OPT and CBS-OPT with respect to UBS-EQL and CBS-EQL grow. In the case of $B_{\rm{b}}=B_{\rm{a}}$, backhaul is the main bottleneck of the end-to-end performance when deploying super cells with $N_{\rm{T}}\geq3$. In this case, both optimal algorithms fully exploit the limited capacity of the bottleneck backhaul link as Fig.~\ref{fig_4_4_3a} shows. Increasing the bandwidth to $B_{\rm{b}}=3B_{\rm{a}}$ provides adequate backhaul capacity and thus the access system becomes the major bottleneck. Again, the optimal \ac{UBS} and optimal \ac{CBS} exhibit a superior performance by achieving the maximum rate limit of the network. For the same set of parameters, the average sum rate is plotted in Fig.~\ref{fig_4_4_3b} against the backhaul bandwidth normalized by the bandwidth of the access system, $\frac{B_{\rm{b}}}{B_{\rm{a}}}$.

\section{Opportunistic Power Control}\label{sec_4_3}
The optical power of backhaul \acp{LED} is opportunistically reduced with an incentive to enhance the \ac{PE} of the backhaul system while maintaining the sum rate performance. A \ac{FPC} strategy is proposed, whereby the transmission power in each backhaul branch is set to a constant operating point. This is a onetime design strategy, meaning that once the set point is chosen, it remains the same for the entire backhaul branch. This greatly simplifies the implementation complexity when applying \ac{FPC} to multi-tier super cells. However, an improperly low value of power can lead to a significant degradation in the network sum rate because of its impact on the capacity of the backhaul system. To reach a practical means to fix the backhaul power, three main schemes are put forward: \ac{MSPC}, \ac{ASPC} and \ac{ARPC}. The performance of a given branch of the super cell depends on the overall rate of the corresponding bottleneck backhaul link. To prevent a backhaul bottleneck for the $k$th branch $\forall k\in\mathcal{T}_1$, the following condition needs to be satisfied:
\begin{equation}\label{eq_4_3_6}
\mathcal{R}_{{\rm{b}}_k} \geq \sum_{i\in\mathcal{L}_k}\mathcal{R}_{{\rm{a}}_i}.
\end{equation}
The following analysis focuses on the design of the backhaul power control coefficient $K_{\rm{b}}$ based on the rate requirement of the bottleneck link\footnote{For the $k$th branch of the backhaul network, a feasible set is defined by $\mathcal{R}_{{\rm{b}}_i} \geq \sum_{j\in\mathcal{L}_i}\mathcal{R}_{{\rm{a}}_j}$, through the system of $N_{\rm{BS}}$ inequalities for all $\text{BS}_i$ $\forall i\in\mathcal{L}_k$. Fulfilling the rate requirement of the bottleneck link ${\rm{b}}_k$ by \eqref{eq_4_3_6} automatically guarantees validating the remaining inequalities for higher tiers.}. The minimum value of $K_{\rm{b}}$ is denoted by $K_{\rm{b,min}}$.

\subsection{Proposed Schemes}\label{sec_4_3_1}
\subsubsection{\ac{MSPC}}
The first criterion is to adjust the backhaul power in response to the maximum sum rate of the access system. The bounds of the access sum rate are related to those of the access \ac{SINR} by noting that $\mathcal{R}_{{\rm{a}}_i}=\frac{1}{M_i}\sum_{u\in\mathcal{U}_i}\mathcal{R}_{\rm{a}}(\gamma_u)$ based on \eqref{eq_3_2_12a}, where $\mathcal{R}_{\rm{a}}(\gamma_u)=\xi_{\rm{a}}B_{\rm{a}}\log_2(1+\gamma_u)$ are $M_i$ \ac{i.i.d.} random variables. By using \eqref{eq_4_1_5}, it follows that $\mathcal{R}_{\rm{min}}\leq\mathcal{R}_{\rm{a}}(\gamma_u)\leq\mathcal{R}_{\rm{max}}$ where $\mathcal{R}_{\rm{min}}=\xi_{\rm{a}}B_{\rm{a}}\log_2(1+\gamma_{\rm{min}})$ and $\mathcal{R}_{\rm{max}}=\xi_{\rm{a}}B_{\rm{a}}\log_2(1+\gamma_{\rm{max}})$ in which $\gamma_{\rm{min}}$ and $\gamma_{\rm{max}}$ are available in \eqref{eq_4_1_6}. Hence, $\mathcal{R}_{{\rm{a}}_i}$ is a bounded random variable such that:
\begin{equation}\label{eq_Proof_Lem_4_5_1}
\frac{1}{M_i}\sum_{u\in\mathcal{U}_i}\mathcal{R}_{\rm{min}}\leq\mathcal{R}_{{\rm{a}}_i}\leq\frac{1}{M_i}\sum_{u\in\mathcal{U}_i}\mathcal{R}_{\rm{max}},
\end{equation}
which then results in:
\begin{equation}\label{eq_Lem_4_5_1}
\mathcal{R}_{\rm{min}}\leq\mathcal{R}_{{\rm{a}}_i}\leq\mathcal{R}_{\rm{max}},
\end{equation}
since $|\mathcal{U}_i|=M_i$. The associated \ac{MSPC} ratio is derived in Proposition~\ref{Proposition_4_3}.
\begin{proposition}\label{Proposition_4_3}
	The minimum power control coefficient for $\text{b}_k$ based on \ac{MSPC} is given by:
	\begin{equation}\label{eq_Prop_4_3_1}
	K_{\rm{b,min}} = \frac{\left(1+\gamma_{\rm{max}}\right)^{\zeta^{-1}N_{\rm{BS}}}-1}{\gamma_{\rm{b}}}.
	\end{equation}
\end{proposition}
\begin{proof}
	On the \ac{RHS} of \eqref{eq_4_3_6}, $\mathcal{R}_{{\rm{a}}_i}$ is replaced by its upper limit from \eqref{eq_Lem_4_5_1}:
	\begin{equation}\label{eq_Proof_Prop_4_3_1}
	\xi_{\rm{b}}B_{\rm{b}}\log_2(1+K_{\rm{b}}\gamma_{\rm{b}}) \geq \sum_{i\in\mathcal{L}_k}\mathcal{R}_{\rm{max}} = N_{\rm{BS}}\xi_{\rm{a}}B_{\rm{a}}\log_2(1+\gamma_{\rm{max}}).
	\end{equation}
	Note that $|\mathcal{L}_k|=N_{\rm{BS}}$ $\forall k\in\mathcal{T}_1$. Expressing the inequality in \eqref{eq_Proof_Prop_4_3_1} in terms of $K_{\rm{b}}$ gives rise to:
	\begin{equation}\label{eq_Proof_Prop_4_3_3}
	K_{\rm{b}} \geq \frac{\left(1+\gamma_{\rm{max}}\right)^{\zeta^{-1}N_{\rm{BS}}}-1}{\gamma_{\rm{b}}}.
	\end{equation}
	The minimum value of $K_{\rm{b}}$ is readily given by the \ac{RHS} of \eqref{eq_Proof_Prop_4_3_3}, which is the desired result.
\end{proof}

\subsubsection{\ac{ASPC}}
The second criterion is to allocate power to the backhaul system so as to satisfy the achievable rate corresponding to the statistical average of the downlink \ac{SINR} over the area covered by each attocell. The average \ac{SINR} of the access system is given by Lemma~\ref{Lemma_4_2}. The \ac{ASPC} ratio is then derived in Proposition~\ref{Proposition_4_2}.

\begin{lemma}\label{Lemma_4_2}
	The average downlink \ac{SINR} is calculated by:
	\begin{equation}\label{eq_Lem_4_2_1}
	\bar{\gamma}_{\rm{a}} = \frac{\gamma_{\rm{min}}+\gamma_{\rm{max}}}{2}+
	\frac{2}{\pi R_{\rm{e}}^2}\int\limits_{\gamma_{\rm{min}}}^{\gamma_{\rm{max}}}\int\limits_{0}^{R_{\rm{e}}}\arcsin^\dagger\left(\mathcal{Z}(r,\gamma)\right)rdrd\gamma.
	\end{equation}
\end{lemma}
\begin{proof}
	Note that different \acp{UE} have the same average rate since $\gamma_u$ $\forall u$ are \ac{i.i.d.}. The expected value of a bounded random variable $x_{\rm{min}}\leq X\leq x_{\rm{min}}$ is given by:
	\begin{equation}\label{eq_Proof_Lem_4_2_3}
	\mathbb{E}[X]=x_{\rm{min}}+\int_{x_{\rm{min}}}^{x_{\rm{max}}}\mathbb{P}[X>x]dx.
	\end{equation}
	The average downlink \ac{SINR} is derived as:
	\begin{equation}\label{eq_Proof_Lem_4_2_2}
	\bar{\gamma}_{\rm{a}} = \gamma_{\rm{min}}+ \underbrace{\int\limits_{\gamma_{\rm{min}}}^{\gamma_{\rm{max}}}\mathbb{P}\left[\gamma_u>x\right]dx}_{I_1}.
	\end{equation}
	By using the \ac{CDF} of $\gamma_u$ in \eqref{eq_4_1_2}, $I_1$ is evaluated as follows:
	\begin{equation}\label{eq_Proof_Lem_4_2_1}
	I_1 = \int\limits_{\gamma_{\rm{min}}}^{\gamma_{\rm{max}}}\left(1-\mathbb{P}\left[\gamma_u\leq\gamma\right]\right)d\gamma = \frac{\gamma_{\rm{max}}-\gamma_{\rm{min}}}{2}+
	\frac{2}{\pi R_{\rm{e}}^2}\int\limits_{\gamma_{\rm{min}}}^{\gamma_{\rm{max}}}\int\limits_{0}^{R_{\rm{e}}}\arcsin^\dagger\left(\mathcal{Z}(r,\gamma)\right)rdrd\gamma.
	\end{equation}
	Substituting $I_1$ in \eqref{eq_Proof_Lem_4_2_2} with \eqref{eq_Proof_Lem_4_2_1} results in \eqref{eq_Lem_4_2_1}.
\end{proof}
\begin{proposition}\label{Proposition_4_2}
	The minimum power control coefficient for $\text{b}_k$ based on \ac{ASPC} is given by:
	\begin{equation}\label{eq_Prop_4_2_1}
	K_{\rm{b,min}} = \frac{\left(1+\bar{\gamma}_{\rm{a}}\right)^{\zeta^{-1}N_{\rm{BS}}}-1}{\gamma_{\rm{b}}},
	\end{equation}
	where $\bar{\gamma}_{\rm{a}}$ is the average downlink \ac{SINR} given by Lemma~\ref{Lemma_4_2}.
\end{proposition}
\begin{proof}\label{Proof_Prop_4_2}
	In the case of \ac{ASPC}, the inequality in \eqref{eq_4_3_6} changes to:
	\begin{equation}\label{eq_Proof_Prop_4_2_1}
	\xi_{\rm{b}}B_{\rm{b}}\log_2(1+K_{\rm{b}}\gamma_{\rm{b}}) \geq \sum_{i\in\mathcal{L}_k}\xi_{\rm{a}}B_{\rm{a}}\log_2(1+\mathbb{E}[\gamma_u]),
	\end{equation}
	where $\mathbb{E}[\gamma_u]=\bar{\gamma}_{\rm{a}}$. It immediately follows that: 
	\begin{equation}\label{eq_Proof_Prop_4_2_2}
	K_{\rm{b}} \geq \frac{\left(1+\bar{\gamma}_{\rm{a}}\right)^{\zeta^{-1}N_{\rm{BS}}}-1}{\gamma_{\rm{b}}}.
	\end{equation}
	The \ac{RHS} of \eqref{eq_Proof_Prop_4_2_2} is, in fact, the minimum value that $K_{\rm{b}}$ can take and this concludes the proof.
\end{proof}

\subsubsection{\ac{ARPC}}
The third criterion for assigning power to the backhaul system takes into account the statistical average of the achievable rate for the access system over the area covered by each attocell. The average data rate of the access system is provided in Lemma~\ref{Lemma_4_1}. The \ac{ARPC} ratio is subsequently derived in Proposition~\ref{Proposition_4_1}.

\begin{lemma}\label{Lemma_4_1}
	The average achievable rate of the access system per attocell is calculated by:
	\begin{equation}\label{eq_Lem_4_1_1}
	\bar{\mathcal{R}}_{\rm{a}} = \frac{\mathcal{R}_{\rm{min}}+\mathcal{R}_{\rm{max}}}{2}+\frac{2\xi_{\rm{a}}B_{\rm{a}}}{\pi R_{\rm{e}}^2\ln2}\int\limits_{\gamma_{\rm{min}}}^{\gamma_{\rm{max}}}\int\limits_{0}^{R_{\rm{e}}}\frac{\arcsin^\dagger\left(\mathcal{Z}(r,\gamma)\right)r}{1+\gamma}drd\gamma.
	\end{equation}
\end{lemma}
\begin{proof}
	By using \eqref{eq_3_2_12a}, the average access system rate for $\text{BS}_i$ is obtained as:
	\begin{equation}\label{eq_Proof_Lem_4_1_1}
		\mathbb{E}\left[\mathcal{R}_{{\rm{a}}_i}\right] = \xi_{\rm{a}}B_{\rm{a}}\,\mathbb{E}\left[\log_2(1+\gamma_u)\right].
	\end{equation}
	Note that $\gamma_u$ $\forall u\in\mathcal{U}_i$ are \ac{i.i.d.}, thus $\mathbb{E}\left[\mathcal{R}_{{\rm{a}}_i}\right]=\bar{\mathcal{R}}_{\rm{a}}$ $\forall i$. Based on \eqref{eq_4_1_2} and \eqref{eq_Proof_Lem_4_2_3}, the expectation in \eqref{eq_Proof_Lem_4_1_1} is therefore expanded as follows:
	\begin{equation}\label{eq_Proof_Lem_4_1_4}
	\bar{\mathcal{R}}_{\rm{a}} = \mathcal{R}_{\rm{min}}+ \underbrace{\int\limits_{\mathcal{R}_{\rm{min}}}^{\mathcal{R}_{\rm{max}}}\mathbb{P}\left[\xi_{\rm{a}}B_{\rm{a}}\log_2(1+\gamma_u)>x\right]dx}_{I_2},
	\end{equation}
	where:
	\begin{subequations}\label{eq_Proof_Lem_4_1_2}
		\begin{align}
		I_2 &= \frac{\xi_{\rm{a}}B_{\rm{a}}}{\ln2}\int\limits_{\gamma_{\rm{min}}}^{\gamma_{\rm{max}}}\left(1-\mathbb{P}\left[\gamma_u\leq\gamma\right]\right)\frac{d\gamma}{1+\gamma},\label{eq_Proof_Lem_4_1_2c}\\
		&= \frac{\mathcal{R}_{\rm{max}}-\mathcal{R}_{\rm{min}}}{2}+\frac{2\xi_{\rm{a}}B_{\rm{a}}}{\pi R_{\rm{e}}^2\ln2}\int\limits_{\gamma_{\rm{min}}}^{\gamma_{\rm{max}}}\int\limits_{0}^{R_{\rm{e}}}\frac{\arcsin^\dagger\left(\mathcal{Z}(r,\gamma)\right)r}{1+\gamma}drd\gamma,\label{eq_Proof_Lem_4_1_2e}
		\end{align}
	\end{subequations}
	The substitution $x=\xi_{\rm{a}}B_{\rm{a}}\log_2(1+\gamma)$ is used to arrive at \eqref{eq_Proof_Lem_4_1_2c}, which does not alter the inequality under a probability measure as the logarithm is a monotonically increasing function. Replacing $I_2$ in \eqref{eq_Proof_Lem_4_1_4} by \eqref{eq_Proof_Lem_4_1_2e} and simplifying leads to \eqref{eq_Lem_4_1_1}.
\end{proof}

\begin{proposition}\label{Proposition_4_1}
	The minimum power control coefficient for $\text{b}_k$ based on \ac{ARPC} is given by:
	\begin{equation}\label{eq_Prop_4_1_1}
	K_{\rm{b,min}} = \frac{\exp\left(\frac{\ln2}{\xi_{\rm{b}}B_{\rm{b}}}N_{\rm{BS}}\bar{\mathcal{R}}_{\rm{a}}\right)-1}{\gamma_{\rm{b}}},
	\end{equation}
	where $\bar{\mathcal{R}}_{\rm{a}}$ is the average achievable rate over an attocell, given by Lemma~\ref{Lemma_4_1}.
\end{proposition}
\begin{proof}\label{Proof_Prop_4_1}
	According to \ac{ARPC}, the \ac{RHS} of \eqref{eq_4_3_6} needs to be modified as follows:
	\begin{equation}\label{eq_Proof_Prop_4_1_1}
	\xi_{\rm{b}}B_{\rm{b}}\log_2(1+K_{\rm{b}}\gamma_{\rm{b}}) \geq \mathbb{E}\left[\sum_{i\in\mathcal{L}_k}\mathcal{R}_{{\rm{a}}_i}\right]=N_{\rm{BS}}\bar{\mathcal{R}}_{\rm{a}}.
	\end{equation}
	Rearranging the inequality in terms of $K_{\rm{b}}$ gives:
	\begin{equation}\label{eq_Proof_Prop_4_1_4}
	K_{\rm{b}} \geq \frac{\exp\left(\frac{\ln2}{\xi_{\rm{b}}B_{\rm{b}}}N_{\rm{BS}}\bar{\mathcal{R}}_{\rm{a}}\right)-1}{\gamma_{\rm{b}}}.
	\end{equation}
	The \ac{RHS} of \eqref{eq_Proof_Prop_4_1_4} represents the minimum allowed value of $K_{\rm{b}}$ and hence the proof is complete.
\end{proof}

\subsection{Probability of Backhaul Bottleneck Occurrence}\label{sec_4_3_2}
To gain insight into the power control performance, a metric called \ac{BBO} is defined as follows.
\begin{definition}\label{Definition_4_1}
	\ac{BBO} is a metric to measure the probability that the aggregate sum rate of the access system in a backhaul branch exceeds the capacity of the corresponding bottleneck link. Equivalently, it evaluates the probability that the condition in \eqref{eq_4_3_6} is violated.
\end{definition}
\noindent Mathematically, the \ac{BBO} probability for the $k$th branch $k\in\mathcal{T}_1$, is expressed by:
\begin{equation}\label{eq_4_3_3}
P_{\rm{BBO}} = \mathbb{P}\left[\sum_{i\in\mathcal{L}_k}\mathcal{R}_{{\rm{a}}_i}>\mathcal{R}_{{\rm{b}}_k}\right],
\end{equation}
where $\mathcal{R}_{{\rm{a}}_i}$ is a random variable that depends on the statistics of $\gamma_u$. There is no exact closed form solution for \eqref{eq_4_3_3} in terms of ordinary functions. Alternatively, a simple but tight analytical approximation is established in Theorem~\ref{Theorem_4_1} with the aid of Lemma~\ref{Lemma_4_3}. Note that $\mathcal{R}_{{\rm{a}}_i}=\frac{1}{M_i}\sum_{u\in\mathcal{U}_i}\mathcal{R}_{\rm{a}}(\gamma_u)$ where $\mathcal{R}_{\rm{a}}(\gamma_u)=\xi_{\rm{a}}B_{\rm{a}}\log_2(1+\gamma_u)$ are \ac{i.i.d.}. The mean of $\mathcal{R}_{\rm{a}}(\gamma_u)$ is readily given by Lemma~\ref{Lemma_4_1}. The variance of $\mathcal{R}_{\rm{a}}(\gamma_u)$ is determined in Lemma~\ref{Lemma_4_3}.
\begin{lemma}\label{Lemma_4_3}
	The variance of $\mathcal{R}_{\rm{a}}(\gamma_u)$ is given by:
	\begin{equation}\label{eq_Lem_4_3_1}
	\sigma_{\mathcal{R}_{\rm{a}}}^2=\mathbb{E}\!\left[\mathcal{R}_{\rm{a}}^2(\gamma_u)\right]-\mathbb{E}^2[\mathcal{R}_{\rm{a}}(\gamma_u)],
	\end{equation}
	where $\mathbb{E}\left[\mathcal{R}_{\rm{a}}(\gamma_u)\right]=\bar{\mathcal{R}}_{\rm{a}}$ and:
	\begin{equation}\label{eq_Lem_4_3_2}
	\mathbb{E}\!\left[\mathcal{R}_{\rm{a}}^2(\gamma_u)\right] = \frac{\mathcal{R}_{\rm{min}}^2+\mathcal{R}_{\rm{max}}^2}{2}+\frac{1}{\pi}\left(\frac{2\xi_{\rm{a}}B_{\rm{a}}}{R_{\rm{e}}\ln2}\right)^{\!\!2}\int\limits_{\gamma_{\rm{min}}}^{\gamma_{\rm{max}}}\int\limits_{0}^{R_{\rm{e}}}\frac{\ln(1+\gamma)}{1+\gamma}\arcsin^\dagger\left(\mathcal{Z}(r,\gamma)\right)rdrd\gamma.
	\end{equation}
\end{lemma}
\begin{proof}
	The second order moment of a bounded random variable $x_{\rm{min}}\leq X\leq x_{\rm{min}}$ is characterized by using $\mathbb{E}\!\left[X^2\right] = x_{\rm{min}}^2+\int_{x_{\rm{min}}}^{x_{\rm{max}}}2x\mathbb{P}[X>x]dx\nonumber$. Therefore:
	\begin{equation}\label{eq_Proof_Lem_4_3_4}
		\mathbb{E}\!\left[\mathcal{R}_{\rm{a}}^2(\gamma_u)\right] = \mathcal{R}_{\rm{min}}^2+\underbrace{\int\limits_{\mathcal{R}_{\rm{min}}}^{\mathcal{R}_{\rm{max}}}2x\mathbb{P}\left[\mathcal{R}_{\rm{a}}(\gamma_u)>x\right]dx}_{I_3}.
	\end{equation}
	Referring to the \ac{CDF} of $\gamma_u$ in \eqref{eq_4_1_2}, $I_3$ is derived as follows:
	\begin{subequations}\label{eq_Proof_Lem_4_3_1}
		\begin{align}
		I_3 &= 2\left(\frac{\xi_{\rm{a}}B_{\rm{a}}}{\ln2}\right)^{\!\!2}\int\limits_{\gamma_{\rm{min}}}^{\gamma_{\rm{min}}}\frac{\ln(1+\gamma)}{1+\gamma}\left(1-\mathbb{P}\left[\gamma_u\leq\gamma\right]\right)d\gamma,\label{eq_Proof_Lem_4_3_1c}\\
		&= \frac{\mathcal{R}_{\rm{max}}^2-\mathcal{R}_{\rm{min}}^2}{2}+\frac{1}{\pi}\left(\frac{2\xi_{\rm{a}}B_{\rm{a}}}{R_{\rm{e}}\ln2}\right)^{\!\!2}\int\limits_{\gamma_{\rm{min}}}^{\gamma_{\rm{max}}}\int\limits_{0}^{R_{\rm{e}}}\frac{\ln(1+\gamma)}{1+\gamma}\arcsin^\dagger\left(\mathcal{Z}(r,\gamma)\right)rdrd\gamma.\label{eq_Proof_Lem_4_3_1e}
		\end{align}
	\end{subequations}
	By substituting \eqref{eq_Proof_Lem_4_3_1e} for $I_3$ in \eqref{eq_Proof_Lem_4_3_4}, the desired result of \eqref{eq_Lem_4_3_2} is deduced.
\end{proof}

\begin{theorem}\label{Theorem_4_1}
	For the $k$th backhaul branch with $M$ \acp{UE} over the total area covered by $N_{\rm{BS}}$ \acp{BS}, the \ac{BBO} probability is tightly approximated by:
	\begin{equation}\label{eq_Theo_4_1_1}
	P_{\rm{BBO}} \approx \sum_{n=1}^{N_{\rm{BS}}}p_n{\rm{Q}}\left(\frac{\mathcal{R}_{{\rm{b}}_k}-n\bar{\mathcal{R}}_{\rm{a}}}{\frac{n}{\sqrt{M}}\sigma_{\mathcal{R}_{\rm{a}}}}\right),
	\end{equation}
	where:
	\begin{equation}\label{eq_Theo_4_3_2}
	p_n = \binom{N_{\rm{BS}}}{n}\sum_{l=0}^{n}(-1)^l\binom{n}{l}\left(\frac{n-l}{N_{\rm{BS}}}\right)^{\!M}.
	\end{equation}
	Also, $\mathcal{R}_{{\rm{b}}_k}$ and $\bar{\mathcal{R}}_{\rm{a}}$ are given by \eqref{eq_3_2_12b} and Lemma~\ref{Lemma_4_1}, respectively; and $\sigma_{\mathcal{R}_{\rm{a}}}$ is the standard deviation of $\mathcal{R}_{\rm{a}}(\gamma_u)$ whose variance is identified in Lemma~\ref{Lemma_4_3}.
\end{theorem}
\begin{proof}\label{Proof_Prop_4_3}
	Let the vector $\mathbf{M}=[M_i]_{N_{\rm{BS}}\times1}$ be composed of the random numbers of \acp{UE} in individual attocells for the $k$th branch. Provided that the total number of \acp{UE} is fixed at $\sum_{i\in\mathcal{L}_k}M_i=M$, $\mathbf{M}$ follows a multinomial distribution. The \ac{BBO} probability in \eqref{eq_4_3_3} is expressed as follows:
	\begin{equation}\label{eq_Proof_Theo_4_3_2}
	P_{\rm{BBO}} = \mathbb{P}\left[\sum_{i\in\mathcal{L}_k}\frac{1}{M_i}\sum_{u\in\mathcal{U}_i}\mathcal{R}_{\rm{a}}(\gamma_u)>\mathcal{R}_{{\rm{b}}_k}\right].
	\end{equation}
	The argument of the probability in \eqref{eq_Proof_Theo_4_3_2} involves positive weights encompassing the reciprocals of the numbers of \acp{UE} in every attocell. An appropriate approximation of this weighted sum can be derived by means of minimizing the \ac{MSE}. This is presented in Lemma~\ref{Lemma_4_4}.
	\begin{lemma}\label{Lemma_4_4}
		Based on the \ac{MMSE} criterion, the summation under the probability in \eqref{eq_Proof_Theo_4_3_2} is approximated as follows:
		\begin{equation}\label{eq_Lem_4_4_1}
		\sum_{i\in\mathcal{L}_k}\frac{1}{M_i}\sum_{u\in\mathcal{U}_i}\mathcal{R}_{\rm{a}}(\gamma_u) \approx \frac{n_{\rm{BS}}}{M}\sum_{i\in\mathcal{L}_k}\sum_{u\in\mathcal{U}_i}\mathcal{R}_{\rm{a}}(\gamma_u),
		\end{equation}
		where $n_{\rm{BS}}$ indicates the aggregate number of non-empty attocells corresponding to the random vector $\mathbf{M}$. The attocell of $\text{BS}_i$ is accounted non-empty if $M_i>0$.
	\end{lemma}
	\begin{proof}\label{Proof_Lem_4_4}
		See Appendix~\ref{ProofLemma4}.
	\end{proof}
	Let $Z=\frac{n_{\rm{BS}}}{M}\sum_{i\in\mathcal{L}_k}\sum_{u\in\mathcal{U}_i}\mathcal{R}_{\rm{a}}(\gamma_u)$. Note that $Z$ is not directly dependent on the exact number of \acp{UE} that each attocell involves, i.e. the elements of $\mathbf{M}$. Rather, it depends on the overall number of non-empty attocells, i.e. $n_{\rm{BS}}$. For each random experiment, $n_{\rm{BS}}$ takes integer values from $1$ to $N_{\rm{BS}}$. Besides, $\sum_{i\in\mathcal{L}_k}\sum_{u\in\mathcal{U}_i}\mathcal{R}_{\rm{a}}(\gamma_u)$ is a sum of $M$ \ac{i.i.d.} random variables $\mathcal{R}_{\rm{a}}(\gamma_u)$, the mean and variance of which are known according to Lemma~\ref{Lemma_4_1} and Lemma~\ref{Lemma_4_3}, respectively. Thus, for a sufficiently large value of $M$, the conditional distribution of $Z$ given $n_{\rm{BS}}=n$ converges to Gaussian based on the \ac{CLT} \cite{Vallentin}. It is deduced that:
	\begin{equation}\label{eq_Proof_Prop_4_3_5}
	Z\big|\{n_{\rm{BS}}=n\} \sim \mathcal{N}\left(n\bar{\mathcal{R}}_{\rm{a}},\frac{n^2}{M}\sigma_{\mathcal{R}_{\rm{a}}}^2\right).
	\end{equation}
	Therefore, by means of Lemma~\ref{Lemma_4_4}, the \ac{BBO} probability in \eqref{eq_Proof_Theo_4_3_2} can be evaluated by conditioning on $n_{\rm{BS}}$ and applying the law of total probability. Combining \eqref{eq_Proof_Prop_4_3_5} with \eqref{eq_Lem_4_4_1} and substituting the result into \eqref{eq_Proof_Theo_4_3_2} gives rise to:
	\begin{equation}\label{eq_Proof_Prop_4_3_6}
	P_{\rm{BBO}} \approx \sum_{n=1}^{N_{\rm{BS}}}\mathbb{P}\left[n_{\rm{BS}}=n\right]\mathbb{P}\left[Z>\mathcal{R}_{{\rm{b}}_k}\big|n_{\rm{BS}}=n\right] = \sum_{n=1}^{N_{\rm{BS}}}p_n{\rm{Q}}\left(\frac{\mathcal{R}_{{\rm{b}}_k}-n\bar{\mathcal{R}}_{\rm{a}}}{\frac{n}{\sqrt{M}}\sigma_{\mathcal{R}_{\rm{a}}}}\right),
	\end{equation}
	where $p_n=\mathbb{P}\left[n_{\rm{BS}}=n\right]$. From combinatorial analysis, the problem of distributing $M$ \acp{UE} into $N_{\rm{BS}}$ attocells refers to the classical occupancy problem with Boltzmann-Maxwell statistics \cite{Feller}. That is to say, there are $N_{\rm{BS}}^M$ permutations and each possible distribution has a probability of $\frac{1}{N_{\rm{BS}}^M}$\footnote{This is an immediate result of the uniform distribution of \acp{UE}.}. Besides, the outcome of the event $\{n_{\rm{BS}}=n\}$ corresponds to the case where exactly $n$ attocells each are occupied by at least one \ac{UE} and the other $N_{\rm{BS}}-n$ remain empty. Let $\{n_{\rm{BS}}'=n'\}$ be the event indicating that exactly $n'$ attocells are empty. The probability of this event is available in closed form \cite{Feller}:
	\begin{equation}\label{eq_Proof_Prop_4_3_7}
	\mathbb{P}\left[n_{\rm{BS}}'=n'\right] = \binom{N_{\rm{BS}}}{n'}\sum_{l=0}^{N_{\rm{BS}}-n'}(-1)^l\binom{N_{\rm{BS}}-n'}{l}\left(1-\frac{n'+l}{N_{\rm{BS}}}\right)^{\!M}.
	\end{equation}
	Upon substituting $n'=N_{\rm{BS}}-n$, \eqref{eq_Proof_Prop_4_3_7} reduces to the desired probability $p_n$ in \eqref{eq_Theo_4_3_2}.
\end{proof}

\subsection{Numerical Results and Discussions}\label{sec_4_3_3}
This section presents a number of case studies to evaluate the performance of the proposed power control schemes using computer simulations. The system parameters are given by Table~\ref{tbl_3_4_1}.

\subsubsection{Power Control Coefficients}
First, the range of variations of the power control coefficients is studied based on Propositions~\ref{Proposition_4_3}, \ref{Proposition_4_2} and \ref{Proposition_4_1} for \ac{MSPC}, \ac{ASPC} and \ac{ARPC}, respectively.

\begin{figure}[!t]
	\centering
	\subfloat[\label{fig_4_3_4a} Power control coefficient.]{\includegraphics[width=0.5\linewidth]{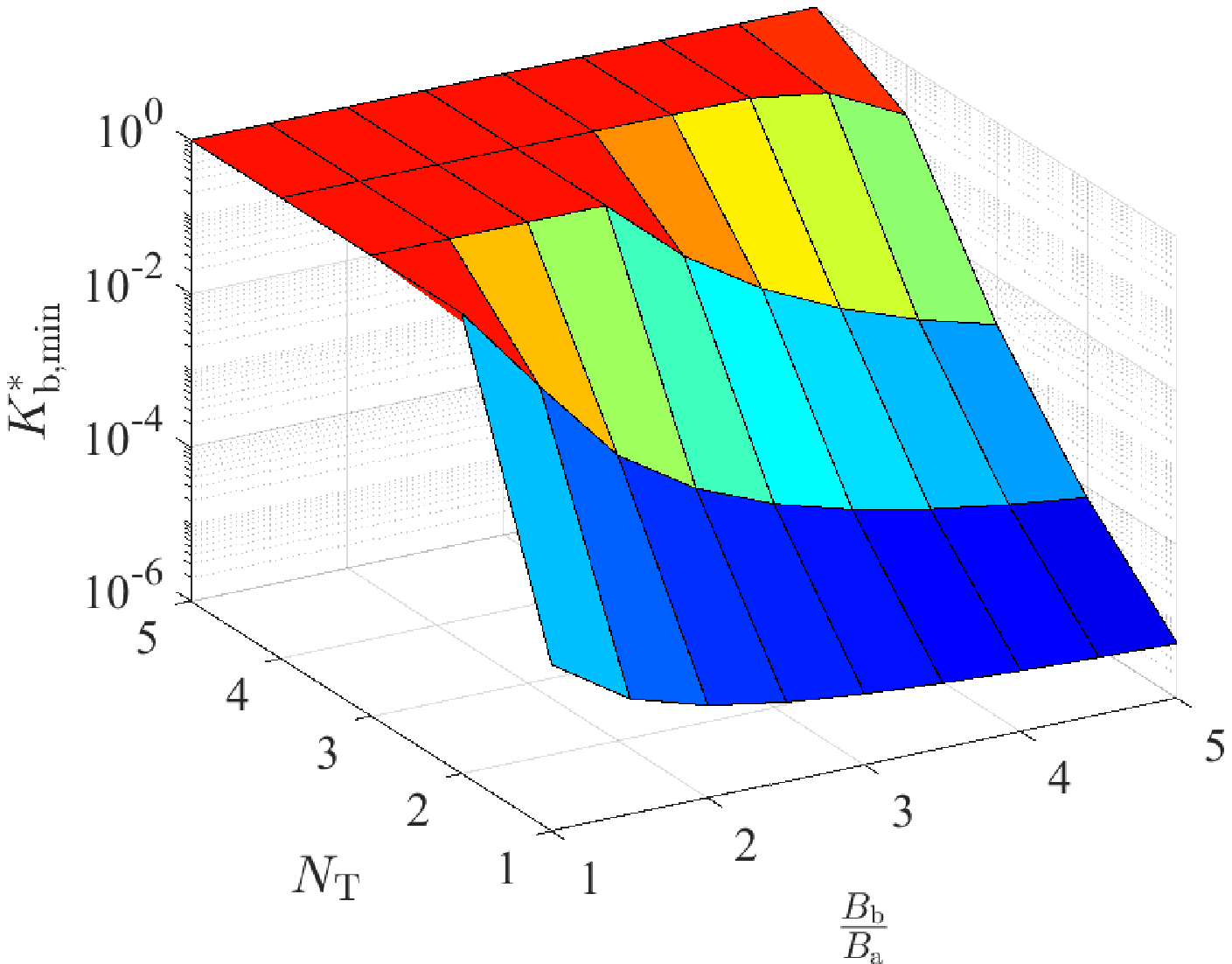}}
	\subfloat[\label{fig_4_3_4b} Backhaul rate.]{\includegraphics[width=0.5\linewidth]{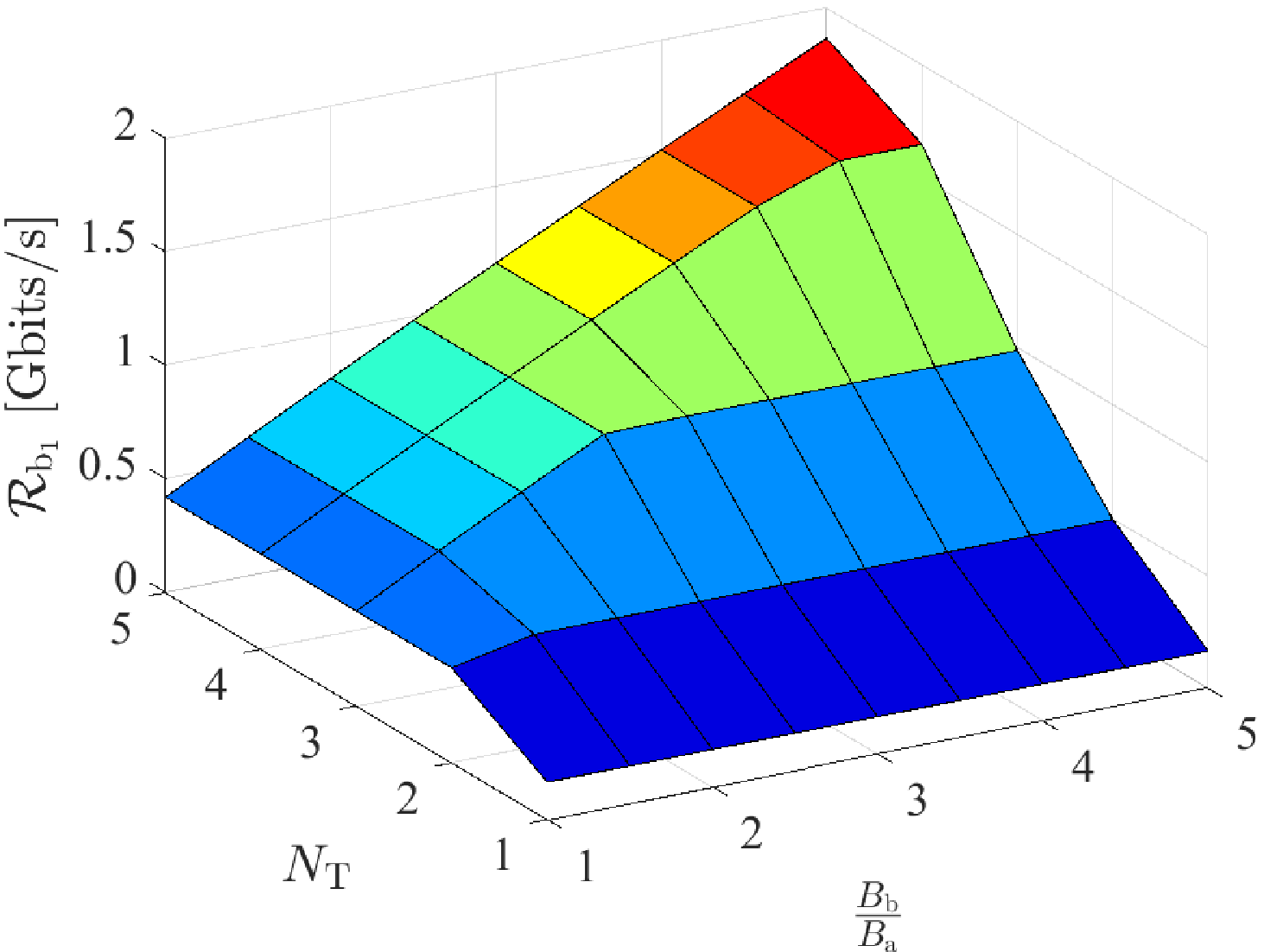}}
	\caption{$K_{\rm{b,min}}^*$ for MSPC and the backhaul rate $\mathcal{R}_{\rm{b}_1}|_{K_{\rm{b}}=K_{\rm{b,min}}^*}$ as a function of the total number of tiers $N_{\rm{T}}$ and the bandwidth ratio $\frac{B_{\rm{b}}}{B_{\rm{a}}}$.}
	\label{fig_4_3_4}
\end{figure}

\begin{figure}[!t]
	\centering
	\subfloat[\label{fig_4_3_5a} Power control coefficient.]{\includegraphics[width=0.5\linewidth]{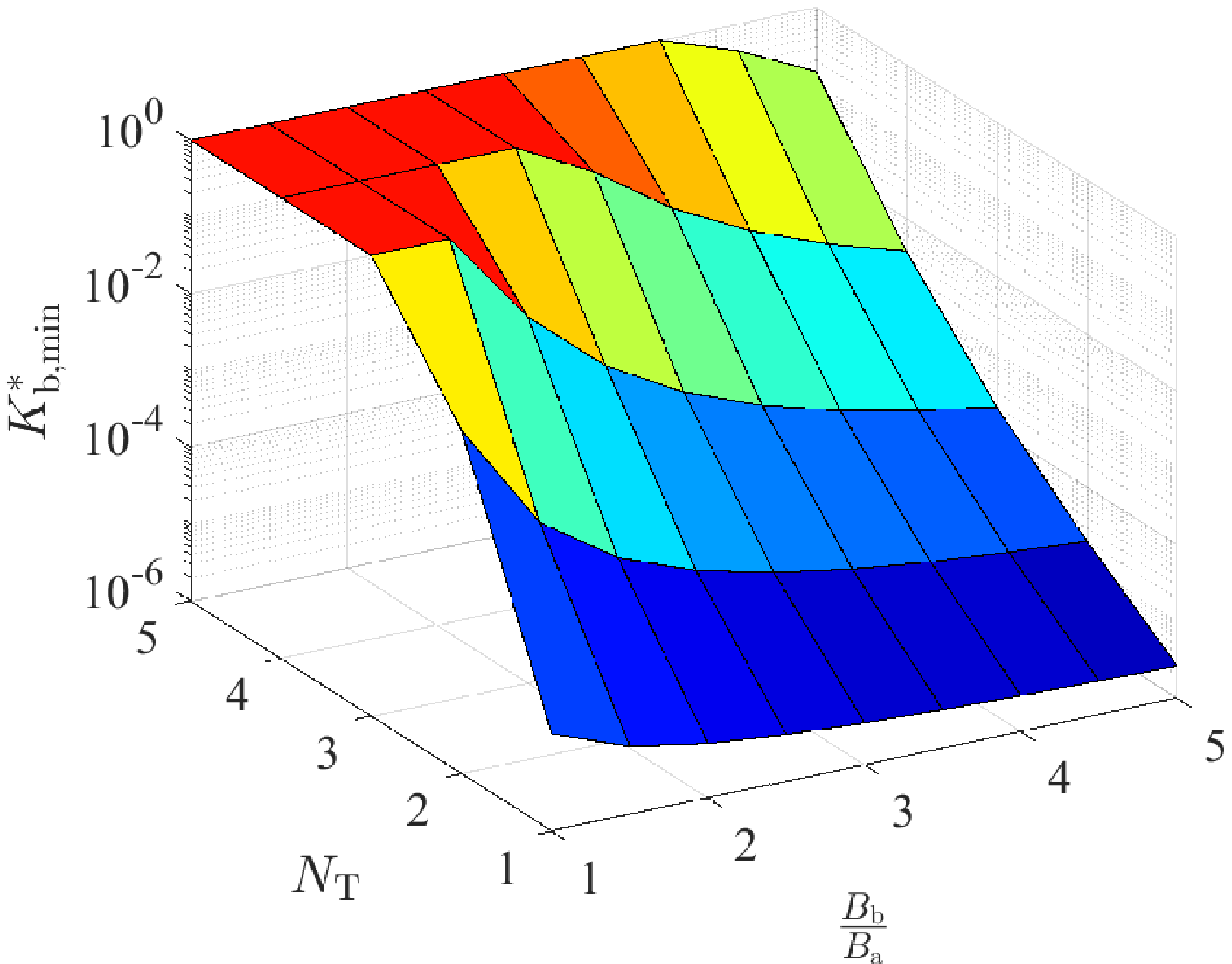}}
	\subfloat[\label{fig_4_3_5b} Backhaul rate.]{\includegraphics[width=0.5\linewidth]{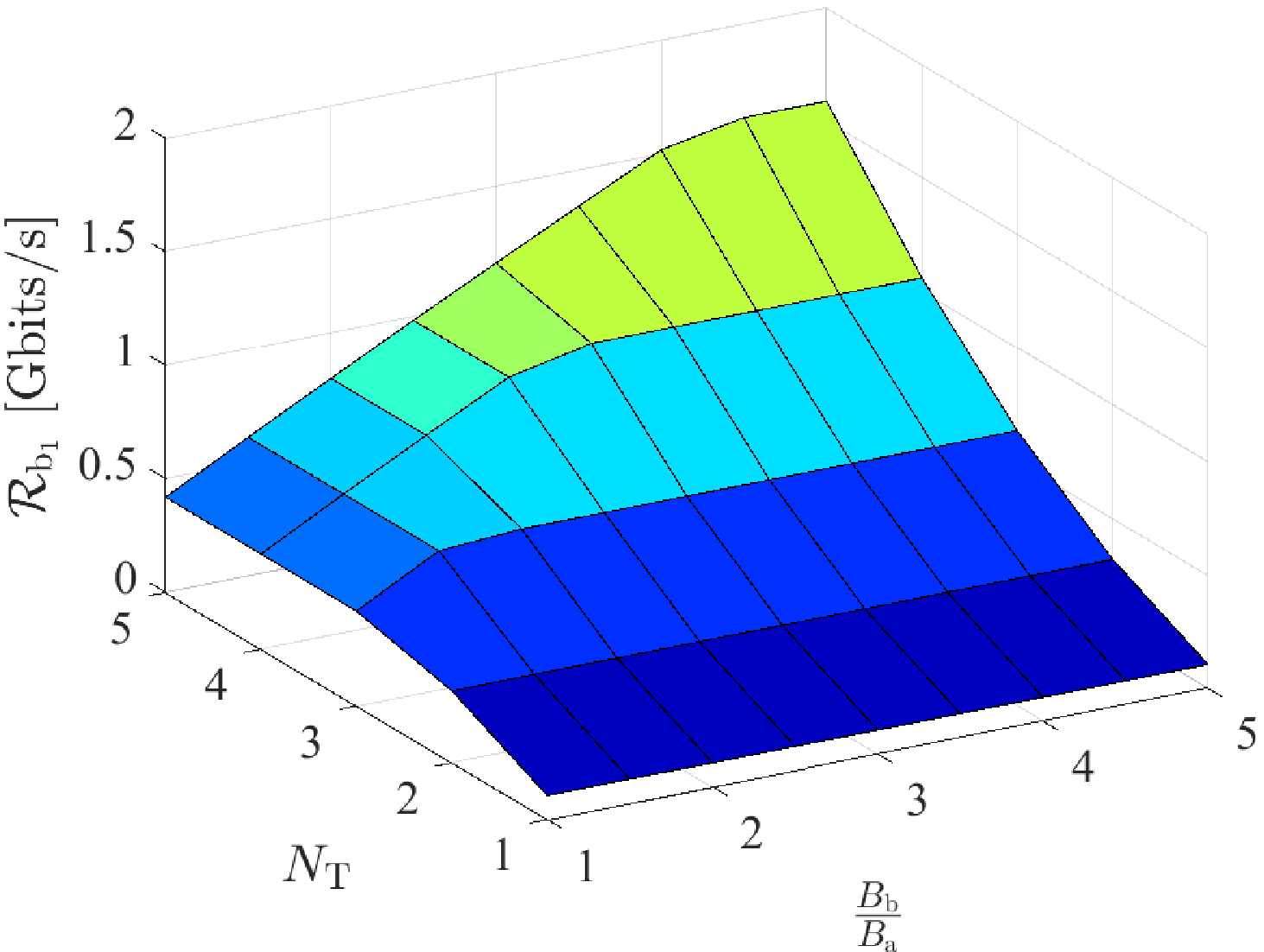}}
	\caption{$K_{\rm{b,min}}^*$ for ASPC and the backhaul rate $\mathcal{R}_{\rm{b}_1}|_{K_{\rm{b}}=K_{\rm{b,min}}^*}$ as a function of the total number of tiers $N_{\rm{T}}$ and the bandwidth ratio $\frac{B_{\rm{b}}}{B_{\rm{a}}}$.}
	\label{fig_4_3_5}
\end{figure}

\begin{figure}[!t]
	\centering
	\subfloat[\label{fig_4_3_6a} Power control coefficient.]{\includegraphics[width=0.5\linewidth]{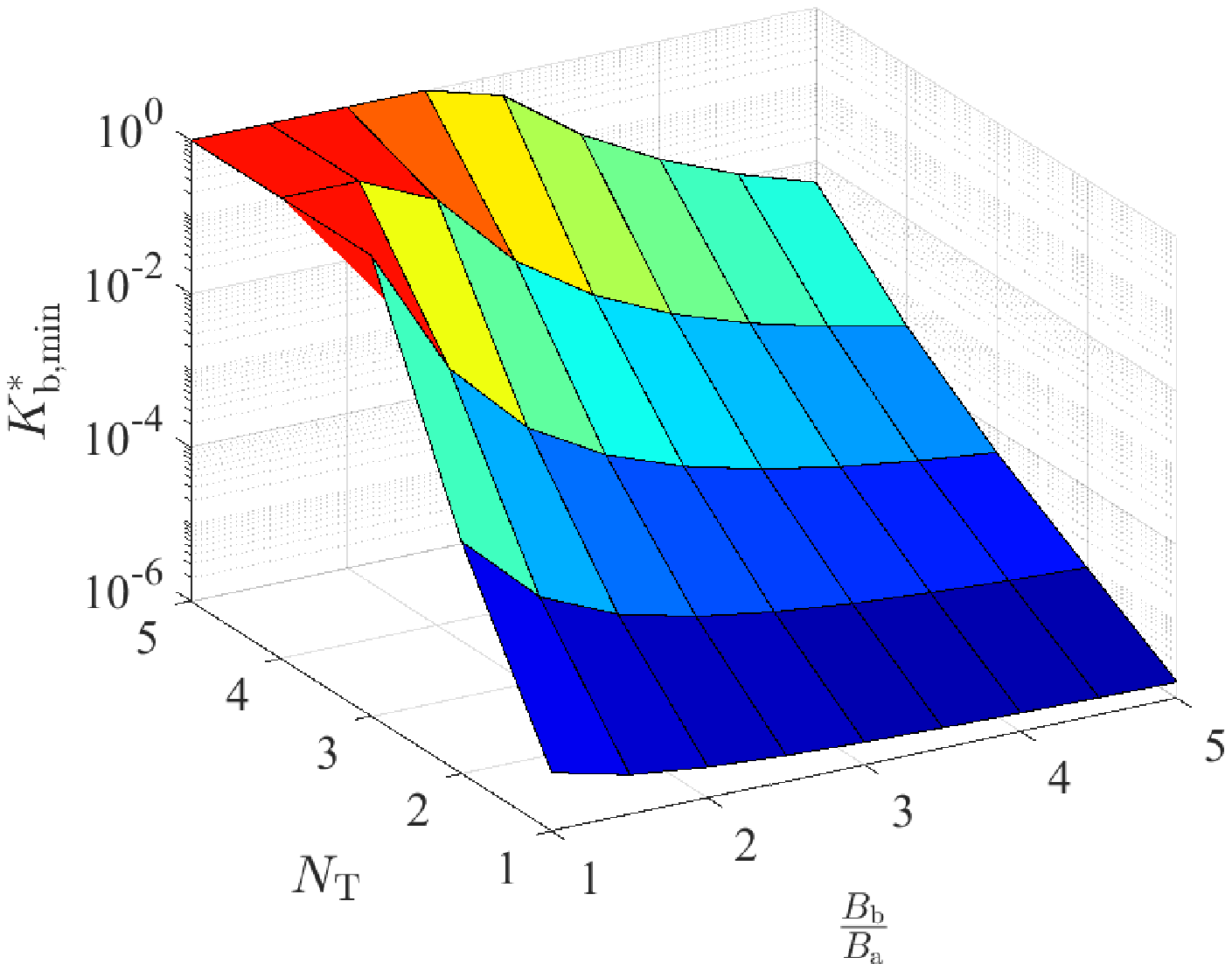}}
	\subfloat[\label{fig_4_3_6b} Backhaul rate.]{\includegraphics[width=0.5\linewidth]{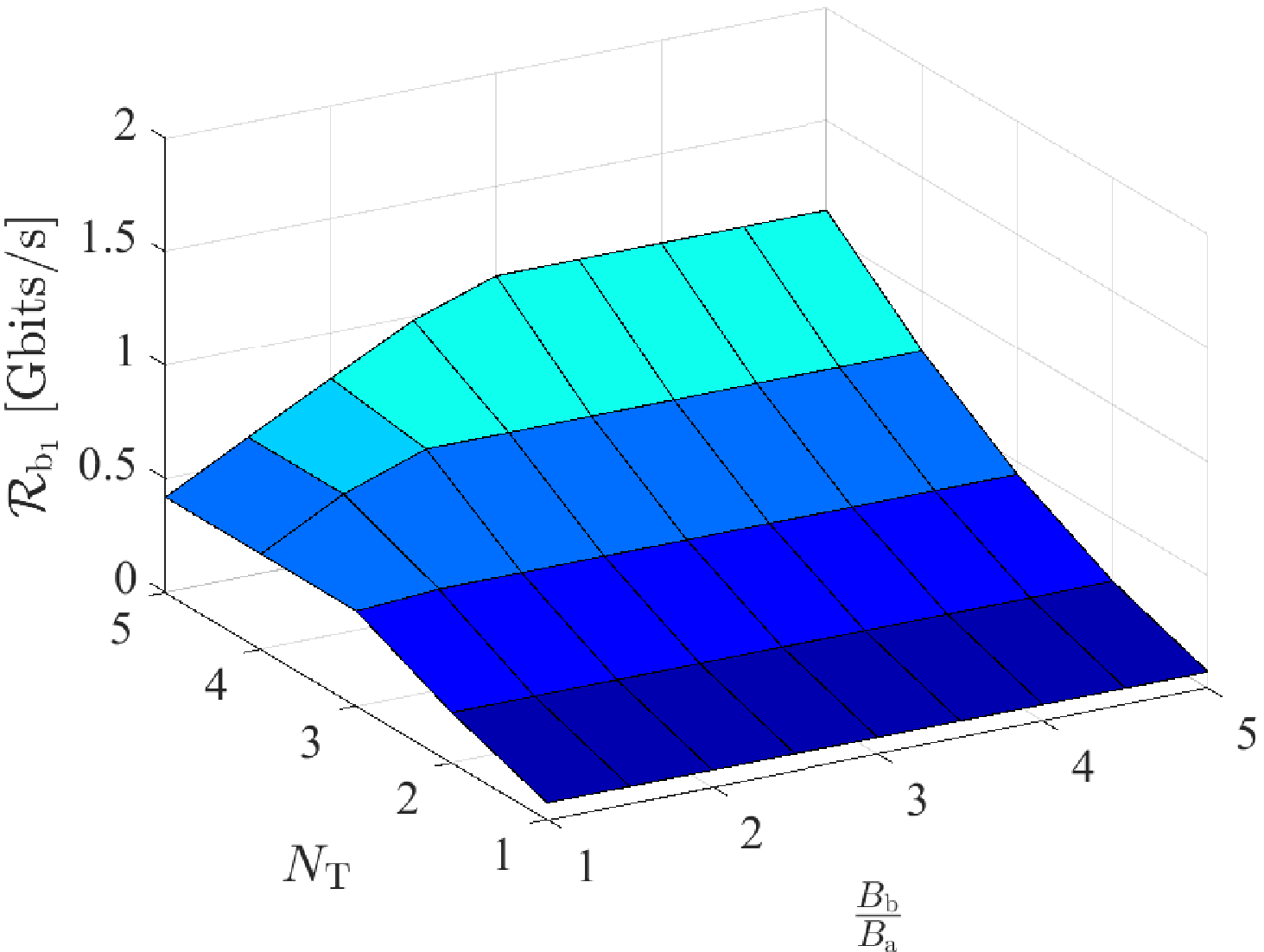}}
	\caption{$K_{\rm{b,min}}^*$ for ARPC and the backhaul rate $\mathcal{R}_{\rm{b}_1}|_{K_{\rm{b}}=K_{\rm{b,min}}^*}$ as a function of the total number of tiers $N_{\rm{T}}$ and the bandwidth ratio $\frac{B_{\rm{b}}}{B_{\rm{a}}}$.}
	\label{fig_4_3_6}
\end{figure}

Figs.~\ref{fig_4_3_4a}, \ref{fig_4_3_5a} and \ref{fig_4_3_6a} demonstrates the range of values of $K_{\rm{b,min}}$ for \ac{MSPC}, \ac{ASPC} and \ac{ARPC} schemes, respectively, as a function of $N_{\rm{T}}$ and the bandwidth ratio $\frac{B_{\rm{b}}}{B_{\rm{a}}}$. The resulting backhaul rate for each scheme is computed by $\mathcal{R}_{\rm{b}_1}|_{K_{\rm{b}}=K_{\rm{b,min}}^*}=\xi_{\rm{b}}B_{\rm{b}}\log_2(1+K_{\rm{b,min}}^*\gamma_{\rm{b}})$ and shown in Figs.~\ref{fig_4_3_4b}, \ref{fig_4_3_5b} and \ref{fig_4_3_6b}. It is observed that the power control coefficient is an increasing function of the total number of the deployed tiers for all three schemes, while it is a decreasing function of the normalized bandwidth. For given values of $N_{\rm{T}}$ and $\frac{B_{\rm{b}}}{B_{\rm{a}}}$, the highest value of $K_{\rm{b,min}}$ is set by \ac{MSPC}, the second highest by \ac{ASPC}, and the lowest by \ac{ARPC}, confirming that:
\begin{equation}\label{eq_4_3_4}
K_{\rm{b,min}}^{\rm{ARPC}}<K_{\rm{b,min}}^{\rm{ASPC}}<K_{\rm{b,min}}^{\rm{MSPC}}.
\end{equation}
The amount of power assigned to the backhaul system by the three schemes and the corresponding backhaul rates also obey the same rule in \eqref{eq_4_3_4}. For a fixed number of tiers, Figs.~\ref{fig_4_3_4a}, \ref{fig_4_3_5a} and \ref{fig_4_3_6a} show that by increasing the backhaul bandwidth, the level of $K_{\rm{b,min}}$ lessens for all the schemes altogether. Hence, more power needs to be allocated to the backhaul system when the bandwidth reduces. This conforms to the intrinsic power-bandwidth tradeoff governing the bottleneck link capacity to be shared between multiple downlink paths \cite{Kazemi5}.

The power control coefficients rise continuously with increase in $N_{\rm{T}}$, as observed from Fig.~\ref{fig_4_3_4}. However, they are not allowed to be increased unboundedly due to practical limitations imposed by the maximum permissible optical power of backhaul \acp{LED}. To set an upper limit for the transmission power of the backhaul system, its counterpart from the access system, $P_{\rm{a}}$, is used, as the access system operates with full power to comply with the illumination requirement\footnote{The maximum allowable backhaul power could be an independent variable to model the practical specification of backhaul \acp{LED}. Despite this possibility, setting a value equal to the power used in the access system simplifies the presentation of results, though it does not influence the generality of the power control analysis.}. This exerts a unit threshold constraint on $K_{\rm{b,min}}$, resulting in:
\begin{equation}\label{eq_4_3_5}
K_{\rm{b,min}}^* = \min[K_{\rm{b,min}},1].
\end{equation}

\begin{figure}[!t]
	\centering
	\subfloat[\label{fig_4_3_9a} $\lambda=1$ UE/Cell]{\includegraphics[width=0.5\linewidth]{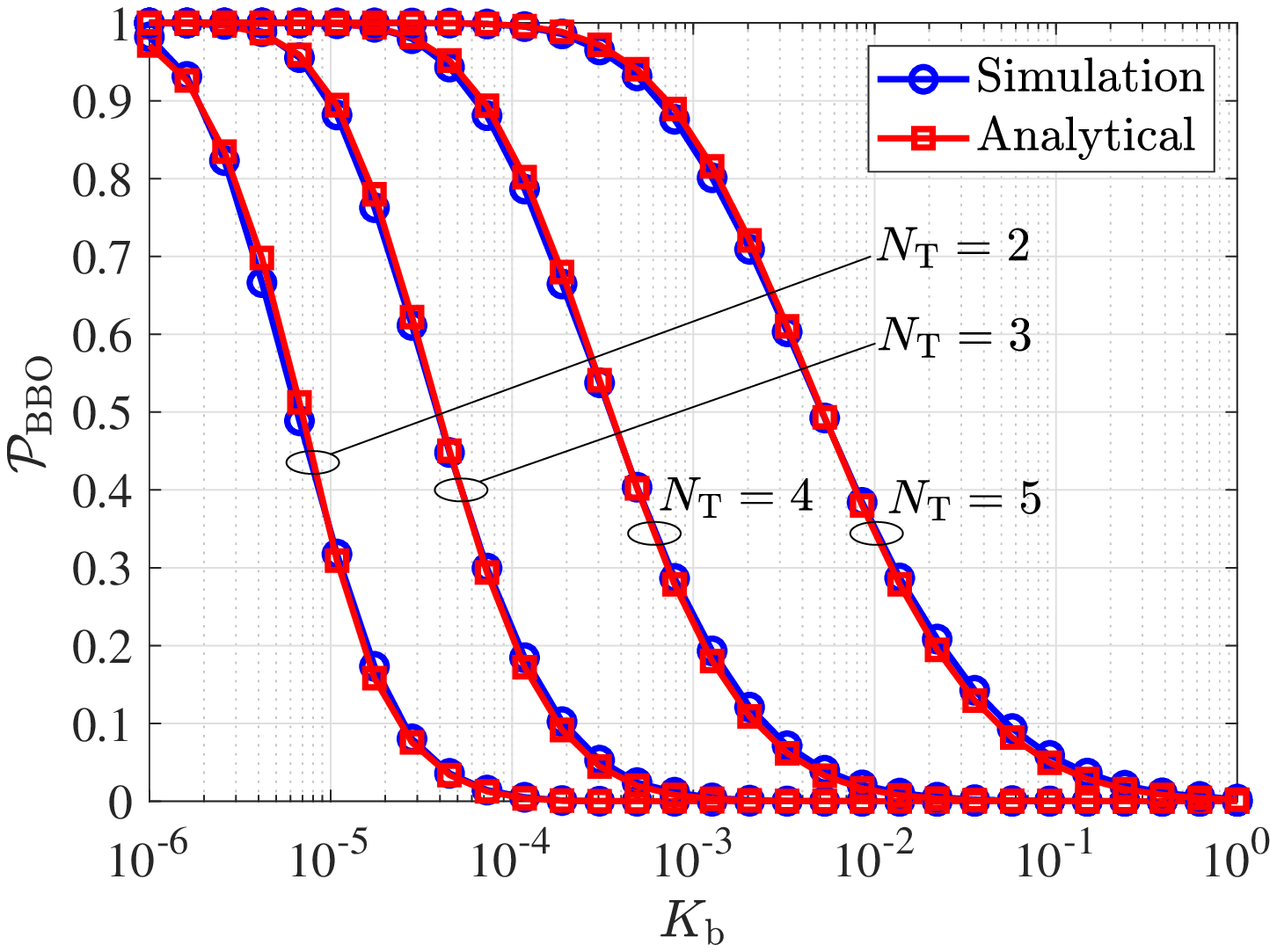}}
	\subfloat[\label{fig_4_3_9b} $\lambda=5$ UE/Cell]{\includegraphics[width=0.5\linewidth]{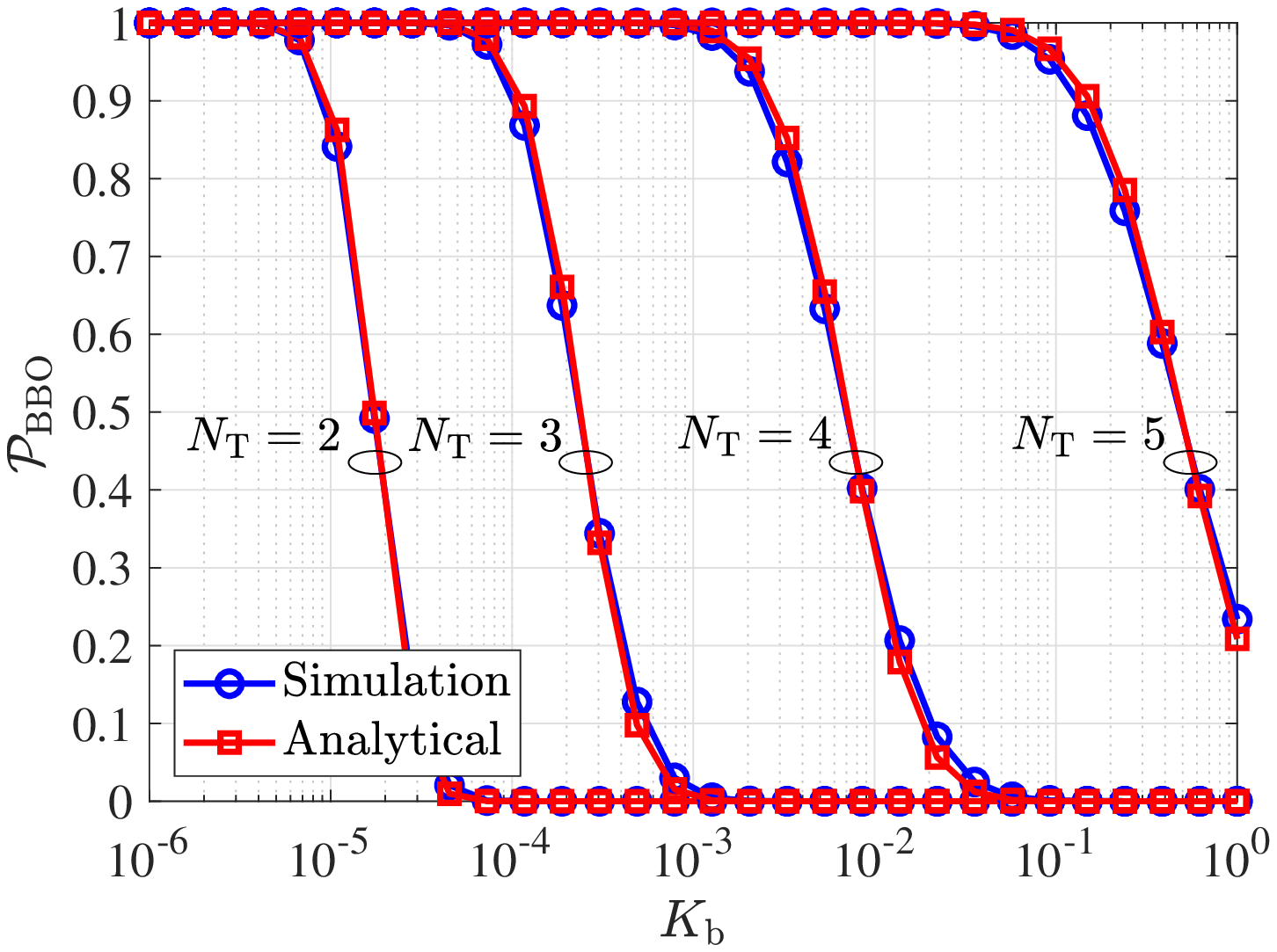}}
	\caption{Analytical and simulation results of the \ac{BBO} probability as a function of $K_{\rm{b}}$ for different values of $N_{\rm{T}}$ and $\lambda$; and $B_{\rm{b}}=3B_{\rm{a}}$. Analytical results are based on \eqref{eq_Theo_4_1_1}.}
	\label{fig_4_3_9}
\end{figure}

\subsubsection{BBO Probability}
For each branch of the super cell, the \ac{BBO} probability can be analytically predicted by way of its approximate expression provided in Theorem~\ref{Theorem_4_1}. To verify the derivation of \eqref{eq_Theo_4_1_1}, the analytical and simulation results are plotted in Fig.~\ref{fig_4_3_9} over a wide range of values of the power ratio $K_{\rm{b}}$. Note that $P_{\rm{BBO}}$ is a function of $K_{\rm{b}}$ through $\mathcal{R}_{\rm{b}_1}$. The simulation results are directly obtained by computing the \ac{BBO} probability in the Monte Carlo domain according to Definition~\ref{Definition_4_1}. For comparison, different combinations of the total number of tiers, $N_{\rm{T}}$, and the average \ac{UE} density, $\lambda$, are considered. 

For both cases of $\lambda=1$ UE/Cell and $\lambda=5$ UE/Cell, as shown in Figs.~\ref{fig_4_3_9a} and \ref{fig_4_3_9b}, respectively, the analytical results closely match with those of the simulations. Nonetheless, there is a slight discrepancy between the two sets of results, because of the underlying approximation. Note that the analytical expression is neither an upper bound nor a lower bound of the \ac{BBO} probability, as it is derived on the basis of the \ac{MMSE} criterion. These results confirm that the formula derived in \eqref{eq_Theo_4_1_1}, though its simple form, does estimate well the actual \ac{BBO} performance of super cells.

\begin{figure}[!t]
	\centering
	\subfloat[\label{fig_4_3_8a} NPC]{\includegraphics[width=0.5\linewidth]{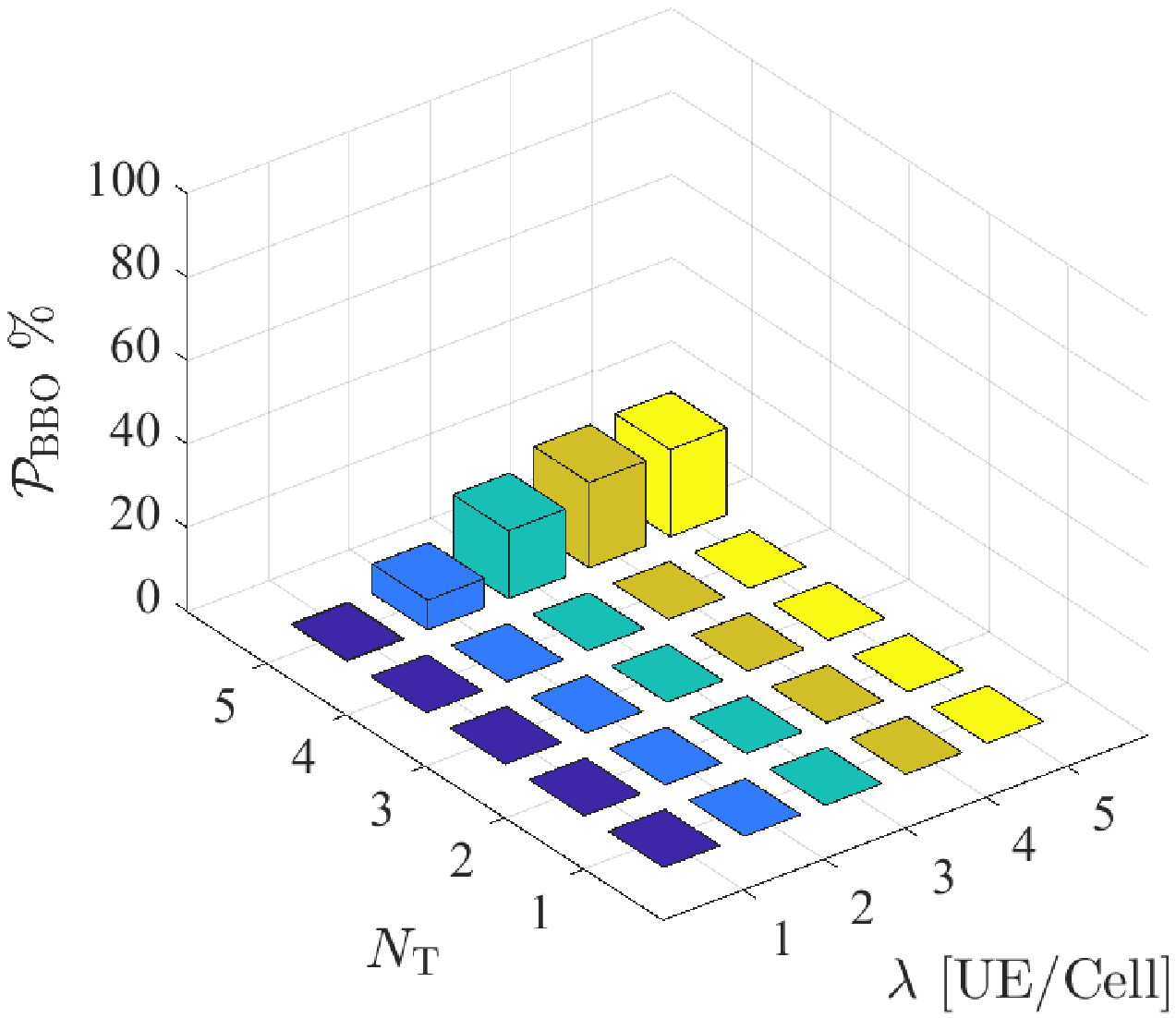}}
	\subfloat[\label{fig_4_3_8b} MSPC]{\includegraphics[width=0.5\linewidth]{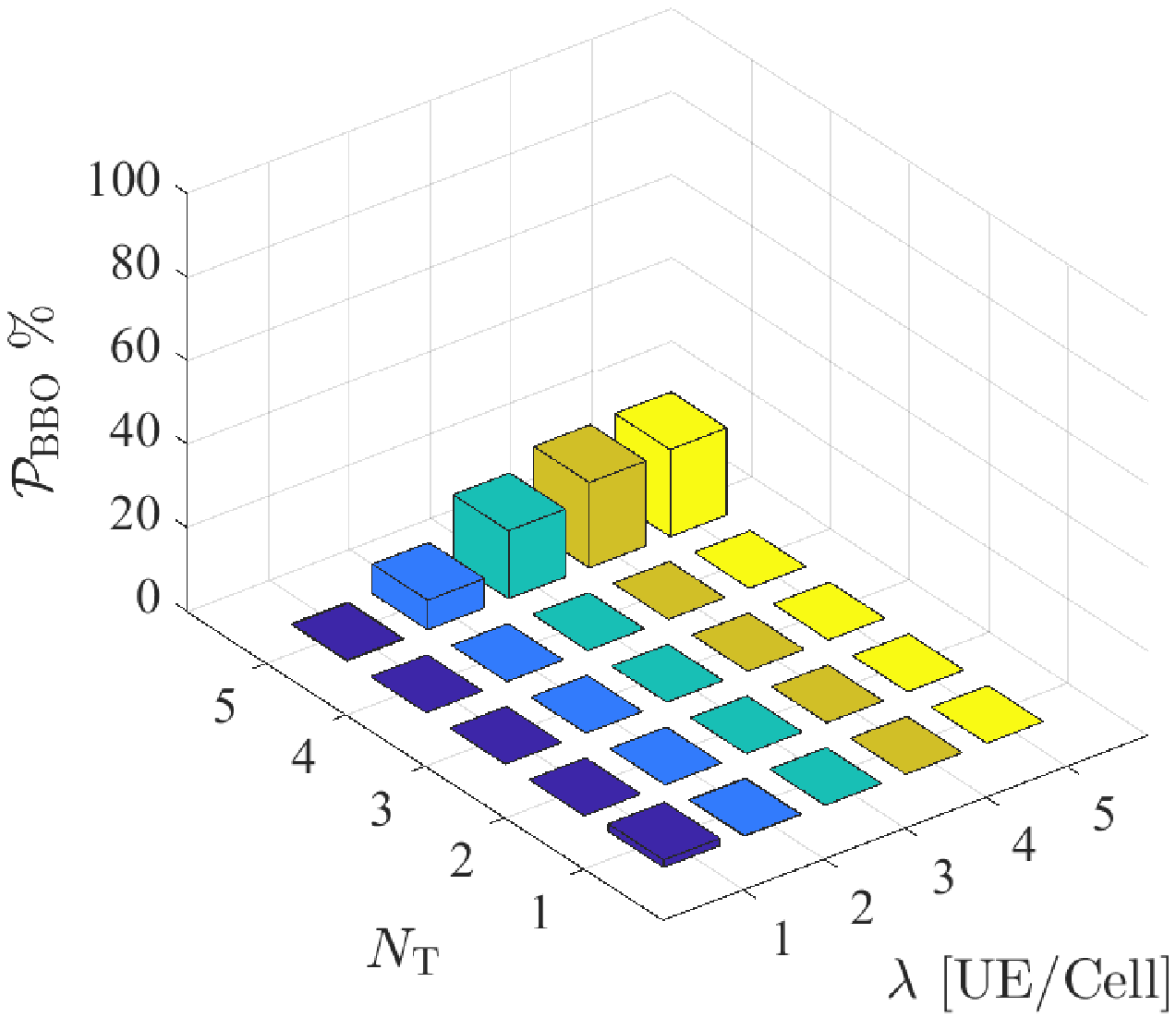}} \\
	\subfloat[\label{fig_4_3_8c} ASPC]{\includegraphics[width=0.5\linewidth]{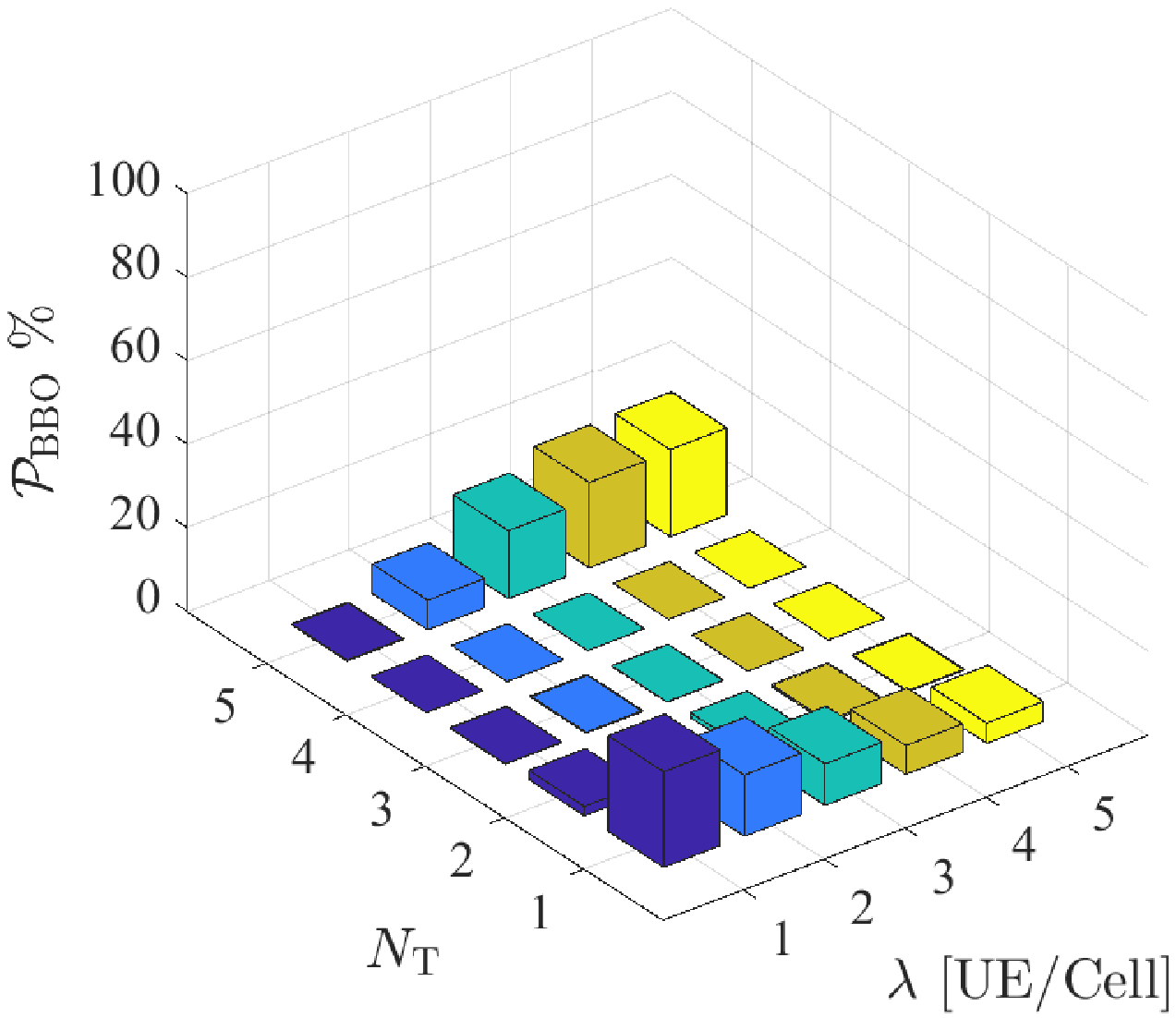}}
	\subfloat[\label{fig_4_3_8d} ARPC]{\includegraphics[width=0.5\linewidth]{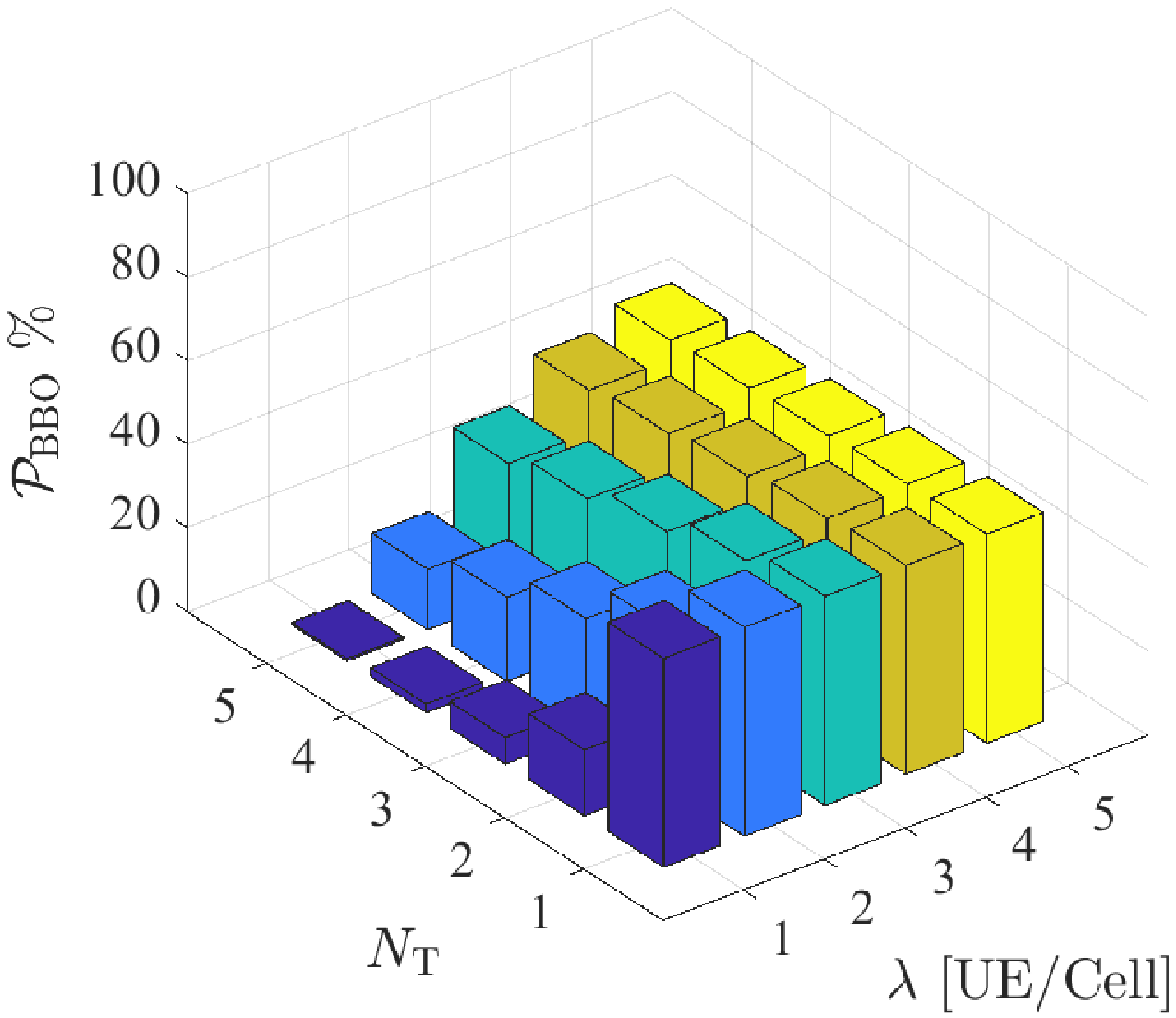}}
	\caption{The BBO probability for NPC, MSPC, ASPC and ARPC schemes versus the total number of tiers $N_{\rm{T}}$ and the UE density $\lambda$ for $B_{\rm{b}}=3B_{\rm{a}}$.}
	\label{fig_4_3_8}
\end{figure}

To shed light on another aspect of the backhaul power control, the resulting \ac{BBO} probability of \ac{MSPC}, \ac{ASPC} and \ac{ARPC} schemes are shown with a percent scale in Fig.~\ref{fig_4_3_8} as a function of $N_{\rm{T}}$ and $\lambda$, for a fixed bandwidth of $B_{\rm{b}}=3B_{\rm{a}}$. These results are obtained by using \eqref{eq_Theo_4_1_1}. The performance of a system with \ac{NPC} in which $P_{{\rm{b}}_i}=P_{\rm{a}}$ $\forall i$ is included for comparison. The results are consistent with those in Figs.~\ref{fig_4_3_4}, \ref{fig_4_3_5} and ~\ref{fig_4_3_6} in the sense that allocating higher power to the backhaul system leads to overall lower values of the \ac{BBO} probability. It is observed that \ac{MSPC} achieves almost equal \ac{BBO} performance as the baseline \ac{NPC} scheme. This is expected from the way \ac{MSPC} is devised by using a high power value just enough to ensure that no backhaul bottleneck takes place, subject to the allowable limit. That is why for both \ac{NPC} and \ac{MSPC}, the \ac{BBO} probability is zero for all cases of $\lambda$ and $N_{\rm{T}}<5$. For $N_{\rm{T}}=5$, however, there is a nonzero chance that the required power to satisfy the access sum rate exceeds the allowed power threshold and therefore backhaul bottleneck inevitably occurs. In this case, the \ac{BBO} probability is increased by adding more \acp{UE}, reaching $20\%$ for $\lambda=5$ UE/Cell.

Besides, \ac{ASPC} performs similar to \ac{NPC} and \ac{MSPC}, except for $N_{\rm{T}}=1$. This can be explained by noting that a one tier super cell involves one attocell per branch, thus any value of $\lambda\geq1$ UE/Cell causes the only attocell of the branch to always be occupied. Unlike \ac{MSPC}, the required power to avoid a backhaul bottleneck in response to such a load may be larger than what \ac{ASPC} computes. The mentioned effect diminishes by increasing the \ac{UE} density as shown in Fig.~\ref{fig_4_3_8c}. When the number of \acp{UE} grows in a single attocell, the range of variations of the access sum rate reduces, thereby lowering the chance for the downlink system to undergo a backhaul bottleneck. Fig.~\ref{fig_4_3_8d} shows that the performance of \ac{ARPC} is worse than all other schemes. The use of \ac{ARPC} leads to $50\%$ \ac{BBO} probability for $\lambda=5$ UE/Cell even for a single tier super cell. For a given $N_{\rm{T}}$, \ac{BBO} is more likely when $\lambda$ increases especially for $N_{\rm{T}}>1$. By contrast, for a fixed value of $\lambda$, \ac{BBO} is less probable when more tiers are added to the super cell. The reason for this trend is because \acp{UE} are associated with the entire branch as a whole and hence they are distributed over a larger number of attocells. This increases the probability that some attocells remain empty, which decreases the aggregate sum rate of the access system. Such a trend decays when the average \ac{UE} density is sufficiently high, i.e. for $\lambda=5$ UE/Cell.

\begin{figure}[!t]
	\centering
	\subfloat[\label{fig_4_3_11a} Average end-to-end sum rate.]{\includegraphics[width=0.5\linewidth]{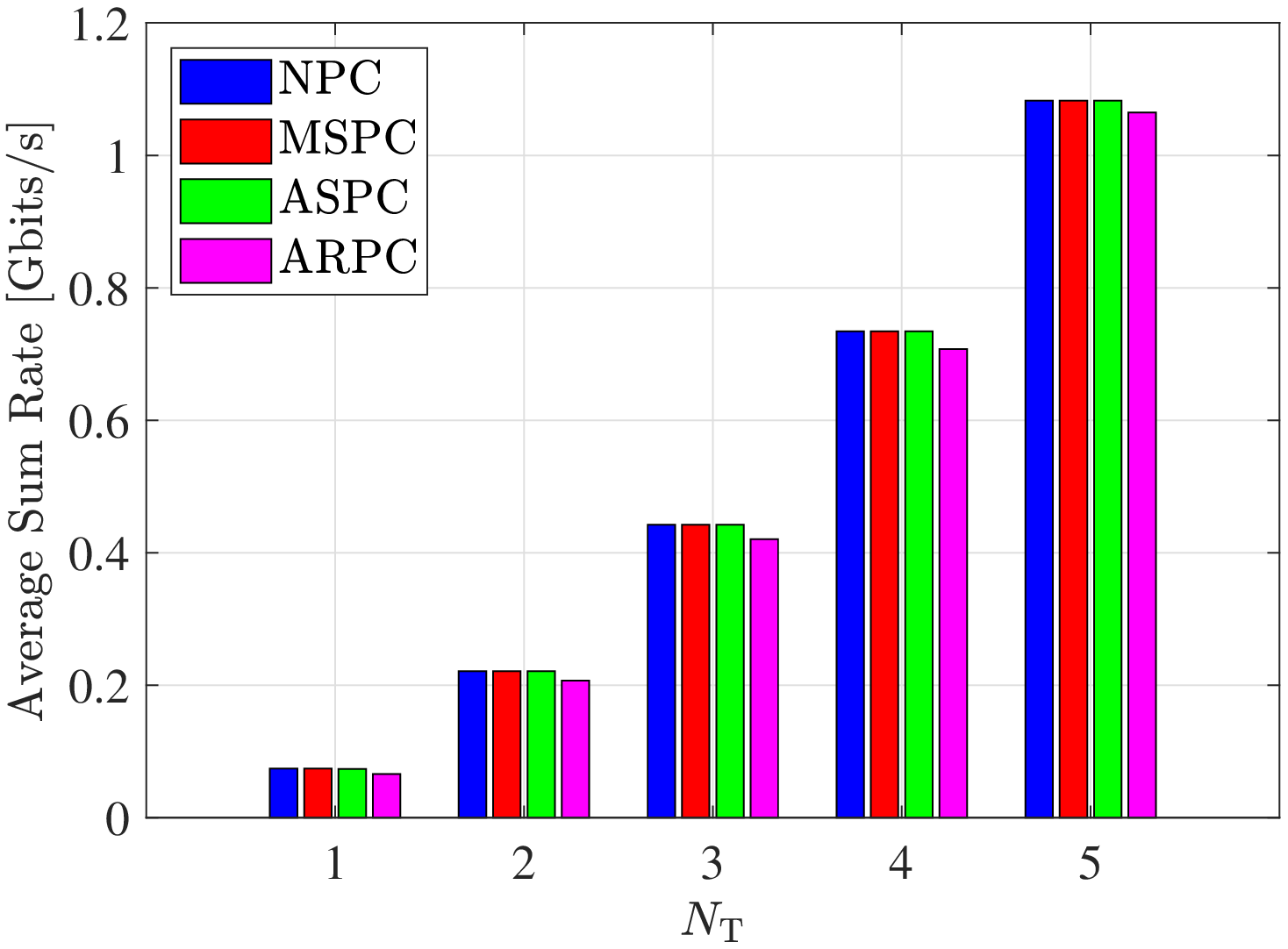}}
	\subfloat[\label{fig_4_3_11b} Power control coefficient.]{\includegraphics[width=0.5\linewidth]{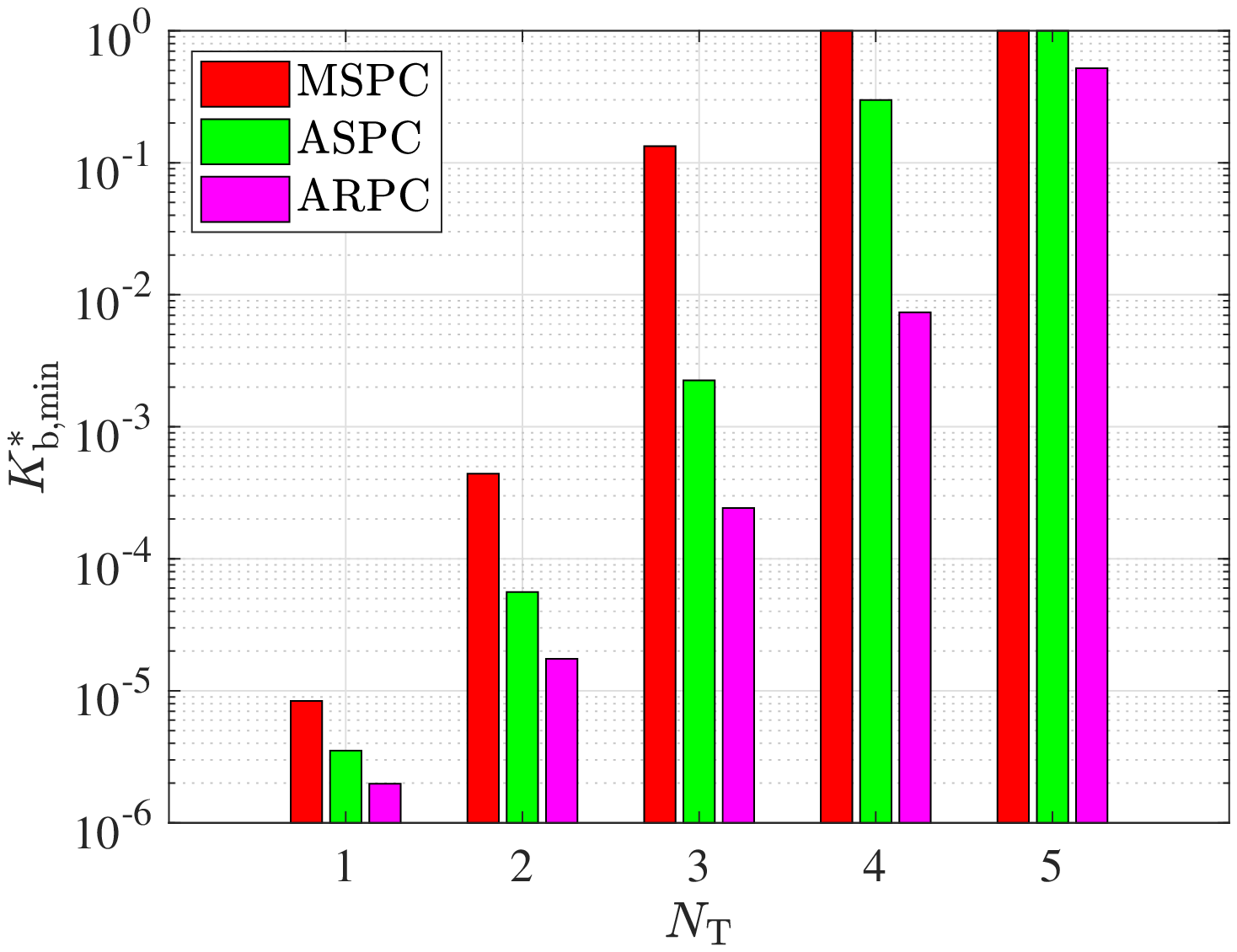}}
	\caption{The average sum rate performance for NPC, MSPC, ASPC and ARPC schemes versus the total number of tiers $N_{\rm{T}}$ for $\lambda=5$ UE/Cell and $B_{\rm{b}}=3B_{\rm{a}}$. The corresponding power control coefficients are shown for comparison.}
	\label{fig_4_3_11}
\end{figure}

\subsubsection{Average Sum Rate Performance}
To measure the end-to-end sum rate performance with power control, the bandwidth allocation ratios for an $N_{\rm{T}}$-tier super cell are computed by applying optimal \ac{CBS} based on Algorithm~\ref{Algorithm2}, per random realization of \acp{UE}.

Fig.~\ref{fig_4_3_11a} demonstrates the average sum rate performance for \ac{NPC}, \ac{MSPC}, \ac{ASPC} and \ac{ARPC} schemes versus $N_{\rm{T}}$ for $\lambda=5$ UE/Cell and $B_{\rm{b}}=3B_{\rm{a}}$. The performance of \ac{NPC} is also shown as a benchmark. It can be observed that \ac{MSPC} and \ac{ASPC} schemes provide the same performance as \ac{NPC} for all values of $N_{\rm{T}}$. They achieve $74$, $221$, $442$, $734$ and $1083$ Mbits/s, for $N_{\rm{T}}=1,2,3,4,5$, respectively. Still, the average sum rate for \ac{ARPC} is slightly lower than the rest of the schemes. The relative performance losses for \ac{ARPC} are around $10\%$, $6\%$, $5\%$, $4\%$ and $2\%$ for $N_{\rm{T}}=1,2,3,4,5$. Note that although the use of \ac{ARPC} leads to high \ac{BBO} probabilities as shown in Fig.~\ref{fig_4_3_8d}, it is of a less impact on the average sum rate performance. This is partly attributed to the optimal \ac{CBS} algorithm which attempts to maximally approach the effective achievable sum rate of the end-to-end system. This could also be anticipated from the function of \ac{ARPC} whereby the backhaul power is tuned to the average sum rate of the access system.

Fig.~\ref{fig_4_3_11b} shows $K_{\rm{b,min}}^*$ associated with each scheme for the same bandwidth of $B_{\rm{b}}=3B_{\rm{a}}$ as used in Fig.~\ref{fig_4_3_11a}. Comparing Fig.~\ref{fig_4_3_11b} with Fig.~\ref{fig_4_3_11a}, it can be observed that remarkable power savings are attained while maintaining the average sum rate performance. For the particular case of $N_{\rm{T}}=3$, by using \ac{MSPC}, the backhaul system operates with only $14\%$ of the full power limit, without affecting the average sum rate. The \ac{PE} can be further improved by employing \ac{ASPC}. Note that both cases of \ac{MSPC} and \ac{ASPC} equally have a zero \ac{BBO} probability according to Fig.~\ref{fig_4_3_9}. For the case of \ac{ARPC}, albeit the improvement in \ac{PE} is achieved at the cost of a slight reduction in the average sum rate performance. From the \ac{PE} perspective, \ac{ASPC} improves upon \ac{MSPC}, and at the same time acquires a \ac{BBO} performance similar to the baseline \ac{NPC} scheme. This suggests that there is an optimum threshold for designing \ac{FPC}-based schemes to strike a tradeoff between the total power minimization and the backhaul bottleneck minimization. The use of \ac{ARPC}, though offering significant power savings, can lead to $50\%$ \ac{BBO} probability regardless of the number of tiers deployed. Such a poor performance disqualifies the impressive \ac{PE} gain that is offered by \ac{ARPC} in terms of the total backhaul power.

\section{Conclusions}\label{section7}
A multi-hop wireless optical backhaul configuration is designed for multi-tier optical attocell networks in a systematic way by means of single-gateway super cells. Resultantly, by expanding the size of super cells, the number of gateways required to supply backhaul connectivity for a network of the same size is progressively reduced, albeit such an advantage comes at a price. The tradeoff between the size and the end-to-end performance is underlined by numerical results, confirming that the number of tiers plays a significant role in determining the network load and, depending on the available bandwidth and power resources, the backhaul rate limit becomes the bottleneck if a large number of tiers is deployed. For efficient use of the backhaul bandwidth, optimal bandwidth scheduling is expounded for both \ac{UBS} and \ac{CBS} policies. Numerical results demonstrate that, under a low \ac{UE} density scenario, both optimal \ac{UBS} and \ac{CBS} algorithms cause the average sum rate performance to almost reach the maximum rate limit as set by access and backhaul systems. They exhibit a superior performance with respect to the baseline equal bandwidth allocation, and the gain is more pronounced when the number of tiers is increased. Under high \ac{UE} density conditions, optimal \ac{CBS} takes the lead relative to optimal \ac{UBS}, and it closely realizes the overall rate limit. Furthermore, a power control framework is established in an attempt to lower the backhaul power using a fixed operating point that does not heavily restrict the network sum rate. The \ac{BBO} probability derived in this paper allows the prediction of the backhaul bottleneck performance. Each of the proposed \ac{FPC} schemes offers a \ac{PE} improvement paired with a certain \ac{BBO} performance. In this respect, \ac{MSPC} achieves a very low \ac{BBO} probability similar to the benchmark \ac{NPC} scheme, while providing considerable power savings especially for fewer number of tiers. By comparison, \ac{ASPC} performs better than \ac{MSPC} in terms of power reduction, and maintains the same \ac{BBO} probability. The use of \ac{ARPC}, though delivering the best \ac{PE} among the candidate schemes, leads to a substantial degradation in \ac{BBO}. From the average sum rate perspective, both \ac{MSPC} and \ac{ASPC} achieve an identical performance compared to \ac{NPC}, and \ac{ARPC} returns a slightly less value because of underestimating the required power.

\appendix
\section*{Proof of Lemma~\ref{Lemma_4_4}}\label{ProofLemma4}
To simplify notation, let $X_u=\mathcal{R}_{\rm{a}}(\gamma_u)$. The expression $\sum_{i\in\mathcal{L}_k}\frac{1}{M_i}\sum_{u\in\mathcal{U}_i}X_u$ is approximated using the \ac{MMSE} criterion. A parameter $\beta$ is introduced to perform the following estimation:
\begin{equation}\label{eq_Proof_Lem_4_4_1}
\underbrace{\sum_{i\in\mathcal{L}_k}\frac{1}{M_i}\sum_{u\in\mathcal{U}_i}X_u}_{Y} \approx \beta\underbrace{\sum_{i\in\mathcal{L}_k}\sum_{u\in\mathcal{U}_i}X_u}_{S}.
\end{equation}
The aim is to determine the optimal estimator $\beta^*$ that minimizes the \ac{MSE} between $Y$ and $\beta S$, where $S = \sum_{u\in\mathcal{U}}X_u$ and $\mathcal{U}$ represents the index set of all the \acp{UE} in the $k$th branch of the network, i.e. $\mathcal{U}=\bigcup_{i\in\mathcal{L}_k}\mathcal{U}_i$. This can be mathematically expressed by:
\begin{mini!}
	{\beta\in\mathbb{R}}{\mathrm{MSE} = \mathbb{E}_X\left[\left(Y-\beta S\right)^2\right]\label{eq_Proof_Lem_4_4_2a}}
	{\label{eq_Proof_Lem_4_4_2}}{}
	\addConstraint{\beta}{>0\label{eq_Proof_Lem_4_4_2b}}
\end{mini!}
The objective \ac{MSE} is expanded as follows:
\begin{equation}\label{eq_Proof_Lem_4_4_3}
\mathrm{MSE} = \mathbb{E}_X\left[Y^2\right]+\beta^2\mathbb{E}\left[S^2\right]-2\beta\mathbb{E}_X\left[YS\right],
\end{equation}
Taking the derivative of the \ac{MSE} with respect to $\beta$ and equating it to zero leads to:
\begin{equation}\label{eq_Proof_Lem_4_4_5}
\frac{d\mathrm{MSE}}{\beta} = 2\beta\mathbb{E}\left[S^2\right]-2\mathbb{E}_X\left[YS\right]=0~\Rightarrow~\beta^* = \frac{\mathbb{E}_X\left[YS\right]}{\mathbb{E}\left[S^2\right]}.
\end{equation}
The expectation $\mathbb{E}_X\left[YS\right]$ in \eqref{eq_Proof_Lem_4_4_5} is expanded as follows:
\begin{subequations}\label{eq_Proof_Lem_4_4_6}
	\begin{align}
	\mathbb{E}_X\left[YS\right] &= \mathbb{E}_X\left[\left(\sum_{i\in\mathcal{L}_k}\frac{1}{M_i}\sum_{u\in\mathcal{U}_i}X_u\right)\left(\sum_{v\in\mathcal{U}}X_v\right)\right],\label{eq_Proof_Lem_4_4_6a} \\
	&= \sum_{i\in\mathcal{L}_k}\frac{1}{M_i}\sum_{u\in\mathcal{U}_i}\sum_{v\in\mathcal{U}}\mathbb{E}\left[X_uX_v\right],\label{eq_Proof_Lem_4_4_6b}
	\end{align}
\end{subequations}
where:
\begin{equation}\label{eq_Proof_Lem_4_4_7}
\mathbb{E}\left[X_uX_v\right] = 
\begin{cases}
\mathbb{E}^2\left[X_u\right] = \bar{\mathcal{R}}_{\rm{a}}^2, & u\neq v \\
\mathbb{E}\left[X_u^2\right] = \sigma_{\mathcal{R}_{\rm{a}}}^2+\bar{\mathcal{R}}_{\rm{a}}^2, & u=v 
\end{cases}
\end{equation}
in which $\bar{\mathcal{R}}_{\rm{a}}$ and $\sigma_{\mathcal{R}_{\rm{a}}}^2$ are given by \eqref{eq_Lem_4_1_1} and \eqref{eq_Lem_4_3_1}, respectively. Therefore:
\begin{subequations}\label{eq_Proof_Lem_4_4_8}
	\begin{align}
	\sum_{u\in\mathcal{U}_i}\sum_{v\in\mathcal{U}}\mathbb{E}\left[X_uX_v\right] &= \sum_{\substack{u\in\mathcal{U}_i,v\in\mathcal{U}\\u\neq v}}\mathbb{E}\left[X_uX_v\right]+\sum_{\substack{u\in\mathcal{U}_i,v\in\mathcal{U}\\u=v}}\mathbb{E}\left[X_uX_v\right],\label{eq_Proof_Lem_4_4_8a} \\
	&= M_i(M-1)\bar{\mathcal{R}}_{\rm{a}}^2+M_i\left(\sigma_{\mathcal{R}_{\rm{a}}}^2+\bar{\mathcal{R}}_{\rm{a}}^2\right),\label{eq_Proof_Lem_4_4_8b} \\
	&= M_i\left(M\bar{\mathcal{R}}_{\rm{a}}^2+\sigma_{\mathcal{R}_{\rm{a}}}^2\right).\label{eq_Proof_Lem_4_4_8c}
	\end{align}
\end{subequations}
By substituting \eqref{eq_Proof_Lem_4_4_8c} into \eqref{eq_Proof_Lem_4_4_6b}, $\mathbb{E}_X\left[YS\right]$ is derived as follows:
\begin{equation}\label{eq_Proof_Lem_4_4_9}
\mathbb{E}_X\left[YS\right] = n_{\rm{BS}}\left(M\bar{\mathcal{R}}_{\rm{a}}^2+\sigma_{\mathcal{R}_{\rm{a}}}^2\right),
\end{equation}
where $n_{\rm{BS}}$ accounts for the number of non-empty attocells. By using \eqref{eq_Proof_Lem_4_4_7}, the expectation $\mathbb{E}\left[S^2\right]$ in \eqref{eq_Proof_Lem_4_4_5} is derived as follows:
\begin{subequations}\label{eq_Proof_Lem_4_4_10}
	\begin{align}
	\mathbb{E}\left[S^2\right] &= \mathbb{E}\left[\left(\sum_{u\in\mathcal{U}}X_u\right)\left(\sum_{v\in\mathcal{U}}X_v\right)\right],\label{eq_Proof_Lem_4_4_10a} \\
	&= \sum_{\substack{u,v\in\mathcal{U}\\u\neq v}}\mathbb{E}\left[X_uX_v\right]+\sum_{\substack{u,v\in\mathcal{U}\\u=v}}\mathbb{E}\left[X_uX_v\right],\label{eq_Proof_Lem_4_4_10b} \\
	&= M\left(M\bar{\mathcal{R}}_{\rm{a}}^2+\sigma_{\mathcal{R}_{\rm{a}}}^2\right).\label{eq_Proof_Lem_4_4_10c}
	\end{align}
\end{subequations}
Finally, by substituting \eqref{eq_Proof_Lem_4_4_9} and \eqref{eq_Proof_Lem_4_4_10c} in \eqref{eq_Proof_Lem_4_4_5}, the optimal estimator reduces to:
\begin{equation}\label{eq_Proof_Lem_4_4_11}
\beta^* = \frac{n_{\rm{BS}}}{M}.
\end{equation}

\bibliographystyle{IEEEtran}

\bibliography{IEEEabrv,References_Full_List}

\begin{IEEEbiography}
	[{\includegraphics[width=1in,height=1.25in,clip,keepaspectratio]{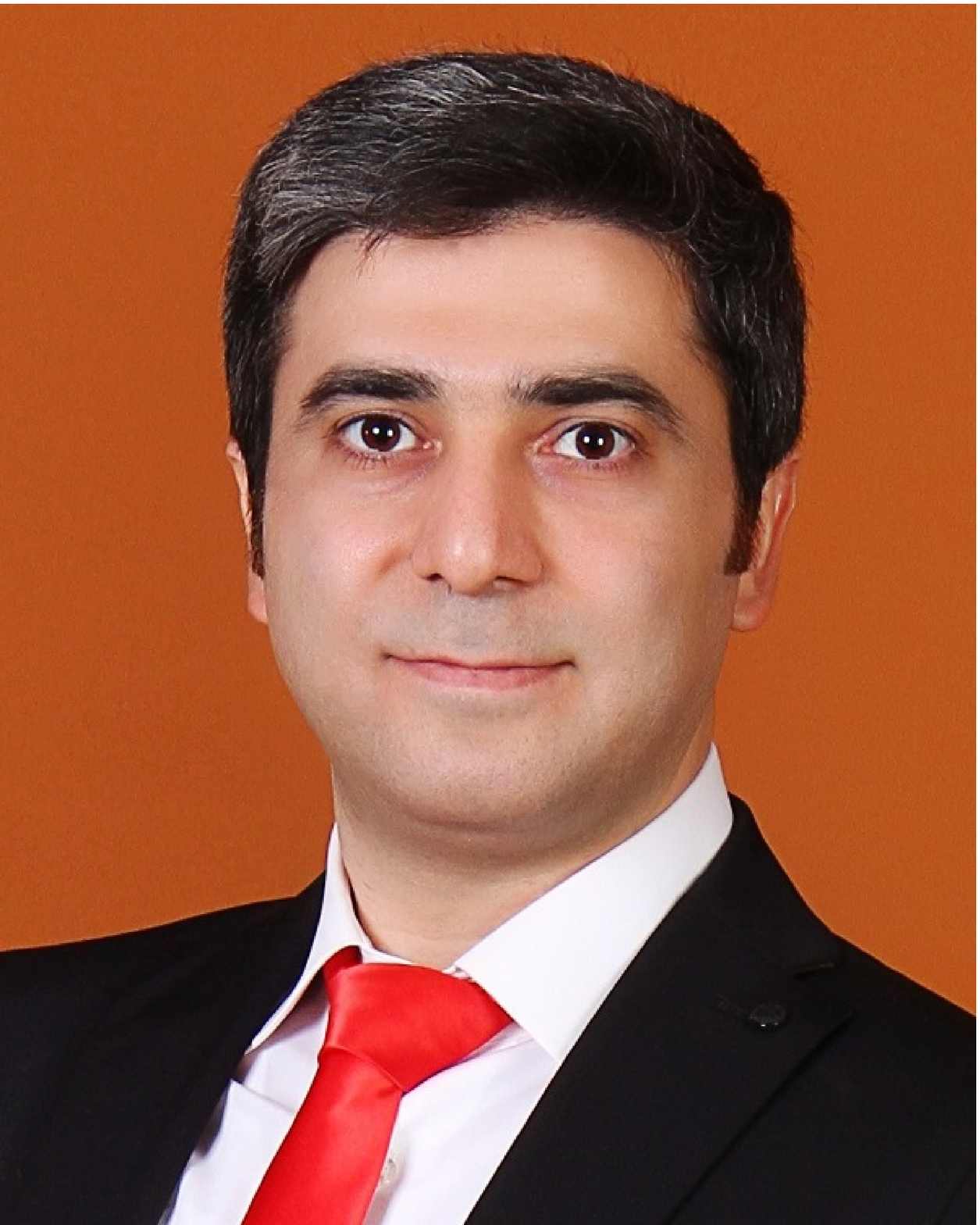}}]{Hossein Kazemi}
	(S'16) received the M.Sc. degree (with a specialty in microelectronic circuits) from Sharif University of Technology, Tehran, Iran, in 2011, the M.Sc. degree (with a focus in communication systems) with honors from \"{O}zye\v{g}in University, Istanbul, Turkey, in 2014, and the Ph.D. degree from the University of Edinburgh, Edinburgh, U.K., in 2019, all in Electrical Engineering. He is currently a postdoctoral research associate at the University of Edinburgh. His research interests mainly include design, analysis and optimization of wireless communication systems and networks.
\end{IEEEbiography}

\begin{IEEEbiography}
	[{\includegraphics[width=1in,height=1.25in,clip,keepaspectratio]{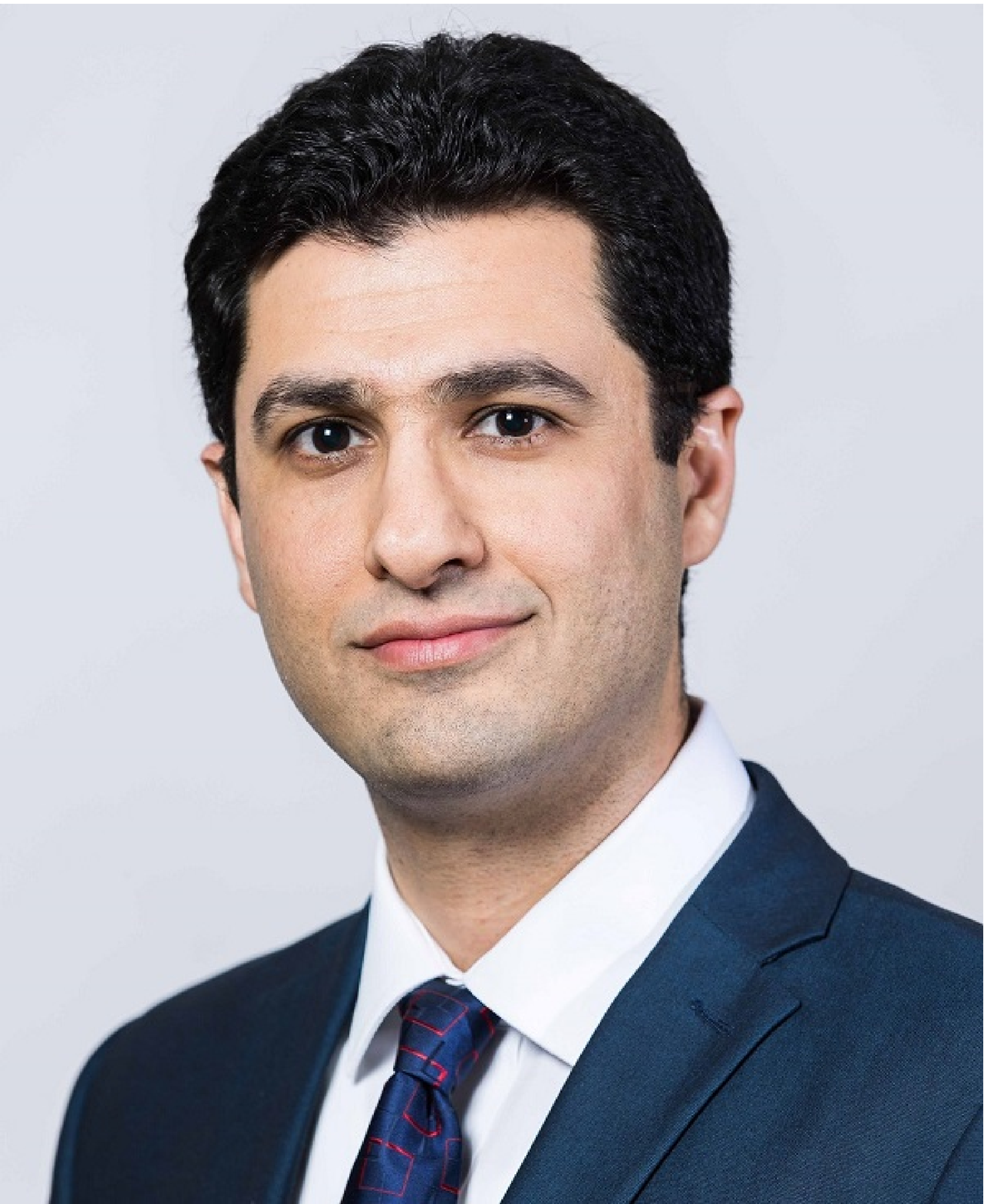}}]{Majid Safari}(S’08-M’11) received his Ph.D. degree in Electrical and Computer Engineering from the University of Waterloo, Canada in 2011. He also received his B.Sc. degree in Electrical and Computer Engineering from the University of Tehran, Iran, in 2003, M.Sc. degree in Electrical Engineering from Sharif University of Technology, Iran, in 2005. He is currently an assistant professor in the Institute for Digital Communications at the University of Edinburgh. Before joining Edinburgh in 2013, he held postdoctoral fellowship at McMaster University, Canada. Dr. Safari is currently an associate editor of IEEE Communication letters. His main research interest is the application of information theory and signal processing in optical communications including fiber-optic communication, free-space optical communication, visible light communication, and quantum communication.
\end{IEEEbiography}

\begin{IEEEbiography}
	[{\includegraphics[width=1in,height=1.25in,clip,keepaspectratio]{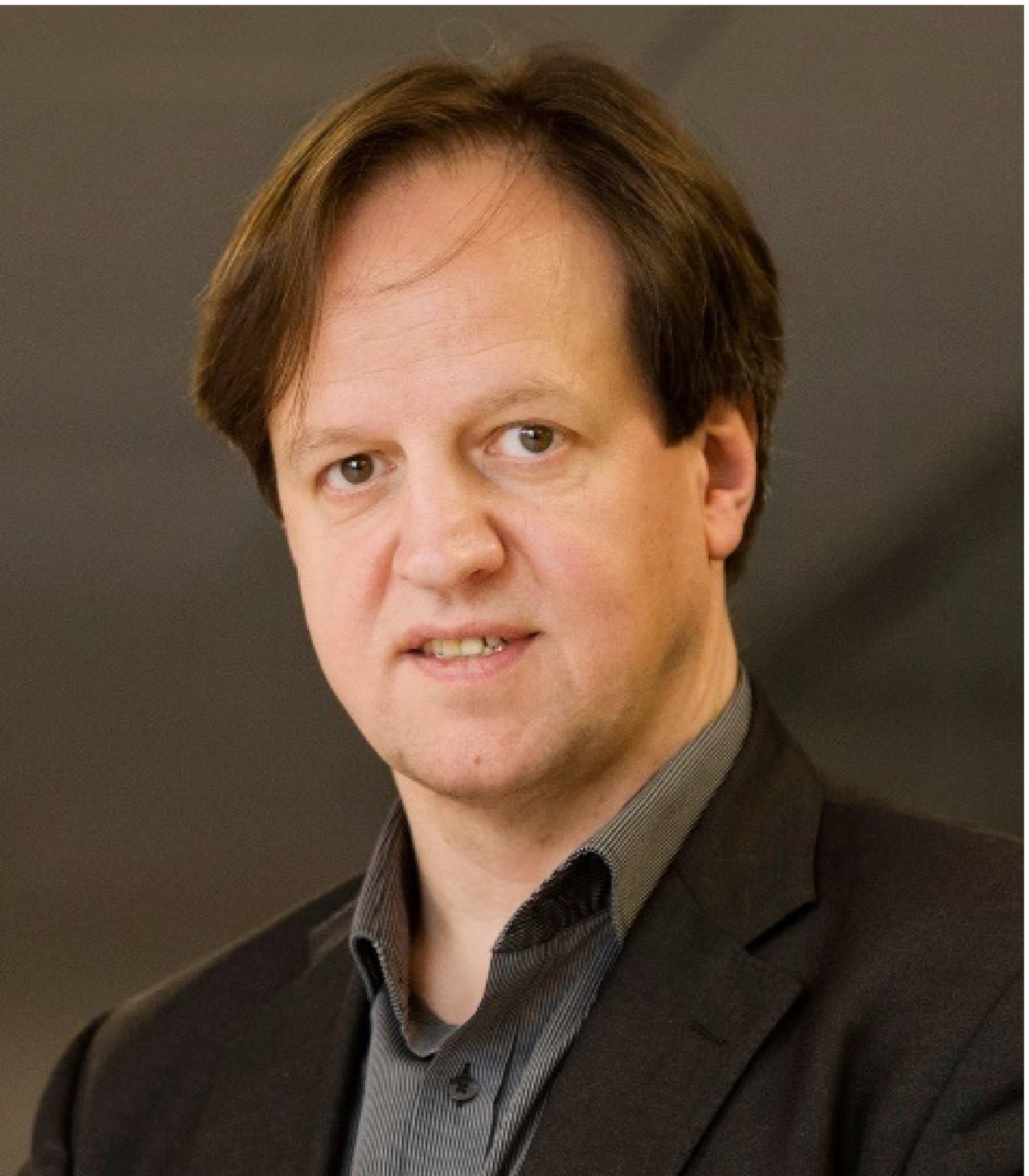}}]{Harald Haas}(S'98-AM'00-M'03-SM'16-F'17) received the Ph.D. degree from the University of Edinburgh in 2001. He currently holds the Chair of Mobile Communications at the University of Edinburgh, and is the Initiator, Co-Founder, and the Chief Scientific Officer of pureLiFi Ltd., and the Director of the LiFi Research and Development Center, the University of Edinburgh. He has authored 400 conference and journal papers, including a paper in Science and co-authored the book \textit{Principles of LED Light Communications Towards Networked Li-Fi} (Cambridge University Press, 2015). His main research interests are in optical wireless communications, hybrid optical wireless and RF communications, spatial modulation, and interference coordination in wireless networks. He first introduced and coined spatial modulation and LiFi. LiFi was listed among the 50 best inventions in \textit{TIME} Magazine 2011. He was an invited speaker at TED Global 2011, and his talk on ``Wireless Data from Every Light Bulb'' has been watched online over 2.4 million times. He gave a second TED Global lecture in 2015 on the use of solar cells as LiFi data detectors and energy harvesters. This has been viewed online over 1.8 million times. He was elected as a fellow of the Royal Society of Edinburgh in 2017. In 2012 and 2017, he was a recipient of the prestigious Established Career Fellowship from the Engineering and Physical Sciences Research Council (EPSRC) within Information and Communications Technology in the U.K. In 2014, he was selected by EPSRC as one of ten Recognising Inspirational Scientists and Engineers (RISE) Leaders in the U.K. He was a co-recipient of the EURASIP Best Paper Award for the \textit{Journal on Wireless Communications and Networking} in 2015, and co-recipient of the Jack Neubauer Memorial Award of the IEEE Vehicular Technology Society. In 2016, he received the Outstanding Achievement Award from the International Solid State Lighting Alliance. He was a co-recipient of recent best paper awards at VTC-Fall, 2013, VTC-Spring 2015, ICC 2016, and ICC 2017. He is an Editor of the IEEE \textsc{Transactions on Communications} and the IEEE \textsc{Journal of Lightwave Technologies}.
\end{IEEEbiography}

\end{document}